\documentclass[conference]{IEEEtran}
\IEEEoverridecommandlockouts
\usepackage[top=2cm, bottom=2cm, left=2cm, right=2cm]{geometry}
\usepackage{multirow,url,diagbox,amsmath,amsthm,amssymb,
graphicx,color,cite,algorithm,algpseudocode,algorithmicx,amsfonts,comment,framed}
\usepackage[table,xcdraw]{xcolor}

\newtheorem{theorem}{Theorem}
\newtheorem{definition}{Definition}
\floatname{algorithm}{Algorithm}

\usepackage{url,caption}
\usepackage{amsmath,amssymb,amsfonts}
\usepackage{tikz}
\usepackage{filecontents}
\usepackage{graphicx}
\usepackage{textcomp}
\usepackage{pifont,framed}
\usepackage{algorithm,algorithmicx}
\usepackage{xcolor}
\usepackage{color}
\usepackage{verbatim}
\usepackage[subfigure]{tocloft}
\usepackage{subfigure}
\usepackage{multirow}
  \usepackage{mathrsfs}
  \usepackage{ upgreek }
\allowdisplaybreaks[4]
\def\BibTeX{{\rm B\kern-.05em{\sc i\kern-.025em b}\kern-.08em
    T\kern-.1667em\lower.7ex\hbox{E}\kern-.125emX}}
\algtext*{EndWhile}
\algtext*{EndIf}

\usepackage{fancyhdr}

\begin{document}
\title{Differential Aggregation against General Colluding Attackers}
\author{
\IEEEauthorblockN{Rong Du\IEEEauthorrefmark{1}, Qingqing Ye\IEEEauthorrefmark{1}, Yue Fu\IEEEauthorrefmark{1}, Haibo Hu\IEEEauthorrefmark{1}, Jin  Li\IEEEauthorrefmark{3}, Chengfang Fang\IEEEauthorrefmark{2}, Jie Shi\IEEEauthorrefmark{2}}

\IEEEauthorblockA{\IEEEauthorrefmark{1}The Hong Kong Polytechnic University,\IEEEauthorrefmark{3}Guangzhou University, \IEEEauthorrefmark{2}Huawei International}
Email: \{roong.du, yuesandy.fu\}@connect.polyu.hk, \{haibo.hu, qqing.ye\}@polyu.edu.hk \\ lijin@gzhu.edu.cn, \{fang.chengfang, shi.jie1\}@huawei.com}

\maketitle

\thispagestyle{fancy}
\fancyhead[C]{\textcolor{red}{This paper has been accepted by IEEE 39th Annual International Conference on Data Engineering (ICDE2023)}}

\begin{abstract}Local Differential Privacy (LDP) is now widely adopted in large-scale systems to collect and analyze sensitive data while preserving users' privacy. However, almost all LDP protocols rely on a semi-trust model where users are curious-but-honest, which rarely holds in real-world scenarios. Recent works ~\cite{wu2021poisoning, cheu2021manipulation, cao2021data} show poor estimation accuracy of many LDP protocols under malicious threat models. Although a few works have proposed some countermeasures to address these attacks, they all require prior knowledge of either the attacking pattern or the poison value distribution, which is impractical as they can be easily evaded by the attackers.

In this paper, we adopt a general opportunistic-and-colluding threat model and propose a multi-group Differential Aggregation Protocol (DAP) to improve the accuracy of mean estimation under LDP. Different from all existing works that detect poison values on individual basis, DAP mitigates the overall impact of poison values on the estimated mean. It relies on a new probing mechanism EMF (i.e., Expectation-Maximization Filter) to estimate features of the attackers. In addition to EMF, DAP also consists of two EMF post-processing procedures (EMF* and CEMF*), and a group-wise mean aggregation scheme to optimize the final estimated mean to achieve the smallest variance. Extensive experimental results on both synthetic and real-world datasets demonstrate the superior performance of DAP over state-of-the-art solutions.
\end{abstract}

\section{Introduction}
\label{introduction}
In the era of big data, privacy-preserving technologies are essential to collect sensitive data and analyse their statistical features (e.g., frequency and mean) while preserving individual's data privacy. Among them, local differential privacy (LDP) is widely deployed for distributed data collection in large-scale systems (e.g., Chrome and iOS). In an LDP protocol, users only provide perturbed data to collectors, who then estimate some statistics, e.g., mean and frequency estimation, the two most fundamental ones ~\cite{acharya2019hadamard,bassily2017practical,ding2017collecting,  wang2017locally,li2020estimating,duchi2018minimax,wang2019collecting,2019PrivKV,nguyen2016collecting}, from these data. A fundamental assumption of LDP, however, is that users are semi-honest, that is, they honestly perturb and send data to the collector according to the protocol. Unfortunately, this assumption rarely holds in real-world scenarios --- any large-scale data collection system cannot rule out the existence of Byzantine users~\cite{awerbuch2004mitigating,rawat2010collaborative,prakash2020mitigating,rawat2010countering}, who are malicious users that collude among themselves to send fake values and influence the estimated statistics in their favor. In particular, the estimated mean has become a popular target in such attacks. For instance, Byzantine users have engaged in product rating fraud for e-commerce sellers to boost their sales~\cite{luca2016fake,street2006promotional,mayzlin2014promotional,hu2011fraud}.
The New York Times reported businesses hiring workers on Mechanical Turk, an Amazon-owned crowdsourcing marketplace, to post fake 5-star Yelp reviews on their businesses~\cite{luca2016fake}.

A recent work shows how poor the performance of an LDP protocol can be under a malicious model~\cite{cheu2021manipulation}. Although a few works have proposed some countermeasures to address these Byzantine attacks in LDP~\cite{cao2021data, wu2021poisoning}, they all require prior knowledge of either the attacking pattern or the poison value distribution, which is impractical as they can be easily evaded by the attackers.

In this paper, we study mean estimation under LDP model against a {\bf general malicious threat model where attackers are opportunistic and colluding}. ``Opportunistic" means such attackers, whose objective is to make the estimated statistical mean deviate as much as possible from the true mean, can manipulate their poison values in their best interests. ``Colluding", also known as Sybil attacks, means the attackers can share their strategy and orchestrate their poison values. This is practical as these attackers can arise from a single Botnet launched by a single attacker. This threat model is more generic than any existing threat model in that it does not limit the attacking strategy, nor does it assume the collector know about the attacking strategy or probabilistic distribution of poison values.

There are three major challenges to address. First, it is difficult to distinguish Byzantine users from normal users, because both poison values and perturbed values appear as noisy data to the collector. Second, in some LDP protocols the perturbed value domain is also enlarged from the original one, which helps Byzantine users to introduce even larger errors. Third, as the LDP perturbation gets stronger (i.e., $\epsilon$ becomes smaller), the value domain can get even larger. The latter two challenges imply that a single Byzantine user can have significant impact on the mean estimation error, also known as the long-tail attack~\cite{hu2017improving,shan2017tail}. In robust statistics~\cite{dixon1974trimming, cerioli2019wild, olive2008applied}, trimming has been applied to address long-tail attacks by removing the tail values, i.e., those before and after a certain percentile in the data distribution. While trimming is simple and effective in many application scenarios, it has several severe limitations in the context of LDP:

\begin{itemize}
\setlength{\itemsep}{0pt}
\setlength{\parsep}{0pt}
\setlength{\parskip}{0pt}
\item It can only be applied to known long-tail distributions. In our problem, however, the distribution, including those from Byzantine users, is unknown to the data collector.
\item The trimming threshold is essential to the effectiveness, but it is hard to determine. Further, if leaked to Byzantine users, this threshold can become a single point of vulnerability, as they can circumvent the trimming by manipulating their values close to but within this threshold.
\item The trimming also removes perturbed values at the tail from normal users, causing non-negligible bias on the estimated mean.

\end{itemize}

For the first time in the literature, our core idea is to {\bf avoid distinguishing a poison value from a normal one}, as trimming or most existing works do. Instead, for mean estimation we just need to {\bf estimate and filter out the collective impact of poison values}. This is through an Expectation-Maximum Filter (EMF) mechanism, which statistically estimates three {\bf Byzantine features}, namely, the Byzantine user population, head/tail (i.e., attack direction), and poison value distribution of Byzantine users. Since EMF only performs well for a small privacy budget, we design a protocol where each user perturbs her numerical value twice with $\epsilon_\alpha$ and $\epsilon_\beta$ budgets respectively (where $\epsilon_\alpha+\epsilon_\beta=\epsilon$, the total privacy budget). The data collector first uses EMF to estimate these Byzantine features from all users with a small $\epsilon_\alpha$, which are then used to estimate and filter out the collective impact of poison values on perturbed values with $\epsilon_\beta$. Compared to trimming, this approach does not prune any value on an individual basis, so it does not rely on a hard setting of trimming threshold.

The above protocol has one security flaw. Since the protocol must be publicly known, the attackers can behave like normal users upon $\epsilon_\alpha$ to avoid exposing the true Byzantine features, while sending poison values upon $\epsilon_\beta$. To address this issue, the final protocol, namely Differential Aggregation Protocol, assigns a {\bf single but random} $\epsilon$ on each user. The idea is to divide users into multiple groups, and each group is assigned a different $\epsilon$. Users report perturbed values according to the $\epsilon$ assigned to their group.\footnote{To guarantee all users have the same privacy budget, those assigned with smaller $\epsilon$ perturb and report
multiple times until the overall privacy budget is depleted. Please refer to Section~\ref{DAP} for details.} The collector then estimates both Byzantine features and the mean in each group, and finally aggregates these estimated means to derive the optimal mean with the smallest variance.

The contributions of this paper are summarized as follows.
\begin{itemize}
\setlength{\itemsep}{0pt}
\setlength{\parsep}{0pt}
\setlength{\parskip}{0pt}
\item We formulate a general threat model for poisoning attacks in LDP. As far as we know, this is the first model in LDP literature that does not impose specific attacking pattern or strategy on the attackers.
\item We present an information probing mechanism EMF, which can estimate statistical information about the Byzantine users, including population, attacking direction and poison value distribution.
\item Based on EMF, we design a multi-group differential aggregation protocol for mean estimation. This protocol consists of EMF, two post-processing schemes of EMF (EMF* and CEMF*), and a group-wise mean aggregation scheme to optimize the final estimated mean with the smallest variance.
\end{itemize}

The rest of the paper is organized as follows. Section \ref{preliminaries} provides preliminaries of LDP, mean estimation under LDP and EM algorithm. Section \ref{problemdefinition} presents the problem definition and formally defines the threat model. Section \ref{EMF} presents the baseline protocol that adopts EMF to estimate the features of Byzantine users, and Section \ref{DAP} introduces the Differential Aggregation Protocol as a more secure and effective protocol. Section \ref{Exp} shows our experimental results and Section \ref{Relatedwork} reviews the related work. Finally, Section \ref{conclusion} concludes this paper and discusses some future work. 
\section{Preliminaries}
\label{preliminaries}
\subsection{Local Differential Privacy}
LDP~\cite{chen2016private, duchi2013local, kasiviswanathan2011can}, a variant of DP, is a state-of-the-art privacy protection technology~\cite{dwork2008differential, dwork2006calibrating, mcsherry2007mechanism}. Different from centralized DP, the data collector is untrusted in an LDP protocol, so she should not know the true information of any individual users. As such, instead of sending original data, users perturb their data using a randomized perturbation mechanism, from which the data collector can estimate certain statistics, e.g., mean and frequency. In essence, LDP ensures that, upon receiving an output $y$, the data collector cannot distinguish with high confidence whether the input is $x$ or $x'$, any two input values. A formal definition of LDP is as follows:

\begin{definition}
A local algorithm $R$ satisfies $\epsilon$-LDP if for any two inputs $x$ and $x'$ and for any output $y$,
\begin{equation*}
\setlength{\abovedisplayskip}{2pt}
\setlength{\belowdisplayskip}{2pt}
\label{conditionaleq1}
 e^{-\epsilon}\leq\frac{Pr(R(x)=y)}{Pr(R(x')=y)}\leq e^\epsilon
\end{equation*}
always holds.
\end{definition}

The level of privacy protection is determined by $\epsilon$. The smaller $\epsilon$ is, the stronger the privacy is protected.

\subsection{Piecewise Mechanism}LDP has been widely adopted to estimate statistics from a large population of users. In this paper, we mainly focus on state-of-the-art \textbf{Piecewise Mechanism (PM)~\cite{wang2019collecting}} on mean estimation of numerical values. As shown in Algorithm \ref{PM}, given input value $v$ $\in$ $[-1,1]$, the probability density function (PDF) of output $v'$ $\in$ $[-C,C]$ has two parts: domain $[l(v),r(v)]$ and domain $[-C,l(v))\cup(r(v),C]$, where $C=\frac{e^{\epsilon/2}+1}{e^{\epsilon/2}-1}$, $l(v)=\frac{C+1}{2}v-\frac{C-1}{2}$ and $r(v)=l(v)+C-1$. Given input $v$, the perturbed value is in range $[l(v),r(v)]$ with high probability and in range $[-C,l(v))\cup(r(v),C]$ with low probability. Because value $v'$ is an unbiased estimator of input value $v$, the data collector can use the mean of collected values as an unbiased estimator of the mean of input values.


\begin{algorithm}[]
\caption{Piecewise Mechanism}
 \begin{flushleft}
\nonumber {\bf Input:}
Original value $v$ and privacy budget $\epsilon$\\
\nonumber {\bf Output:}
Perturbed value $v'$
\end{flushleft}
\begin{algorithmic}[1]\small
\State
Sample $x$ uniformly at random from [0,1];
\If { $x<\frac{e^{\epsilon/2}}{e^{\epsilon/2}+1}$}\State
Sample $v'$ uniformly at random from $[l(v), r(v)]$
\Else\State
 Sample $v'$ uniformly at random from $[-C,l(v)) \cup (r(v),C]$
\EndIf
\State \Return $v'$
\end{algorithmic}
\label{PM}
\end{algorithm}

\subsection{Expectation Maximization}
Given a set of observed values $X$ in a statistical model, a straightforward approach to estimate an unknown parameters $\theta$ of it is to find the maximum likelihood estimation (MLE). Generally, we can get $\theta$ by setting all first-order partial derivatives of the likelihood function $l$ to zero and solve them. However, it is impossible to attain $\theta$ in this way where latent variables $Z$ exist --- the result will be a set of interlocking equations where the solution of $\theta$ needs the values of $Z$ and vice versa. When one set of equations is substituted for the other, the result is an unsolvable equation. The expectation-maximization (EM) algorithm~\cite{arthur1977maximum} can effectively find the MLE by performing expectation (E) steps and maximization (M) steps iteratively when there are latent variables.

\textbf{E step} produces a function $Q(\theta|\theta^t)$ that evaluates the log-likelihood expectation of $\theta$ given the current estimated parameters $\theta^t$:
{\setlength{\abovedisplayskip}{2pt}
\setlength{\belowdisplayskip}{2pt}
\begin{equation*}
Q(\theta|\theta^t)=\mathbb{E}_{Z|X,\theta^t}[logl(\theta;X,Z)].
\end{equation*}}
\textbf{M step} calculates parameters that maximize the expected log-likelihood obtained in the E step:
{\setlength{\abovedisplayskip}{2pt}
\setlength{\belowdisplayskip}{0pt}
\begin{equation*}
\theta^{t+1}=\mathop{\arg\max}\limits_{\theta}Q(\theta|\theta^t).
\end{equation*}}

\section{Problem Definition and Framework Overview}
\label{problemdefinition}

An essential assumption of most existing LDP works is that users will report their values honestly, which is impractical in real-world applications. Some recent studies show that LDP protocols are vulnerable to Byzantine attacks~\cite{cheu2021manipulation, cao2021data} and the situation becomes even worse when the perturbation is more substantial, that is, with a smaller privacy budget $\epsilon$.\footnote{$\epsilon$ is usually no more than 5.0 in existing LDP schemes, and no more than 3.0 in these attacks.} In this section, we first define our threat model, and then present a framework for LDP mean estimation under this model.

\subsection{Threat Model}
In this paper, we assume an {\bf unknown} number\footnote{To achieve Byzantine fault tolerance (BFT), the proportion of Byzantine users is bounded by 1/2. Otherwise, there is no guarantee on the convergence of the optimal estimated mean to the true mean.} of {\bf colluding} Byzantine users know the LDP perturbation mechanism and the privacy budget $\epsilon$, so they can send arbitrary values in the perturbation output domain $[D_L, D_R]$ to the data collector to deviate from the true mean. We formalize this attack as the threat model below.

\begin{definition}
\textbf{General Byzantine Attack (GBA).} Given a normalized perturbation value domain $[D_L, D_R]$ and $m$ colluding Byzantine users $U_B$ with original values $V_B=\{v_1, ..., v_m\}$, a general Byzantine attack from $U_B$, denoted by $GBA(U_B)$, reports arbitrary poison values $V'_B=GBA(V_B, D_L, D_R)$ to the collector, where $V'_B\in [D_L, D_R]^m$.
\end{definition}

This $GBA$ model also covers those cases when Byzantine users further perturb poison values with the same LDP protocol as normal users, because the perturbed poison values still fall in $[D_L, D_R]$. Two GBAs can be equivalent in terms of the degree they influence the \textbf{true mean $O$}, which is formalized as below.

\begin{definition}
\textbf{Equivalent GBAs.} Let $GBA(U_{B1})$ and $GBA(U_{B2})$ be two General Byzantine Attacks defined on the same perturbation value domain $[D_L, D_R]$ from two sets of Byzantine users $U_{B1}$ and $U_{B2}$. We say they are equivalent, or alternatively, $GBA(U_{B1})$ can be \textbf{reduced to} $GBA(U_{B2})$, if and only if the following equation on their reporting poison values $V'_{B1}$ and $V'_{B2}$ holds:
{\setlength{\abovedisplayskip}{2pt}
\setlength{\belowdisplayskip}{2pt}
\begin{equation*}
\sum\limits_{v'_{B1}\in V'_{B1}}{(v'_{B1}-O)}=\sum_{v'_{B2}\in V'_{B2}}{(v'_{B2}-O)}.
\end{equation*}}
\end{definition}

Among all GBAs, we are particularly interested in those whose poison values are coordinated to bias towards one side, a.k.a., the {\bf poisoned side}. We call them {\bf Biased Byzantine Attacks}, which are formally defined below.

\begin{definition}\textbf{Biased Byzantine Attack (BBA).}
Given a normalized perturbation value domain $[D_L, D_R]$ and $m$ colluding Byzantine users $U_B$ with original values $V_B=\{v_1, ..., v_m\}$, a biased Byzantine attack from $U_B$, denoted by $BBA(U_B)$, reports poison values $V'_B=BBA(V_B, D_L, D_R)$, where either $V'_B\in[D_L,O]^m$ or $V'_B\in[O, D_R]^m$.

\label{definition of BBA}
\end{definition}

Note that while $GBA$ is our threat model, $BBA$ is not. But the latter is a handy logical model in the estimation of Byzantine user features, which will be elaborated in the next section. 
On the other hand, since mean only involves additive operations, Theorem \ref{theo of CBBA} below shows that in mean estimation, any $GBA$ can be \textbf{reduced to} a $BBA$:

\begin{theorem}For any $GBA(U_{B1})$ reporting poison values $V'_{B1}$, there exists a $BBA(U_{B2})$ reporting poison values $V'_{B2}$, such that
{\setlength{\abovedisplayskip}{1pt}
\setlength{\belowdisplayskip}{1pt}
\begin{equation*}
\sum\limits_{v'_{B1}\in V'_{B1}}{(v'_{B1}-O)}=\sum\limits_{v'_{B2}\in V'_{B2}}{(v'_{B2}-O)}.
\end{equation*}}
 \label{theo of CBBA}
\end{theorem}
\begin{proof} Now we consider a $GBA(U_{B1})$ where poison values exist on both sides. Let $V_{B1}^L$ denote the set of poison values on $[D_L, O]$, and $V_{B1}^R$ denote those on $[O,D_R]$.

Without loss of generality, consider the case where the mean of the left side is larger. Let
\begin{equation}
\label{eq001}
\setlength{\abovedisplayskip}{2pt}
\setlength{\belowdisplayskip}{2pt}
C_0=\Sigma (V_{B1}^L-O)+\Sigma (V_{B1}^R-O)<0.
\end{equation}

Let us choose the largest poison value $\mathbf{y_r}$ in $V_{B1}^R$, and an arbitrary subset of poison values $\mathcal{Y_L}$ from $V_{B1}^L$, such that the following two formulas are satisfied concurrently:
\begin{equation}
\setlength{\abovedisplayskip}{2pt}
\setlength{\belowdisplayskip}{2pt}
\Sigma(\mathcal{Y_L}-O)+\mathbf{y_r}-O \leq 0,
\label{equ11}
\end{equation}
and for any $\mathbf{y_{l}} \in \mathcal{Y_L}$,
\begin{equation}
\setlength{\abovedisplayskip}{2pt}
\setlength{\belowdisplayskip}{2pt}
\Sigma(\mathcal{Y'_L}-O)+\mathbf{y_r}-O>0,
\label{equ22}
\end{equation}
where $\mathcal{Y'_L}$ is the set of elements in $\mathcal{Y_L}$ excluding $\mathbf{y_{l}}$.

This is always available, because
\begin{equation*}
\setlength{\abovedisplayskip}{2pt}
\setlength{\belowdisplayskip}{2pt}
\Sigma (O-V_{B1}^L)> \Sigma (V_{B1}^R-O) \geq \mathbf{y_r}-O.
\end{equation*}
Add $\mathbf{y_{l}}-O$ on both sides of Equ. \ref{equ22}, we then have
\begin{equation}
\begin{split}
\setlength{\abovedisplayskip}{2pt}
\setlength{\belowdisplayskip}{2pt}
 \Sigma(\mathcal{Y'_L}-O)+\mathbf{y_{l}}-O+\mathbf{y_r}-O& \ > \mathbf{y_{l}}-O \\
 \Sigma(\mathcal{Y_L}-O)+\mathbf{y_r}-O&> \mathbf{y_{l}}-O
\end{split}
\label{equ33}
\end{equation}
Let \begin{equation}
\setlength{\abovedisplayskip}{2pt}
\setlength{\belowdisplayskip}{2pt}
\mathbf{y'_l}-O=\Sigma(\mathcal{Y_L}-O)+\mathbf{y_r}-O
\label{equ544}
\end{equation}
According to Equ. \ref{equ11} and Equ. \ref{equ33}, we have
\begin{equation*}
\setlength{\abovedisplayskip}{2pt}
\setlength{\belowdisplayskip}{2pt}
\mathbf{y_{l}}-O < \mathbf{y'_l}-O \leq 0
\end{equation*}
Apparently, $\mathbf{y'_l} \in [D_L, O]$. Let us remove $\mathbf{y_r}$ from $V_{B1}^R$, $\mathcal{Y_L}$ from $Y_L$, and add $\mathbf{y'_l}$ into $V_{B1}^L$. Such an operation eliminates a poison value on the right hand without changing $C_0$ according to Equ. \ref{eq001} and Equ. \ref{equ544}. Hence, the generated distribution can be reduced to the initial one.

Repeating the same operation until $V_{B1}^R$ is empty, we finally obtain a $V_{B2}'=V_{B1}^L$. This is achievable, as both $V_{B1}^L$ and $V_{B1}^R$ are finite, and $C_0<0$. We finally obtain a $BBA(U_{B2})$ reporting $V_{B2}'$ where poison values are on the left hand.
 \label{theo of CBBAtoBBA}
\end{proof}

\color{black}{As a final note, our GBA threat model essentially follows the \textbf{general manipulation attack} of~\cite{cheu2021manipulation}, in which Byzantine users can freely choose to report any poison values in the domain, without following a distribution imposed by the LDP perturbation. A special general manipulation attack is the \textbf{input manipulation attack}, in which Byzantine users perturb (or partially perturb) poison values with the same LDP protocol as normal users. Intuitively, this latter attack is far less effective than the former attack but is harder to detect due to the disguise from the perturbation. While this work focuses on the defense against general manipulation attacks, we will show in the experiment that it can also work with existing detection techniques, such as k-means-based defense~\cite{li2022fine} (designed for input manipulation attacks), outlier detection by boxplot~\cite{schwertman2004simple}, and isolation forests~\cite{liu2008isolation, ding2013anomaly}.}
\color{black}

\subsection{System Model and Framework}Fig. \ref{systemframework} shows the system model and our aggregation framework. There are $N$ users, among which $n$ are normal users (green) and $m$ are Byzantine users (red). Normal users perturb and normalize their values $v_i$ into $v_i'\in [D_L,D_R]$ according to the LDP perturbation mechanism, and send them to the data collector. Byzantine users choose and send poison values to the data collector (step \ding{172}).
 The goal of the data collector is to estimate the mean of {\bf normal users}. In contrast to existing detection-based methods~\cite{liu2007insider, li2010catching, rudrapal2013internal}, in our framework the data collector first has an initial guess on the true mean $O$, and then probes collected values (step \ding{173}) to estimate three features of Byzantine users, namely, the proportion of Byzantine users, the poisoned side, and the frequency histogram of poison values. Based on them, the data collector estimates the aggregated mean (step \ding{174}). 

In the rest of this paper, we will illustrate how this framework can be used for mean aggregation under LDP privacy model where Byzantine users exist. In Section~\ref{EMF}, we present a baseline protocol, followed by a security-enhanced protocol in Section~\ref{DAP}.

\begin{figure}[]
\centering
\setlength{\abovecaptionskip}{0.cm}
\includegraphics[width=0.45\textwidth]{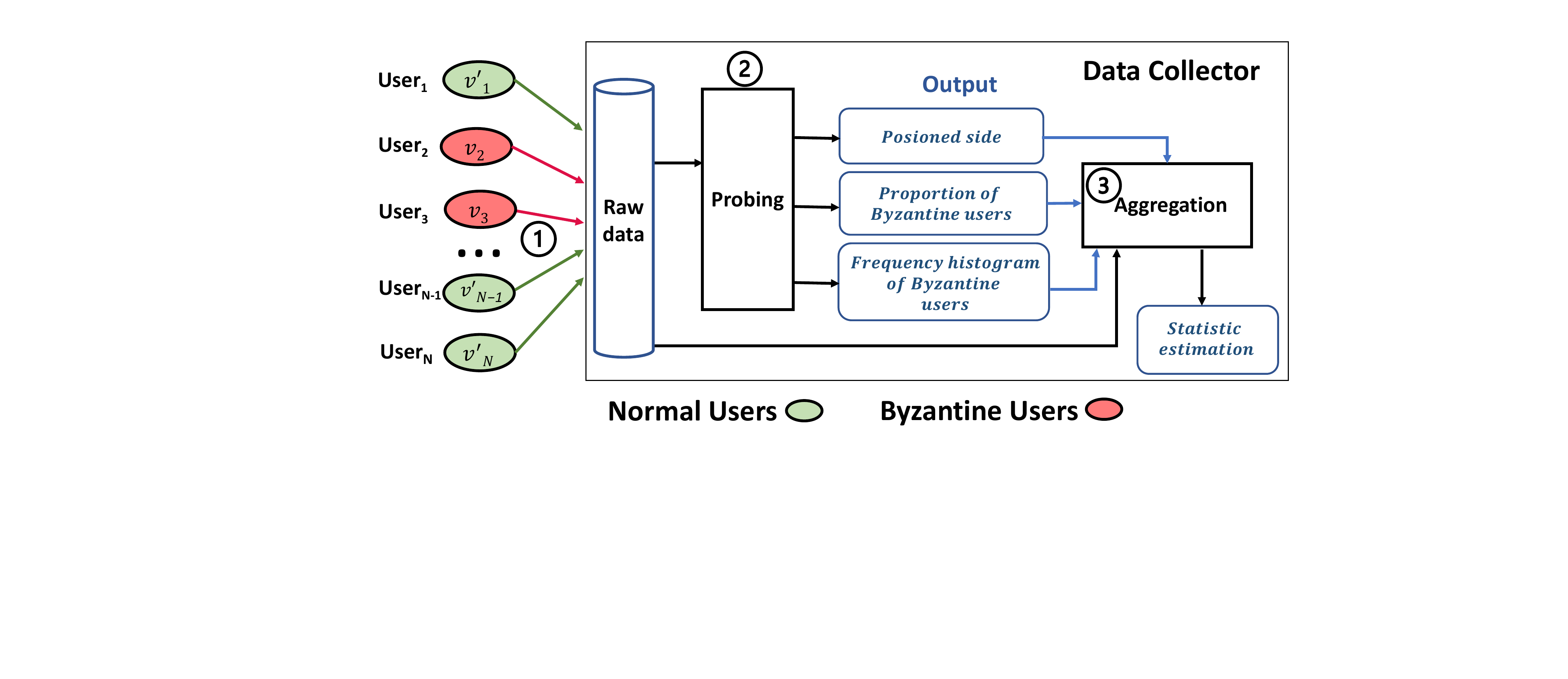}
\caption{System model and aggregation framework}
\label{systemframework}
\vspace{-3mm}
\end{figure} 
\section{A Baseline Protocol for Mean Estimation}
\label{EMF}
This section proposes a baseline protocol for numerical mean aggregation against Biased Byzantine Attacks (and thus General Byzantine Attacks). For the collector to probe and aggregate the mean, each user perturbs her numerical value twice with $\epsilon_\alpha$ and $\epsilon_\beta$ budgets respectively by an LDP perturbation mechanism. Here $\epsilon_\alpha+\epsilon_\beta=\epsilon$, the overall privacy budget, and $\epsilon_\alpha \ll \epsilon_\beta$. Upon receiving the two perturbed data sets $V'({\alpha})$ and $V'({\beta})$, the collector first probes on $V'({\alpha})$ for the three features of Byzantine users using Expectation-Maximization Filter (EMF), which will be elaborated in Section~\ref{sec:emf}. Based on these features, the collector then estimates the mean from $V'({\beta})$. In addition, since EMF can only probe on Biased Byzantine Attacks, we also need to have an initial guess on the true mean $O$. In the rest of this section, we first describe how to initialize $O$, and then present the EMF algorithm and mean estimation. In what follows we illustrate our algorithms using Piecewise Mechanism (PM) as the perturbation mechanism, and thus the perturbed value domain $[D_L, D_R]$ becomes $[-C,C]$.

\subsection{Initializing True Mean $O$}

According to Definition \ref{definition of BBA}, an initial true mean $O$ is needed to define a Biased Byzantine Attack. In essence, $O$ narrows down the poison values as in BBA they can only appear on one side of $O$. In order not to exclude potential poison values from our analysis, we should have a pessimistic initialization of $O$ as $O'$ so that $O'\le O$ if the poisoned side is on the right and vice versa. As such, the poison values' range of the true general Byzantine attack is a subset of that of the BBA in our analysis. The following theorem provides such a pessimistic initialization $O'$.

\begin{theorem}Given a collected value set $V'$ and an upper bound of Byzantine user proportion $\gamma_{sup}$, let $V'_B$ denote the set of poison values, and $T$ denote the set of the largest $\lceil \gamma_{sup}\cdot N\rceil$ values in $V'$. The following $O'$ is a pessimistic initialization of $O$, i.e., $O'\le O$ if the poisoned side is on the right and vice versa.
\begin{equation*}
\setlength{\abovedisplayskip}{2pt}
\setlength{\belowdisplayskip}{0pt}
\begin{split}
O'=\frac{1}{1-\gamma_{sup}}(\sum\limits_{v'_i\in V'}v'_i-\sum\limits_{v'_j\in T}v'_j)
\end{split}
\end{equation*}

\label{theo of cutoffpoint}
\end{theorem}
\begin{proof}
The true mean $O$ can be obtained by removing the effect of $Y$ from $V'$:
\begin{equation}
\setlength{\abovedisplayskip}{2pt}
\setlength{\belowdisplayskip}{2pt}
 O=\frac{\sum\limits_{v'_i\in V'}v'_i -\sum\limits_{v'_j\in V'_B} v'_j}{1-\gamma}.
 \label{equ111}
\end{equation}
Likewise, we can obtain $O'$ from removing the effect of $T$:
\begin{equation}
\setlength{\abovedisplayskip}{2pt}
\setlength{\belowdisplayskip}{2pt}
O'=\frac{\sum\limits_{v'_i\in V'}v'_i -\sum\limits_{v'_j\in T} v'_j}{1-\gamma_{sup}}.
 \label{equ222}
\end{equation}
Since the values in $T$ are the largest $\gamma_{sup}$ in $Y$, so we have
\begin{equation*}
\setlength{\abovedisplayskip}{2pt}
\setlength{\belowdisplayskip}{2pt}
\sum\limits_{v'_i\in T} v'_i \geq  \sum\limits_{v'_j\in V'_B} v'_j
\end{equation*}
Compare Equ. \ref{equ111} and Equ. \ref{equ222}, we have $O'\leq O$.
 \label{proof of cutoffpoint}
\end{proof}

According to our threat model, $\gamma_{sup}=0.5$ by default, and can be further reduced with prior knowledge.\footnote{For example, if the collector knows 60\% of the user population are using iOS, which is free from malware or virus that can turn them into Byzantine attackers, he can set $\gamma_{sup}=0.4$.} For ease of presentation, in the rest of this paper we simply set $O'=0$ unless otherwise stated and use the right hand side as the poisoned side.
\begin{figure}
\centering
\setlength{\abovecaptionskip}{0.cm}
\includegraphics[width=0.38\textwidth]{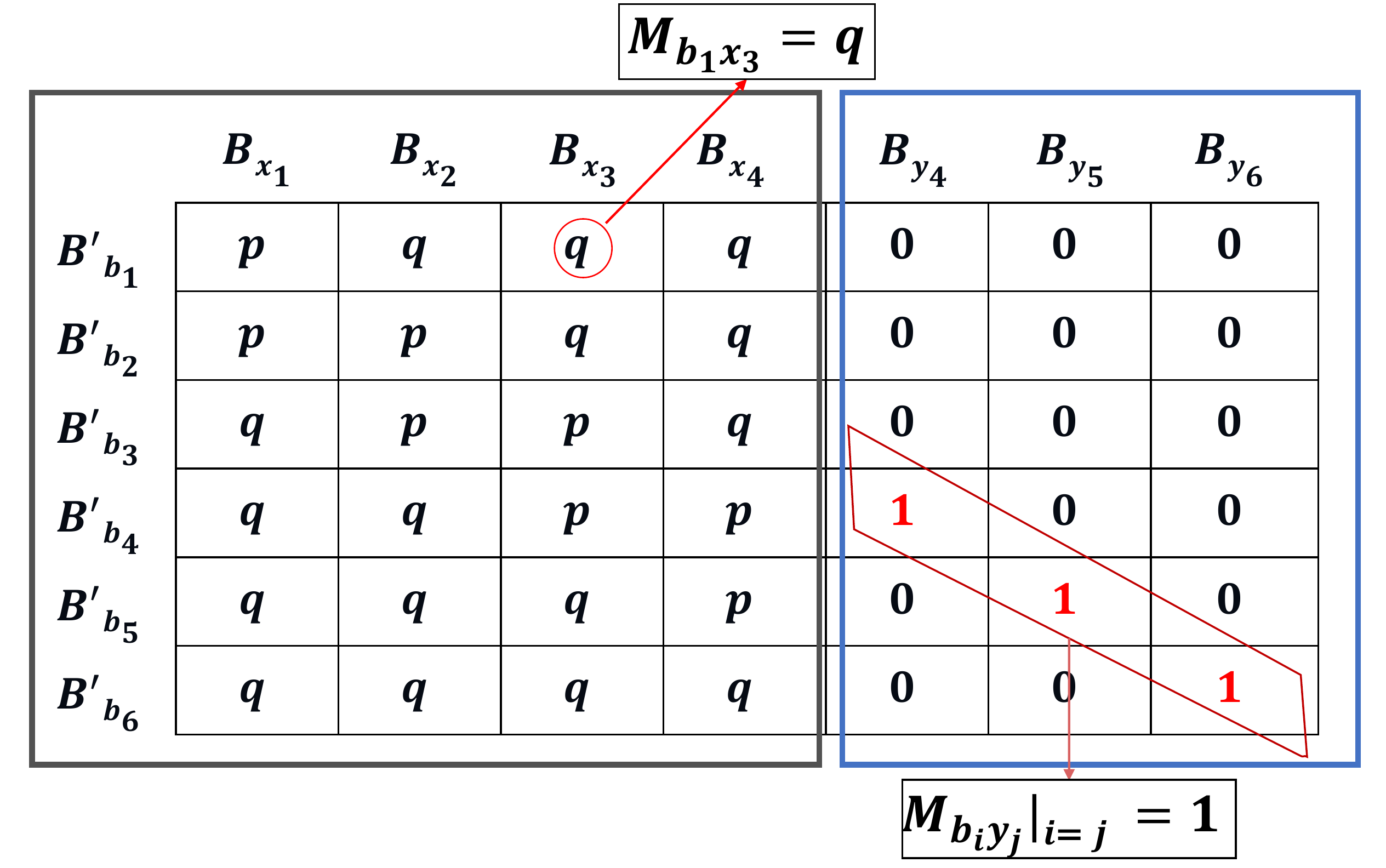}
\caption{Transform matrix M}
\label{the figure of matrix}
\vspace{-3mm}
\end{figure}
\subsection{Expectation-Maximization Filter}
\label{sec:emf}
Recall that normal users perturb their values by using Piecewise Mechanism in Algorithm \ref{PM}, whereas a biased Byzantine attacker chooses and sends a poison value directly. Such difference inspires us to reconstruct the original distribution and then probe some features of Byzantine users. Hereafter, we present a novel algorithm, namely the Expectation-Maximization Filter (EMF), to estimate such information by Maximum Likelihood Estimation (MLE)~\cite{li2020estimating}.

Let $V=\{v_1,v_2,..., v_{N}\}$ denote the original values of users, and $V'=\{v_1^{'},v_2^{'},..., v_{N}^{'}\}$ denote the collected ones. Since the probing only needs coarse precision, we discretise the original value domain $[-1,1]$ and perturbed value domain $[-C,C]$, respectively. Specifically, the former is discretised into $d$ buckets, and $d'$ for the latter, so let $\textbf{x}=\{x_1 ,...,x_d\}$ denote the frequency histogram of normal users in buckets $B_{x}=\{B_{x_1},...,B_{x_d}\}$, and $\{b_1,...,b_{d'}\}$ denote those in buckets $B'=\{B'_{b_1},...,B'_{b_d}\}$, respectively. Recall that $O'=0$,\footnote{When $O'\neq 0$, $[O',C]$ (resp. $[-C,O']$) will be discretised into $\lceil\frac{d'\cdot(C-O')}{2C}\rceil$ (resp. $\lceil\frac{d'\cdot(O'+C)}{2C}\rceil$) buckets.} so the poison values will only appear in the right half of the perturbed domain $[0,C]$ with $d'/2$ buckets. Let $\textbf{y}=\{y_{d'/2+1}, ..., y_{d'}\}$ denote the frequency histogram of Byzantine users in buckets $B_{y}=\{B_{y_{d'/2}},...,B_{y_{d'}}\}$.

With that said, the objective of EMF is to reconstruct the frequency histogram $F$, which comprises of the frequency counts of both normal users $\textbf{x}$ and poison values $\textbf{y}$ in the original value domain. That is, $F=\{\textbf{x,\ y}\}=\{x_1,...,x_k,...,x_d,y_{d'/2+1},...,y_j,...,y_{d'}\}$. To adopt MLE, we obtain the log-likelihood function of $F$ as follows:
\begin{equation}
\label{Object function}
\setlength{\abovedisplayskip}{2pt}
\setlength{\belowdisplayskip}{2pt}
\begin{split}
&l(F)=lnPr[V^{'}|F]=ln\prod_{i=1}^{N}Pr[v_i'|F]=
      \\ &\sum_{i=1}^{N}ln(\sum_{k=1}^{d}x_kPr[v'_i|v_i\in B_{x_k}]+\!\sum_{j=d'/2+1}^{d'}\!y_jPr[v'_i|v_i\in B_{y_j}])
\end{split}
\end{equation}
where $Pr[v'_i|v_i\in B_{x_k}]$ and $Pr[v'_i|v_i\in B_{y_j}]$ are constants, thus $l(F)$ is a concave function and EM algorithm can converge to the maximum likelihood estimator~\cite{bilmes1998gentle}.

To facilitate the derivation, we use a transform matrix \textbf{M} to capture the transformation probabilities $Pr[v'_i|v_i]$ for users. As shown in Fig. \ref{the figure of matrix}, \textbf{M} is a $d'\times (d+d'/2)$ matrix with all elements values in $[0,1]$. It consists of two parts. The left-hand side is a $d'\times d$ matrix for normal users, where $M_{b_ix_k}=Pr[v'_i\in B_{b_i}|v_i\in B_{x_k}]$. Given the original value $v$ of a normal user and a privacy budget $\epsilon$, let $p$ denote the probability that the perturbed value $v'$ falls in a bucket of $[l(v), r(v)]$ in Algorithm \ref{PM}, and $q$ otherwise. For example, in Fig. \ref{the figure of matrix}, given an original value in bucket $B_{x_3}$, $M_{b_1x_3}=q$ denotes the conditional probability that the output falls into bucket $B'_{b_1}$ is $q$. The right-hand side of \textbf{M} is a $d'\times (d'/2)$ matrix for poison values, and likewise, $M_{b_iy_j}=Pr[v'_i\in B_{b_i}|v_i\in B_{y_j}]$. Since Byzantine users send poison values directly to the data collector, $M_{b_iy_j}=1$ when $i=j$, and $M_{b_iy_j}=0$ when $i\neq j$.

Now that the transform matrix $\textbf{M}$ is derived, the EMF algorithm reconstructs the frequency histogram ($\hat{x}$ for normal values and $\hat{y}$ for poison values) as illustrated in Algorithm \ref{the algorithm of EMF}, which can be derived using the methodology described in the literature~\cite{li2020estimating}. Note that $c_i$ in $E$-step denotes the count of perturbed values from $V'$ in bucket $B'_{b_i}$. First, the algorithm assigns some non-zero initial values to $\hat{x}$ and $\hat{y}$, subject to $\hat{x}+\hat{y}=1$ (line 1). Next, it executes the EM algorithm, alternating the $E$-step and the $M$-step. The $E$-step evaluates the log-likelihood expectation by the observed counts $c_i$ and the current $\hat{x}$ and $\hat{y}$ (lines 3-9), and the $M$-step calculates $\hat{x}$ and $\hat{y}$ that maximize the expected likelihood (lines 10-15) as inputs for the next round $E$-step. Finally, the algorithm returns the estimated frequency histogram $\hat{x}$ and $\hat{y}$ when the convergence condition is met (line 16).
{\small{
\begin{algorithm}[] \small
\caption{Expectation-Maximization Filter (EMF)}
 \begin{flushleft}
{\bf Input:}
Transform matrix $\textbf{M}$ and collected values $V'$\\
{\bf Output:}
$\hat{x}$,  $\hat{y}$
 \end{flushleft}
\begin{algorithmic}[1]
\State
Initialization: $\hat{x_k}=\hat{y_j}=\frac{1}{d+d'/2}$
\While{not converge}\State
\textbf{E-step:}
$\forall\ k\in\{1,...,d/2\}$
\State
\indent $P_{x_k}=\hat{x_k}\sum_{i=1}^{d'}c_i\frac{M_{b_ix_k}}{\sum_{t=1}^{d}M_{b_ix_t}\hat{x_t}}$
 \State
$\forall\ k\in\{d/2+1,...,d\}$
 \State
\indent $P_{x_k}=\hat{x_k}\sum_{i=1}^{d'}c_i\frac{M_{b_ix_k}}{\sum_{t=1}^{d}M_{b_ix_t}\hat{x_t}+M_{b_iy_i}\hat{y_i}}$
 \State
 $\forall\ j\in\{d'/2+1,...,d'\}$
  \State
 \indent
 $P_{y_j}=\hat{y_j}\sum_{i=1}^{d'}c_i\frac{M_{b_iy_j}}{\sum_{t=1}^{d}M_{b_ix_t}\hat{x_t}+M_{b_iy_i}\hat{y_i}}$
\State
\textbf{M-step:}
$\forall\ k\in \{1,...,d\}$
\State
   \indent \indent  $\hat{x_k}=\frac{P_{x_k}}{\sum_{i=1}^{d}P_{x_i}+\sum_{j=d'/2+1}^{d'}P_{y_j}}$
   \State
$\forall\ j\in \{d'/2+1,...,d'\}$
\State
\indent \indent  $\hat{y_j}=\frac{P_{y_j}}{\sum_{k=1}^{d}P_{x_k}+\sum_{i=d'/2+1}^{d'}P_{y_i}}$
\EndWhile
\algtext*{EndWhile}
\State \Return $\hat{x}$, $\hat{y}$
\end{algorithmic}
\label{the algorithm of EMF}
\end{algorithm}}}

\subsection{Byzantine Feature Estimation}
\label{Probing}
The outputs of EMF $\hat{x}$ and $\hat{y}$ on the collected value set $V'$ is used to probe Byzantine users' features. The first feature is the poisoned side $\mathcal{S}$, and the pseudo-code is shown in Algorithm \ref{Poisoned side2}. First, it applies EMF separately to the transform matrices $\textbf{M}_R$ and $\textbf{M}_L$ (lines 1-2). If the probing range is $[0,C]$, $\textbf{M}_R$ is used with buckets $\{B_{y_{d'/2+1}},...,B_{y_{d'}}\}$ in the right-hand matrix of $\textbf{M}$ for poison values; otherwise $\textbf{M}_L$ is used with buckets $\{B_{y_{1}},...,B_{y_{d'/2}}\}$. Then, to determine which side is more probable, it calculates the variance of $\hat{x_L}$ and $\hat{x_R}$ (lines 3-4), where $\hat{x_L}$ (resp. $\hat{x_R}$) denotes the frequency histogram for normal users when the input matrix is $\textbf{M}_L$ (resp. $\textbf{M}_R$). Finally, it chooses the poisoned side with a lower variance (lines 5-8).
\begin{algorithm}[] \small
\caption{Poisoned Side Probing}
\begin{flushleft}
 {\bf Input:}
Transform matrix $\textbf{M}_L$, transform matrix $\textbf{M}_R$ and collected values $V'$\\
 {\bf Output:}
Poisoned side $\mathcal{S}$
\end{flushleft}
\begin{algorithmic}[1]
\State
$(\hat{x_L},\hat{y_L})=EMF(\textbf{M}_L,V')$;
\State
$(\hat{x_R},\hat{y_R})=EMF(\textbf{M}_R,V')$;
\State
$Var_L=Variance(\hat{x_L})$;
\State
$Var_R=Variance(\hat{x_R})$;
\If{$Var_L<Var_R$} \Return $Left\ side$
\Else \ \Return $Right\ side$
\EndIf
\end{algorithmic}
\label{Poisoned side2}
\end{algorithm}

Upon knowing the poisoned side, we can now determine the remaining two features. The second one of Byzantine users is the frequency histogram $\hat{y}$ of poison values which can be chosen from $\hat{y_L}$ or $\hat{y_R}$ in Algorithm \ref{Poisoned side2} accordingly. The third one is the proportion of Byzantine users $\hat{\gamma}$, which can be derived from $\hat{y}$, the estimated frequency histogram of poison values by EMF:
\begin{equation}
\setlength{\abovedisplayskip}{1pt}
\setlength{\belowdisplayskip}{1pt}
\label{The ratio of Byzantine users}
\hat{\gamma}=\sum_{j=d'/2+1}^{d'}{\hat{y_j}}=\frac{\hat{m}}{N}\approx  \frac{m}{N}=\gamma,
\end{equation}
where $\hat{m}$ denotes the estimated number of Byzantine users, and $\gamma$ denotes the true proportion of Byzantine users.

When $\epsilon\rightarrow0$, the estimated frequency histogram $\hat{x}$ of normal users converges to a uniform distribution, whereas that of Byzantine users $\hat{y}$ converges to the true distribution of poison values. This, along with the correctness of Algorithm \ref{Poisoned side2} and Equ. \ref{The ratio of Byzantine users}, can be proved by the following theorem.
\begin{theorem}
\label{lemma of poisoned direction and the proportion of Byzantine users}
Let $a=\{a_{d'/2+1},...,a_{d'}\}$ denote the count of poison values in buckets $\{B_{d'/2+1},...,B_{d'}\}$. When $\epsilon\rightarrow0$, the convergence results are $\hat{x_k}=\frac{n}{Nd}(k\in\{1,...,d'\})$ and $\hat{y_j}=\frac{a_j}{N}(j\in\{d'/2+1,...,d'\})$.
\end{theorem}

\begin{proof}
When $\epsilon\rightarrow 0$, all inputs from normal users are equally perturbed into $d'$ buckets with probability $\frac{1}{d'}$, which leads to a uniform distribution. Let $a=\{a_{d'/2+1},...,a_{t},...,a_{d'}\}$ denote the count of poison values in corresponding buckets $B'_{b_i}$, and we have $c_t\rightarrow\frac{n}{d'}, t\in\{1,...,d'/2\}$ and $c_t\rightarrow \frac{n}{d'}+a_{t}, t\in\{d'/2+1,...,d'\}$.

Note that $\sum_{k=1}^{d} \hat{x_k}+\sum_{j=d'/2+1}^{d'}\hat{y_j}=1$, the Lagrangian function of Equ. \ref{Object function} can be written as:
\begin{equation*}
\setlength{\abovedisplayskip}{2pt}
\setlength{\belowdisplayskip}{2pt}
L(F)=l(F)+\omega(\sum_{k=1}^{d} \hat{x_k}+\sum_{j={d'/2+1}}^{d'}\hat{y_j}-1).
\end{equation*}
Let all first-order partial derivatives of L w.r.t. $\hat{x_k}$ and $\hat{y_j}$ equal zero
\begin{equation*}
\begin{split}
\setlength{\abovedisplayskip}{2pt}
\setlength{\belowdisplayskip}{2pt}
&\frac{\partial L(F)}{\partial \hat{x_k}}=\sum_{t=1}^{d'}c_t\frac{\frac{1}{d'}}{\sum_{k=1}^{d}\hat{x_k}\frac{1}{d'}}+\omega, \  k\in\{1,...,\frac{d}{2}\}\\
\setlength{\abovedisplayskip}{2pt}
\setlength{\belowdisplayskip}{2pt}
&\frac{\partial L(F)}{\partial \hat{x_k}}=\sum_{t=1}^{d'}c_t\frac{\frac{1}{d'}}{\sum_{k=1}^{d}\hat{x_k}\frac{1}{d'}+\hat{y_t}}+\omega, \  k\in\{\frac{d}{2}+1,...,d\}\\
 &\frac{\partial L(F)}{\partial \hat{y_j}}=c_j\frac{1}{\sum_{k=1}^{d}\hat{x_k}\frac{1}{d'}+\hat{y_j}}+\omega,\  j\in\{\frac{d'}{2}+1,...,d'\}
\end{split}
\end{equation*}
we have:
\begin{equation}
\label{resultofwhole12}
\setlength{\abovedisplayskip}{2pt}
\setlength{\belowdisplayskip}{2pt}
\hat{x_k}\rightarrow \frac{n}{Nd},\ \ \hat{y_j}\rightarrow \frac{a_j}{N},j\in\{d'/2+1,...,d'\}, \ \omega\rightarrow -N.
\end{equation}

According to Equ. \ref{resultofwhole12}, the frequency histogram $\hat{x}$ of normal users converges to a uniform distribution, whereas that of Byzantine users $\hat{y}$ converges to the true distribution of poison values.
 \label{proof of cutoffpoint}
\end{proof}

When $\epsilon\rightarrow 0$, $\hat{y}$ converges to the true distribution of poison values, so Equ. \ref{The ratio of Byzantine users} can well estimate the proportion of Byzantine users. The fact that $\hat{x}$ converges to a uniform distribution explains the correctness of Algorithm \ref{Poisoned side2} --- if we choose the correct side (say, the right side), the distribution of $\hat{x}$ should be a uniform distribution with a very small variance; If we select the left side and use $M_L$ as the EMF input, only collected values from $[-C,0]$ can converge into $B_{y_i}$. Therefore, poison values in $[0,C]$ can only converge to $B_{x_i}$, which are regarded as normal values. The convergence result $\hat{x_L}$ will be biased towards $[0,C]$, which is far from a uniform distribution. In this case, the variance of $\hat{x}$ appears to be far greater than selecting the reverse side.

Since $\epsilon_\alpha \ll \epsilon_\beta$, we choose $V'(\alpha)$ over $V'(\beta)$ to probe Byzantine features, which is more accurate than using the latter according to Theorem \ref{lemma of poisoned direction and the proportion of Byzantine users}.

\subsection{Mean Estimation}
\label{Statistic Estimation}
Now that we have probed the features of Byzantine users from $V'(\alpha)$, we can leverage them to estimate the mean in $V'(\beta)$. Since the poison values in $V'(\alpha)$ and $V'(\beta)$ form a unified attack, their deviation $D$ from the true mean $O$ is assumed the same.\footnote{Section ~\ref{DAP} has further eliminated this assumption.} Therefore, we have:
\begin{equation*}
\setlength{\abovedisplayskip}{2pt}
\setlength{\belowdisplayskip}{2pt}
\begin{split}
D=\frac{m(\mathbb{M}_{\alpha}-O)}{N}=\frac{m(\mathbb{M}_{\beta}-O)}{N},
\end{split}
\end{equation*}
where $\mathbb{M}_{\alpha}$ denotes the mean of poison values in $V'(\alpha)$, and $\mathbb{M}_{\beta}$ for $V'(\beta)$. Thus, $\mathbb{M}_{\alpha}=\mathbb{M}_{\beta}$.

$\mathbb{M}_{\alpha}$ can be obtained from the frequency histogram estimation $\widehat{y(\alpha)}$, i.e., the convergence result for poison values of $V'(\alpha)$:
\begin{equation}
\label{mean of poison value1}
\setlength{\abovedisplayskip}{1pt}
\setlength{\belowdisplayskip}{1pt}
\begin{split}
\mathbb{M}_{\alpha}=\frac{\sum\widehat{y_j({\alpha})}\nu_j}{\sum\widehat{y_j({\alpha})}},\ j\in\{d'/2+1,...,d'\},
\end{split}
\end{equation}
where $\nu_j$ denotes the median value of $B_{y_j}$.

Because the perturbed mean gained from PM is an unbiased estimation of the original one, the mean $\widetilde{M}$ of normal users with $\epsilon_{\beta}$ can be derived as:
\begin{equation}
\label{mean of normal value}
\setlength{\abovedisplayskip}{1pt}
\setlength{\belowdisplayskip}{1pt}
 \widetilde{M}=\frac{\sum\limits_{v'({\beta)}\in V'({\beta})}{v'({\beta})}-\hat{m} \mathbb{M}_{\beta}}{N-\hat{m}}.
\end{equation}
Substitute Equ. \ref{mean of poison value1} into Equ. \ref{mean of normal value}. Notice that $\mathbb{M}_{\alpha}=\mathbb{M}_{\beta}$, $\widetilde{M}$ can be written as:
\begin{equation*}
\setlength{\abovedisplayskip}{1pt}
\setlength{\belowdisplayskip}{1pt}
\begin{split}
 \widetilde{M}=\frac{\sum\limits_{v'({\beta)}\in V'({\beta})}{v'({\beta})}-\hat{m}\frac{\sum\widehat{y_j({\alpha})}\nu_j}{\sum\widehat{y_j({\alpha})}}}{N-\hat{m}}, \ j\in\{d'/2+1,...,d'\}.
\end{split}
\end{equation*} 
\section{Differential Aggregation Protocol}
\label{DAP}
The baseline protocol has one flaw as they may choose to send some normal values at $\epsilon_\alpha$ to hide their attacking intention, and send poison values at $\epsilon_\beta$. Since the probed Byzantine features become less accurate, the accuracy of the mean estimation is degraded significantly. In essence, this flaw is caused by {\bf two separate and fixed} $\epsilon$'s in the baseline protocol, so Byzantine users can easily tell the smaller one is $\epsilon_\alpha$ and serves for probing Byzantine features, whereas the larger one serves for statistics estimation.

This motivates us to enhance the baseline protocol with a {\bf single but random} $\epsilon$. In this section, we present a multi-group collection mechanism, namely Differential Aggregation Protocol (DAP). Our idea is to randomly assign users into $h$ groups, each with its own $\epsilon$ setting. The collector performs EMF in each group to probe Byzantine users' features, as in the baseline protocol. And then the collector estimates a mean from each group, based on which an inter-group mean is aggregated. There are a few advantages. First, Byzantine users cannot differentiate if their values are used for probing or estimation and thus take the above strategy. Second, there is no need to split privacy budgets for users, which can improve the estimation accuracy. Third, this protocol can naturally handle users with different privacy budgets.

As illustrated in Fig. \ref{sysframe of DAP}, the workflow of the DAP protocol has five stages:
\begin{enumerate}
\setlength{\itemsep}{0pt}
\setlength{\parsep}{0pt}
\setlength{\parskip}{0pt}
\item \textbf{Grouping.} The data collector allocates users into groups and assigns each group with a dedicated privacy budget.
\item \textbf{Perturbation.} Users in each group perturb their values according to their assigned privacy budgets and send them to the data collector.
\item \textbf{Probing.} The data collector executes the EMF algorithm for each group to probe the Byzantine features.
\item \textbf{Intra-group Estimation.} The data collector attains a mean from each group from the output of EMF.
\item \textbf{Inter-group Aggregation.} The data collector combines estimated means from all groups into one.
\end{enumerate}

\begin{figure}
\centering
\includegraphics[width=0.45\textwidth]{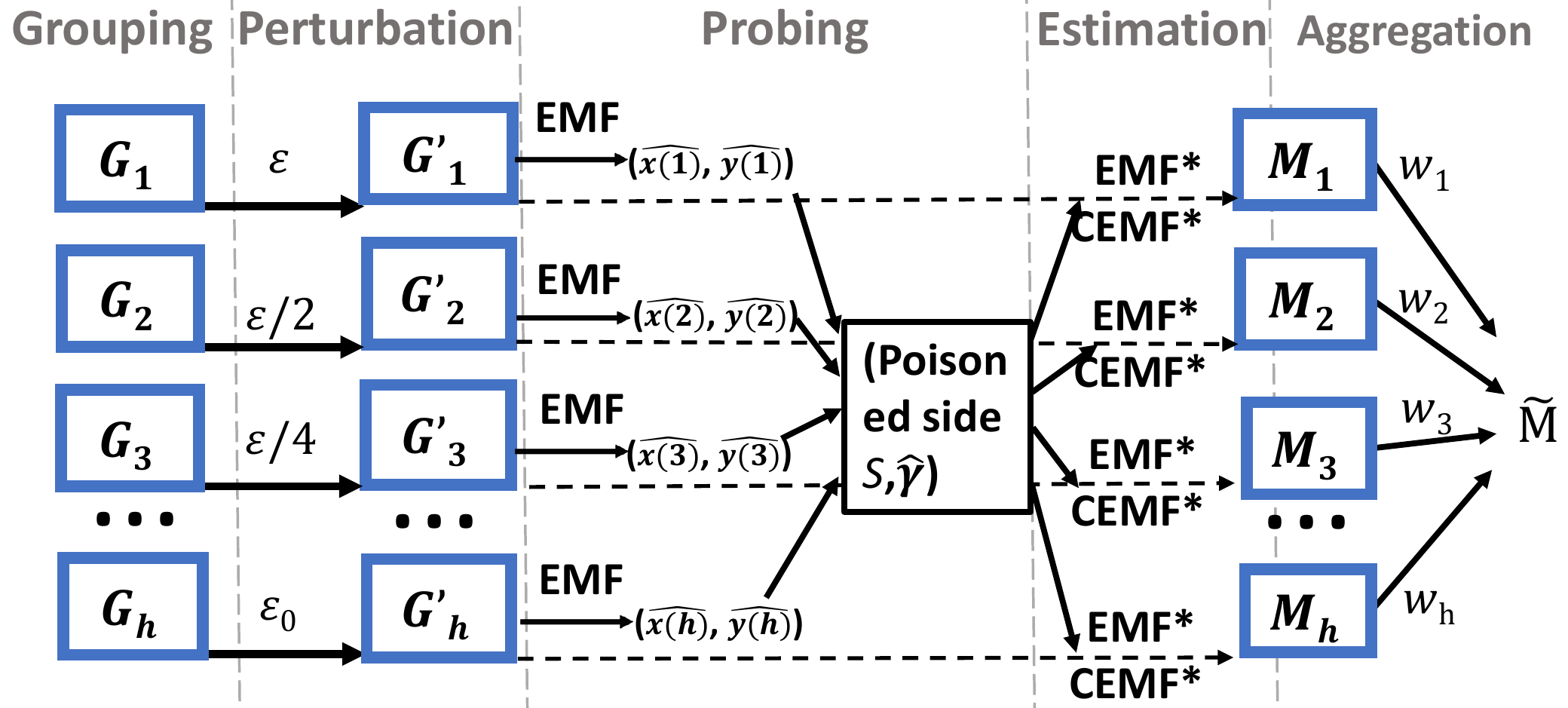}
\setlength{\abovecaptionskip}{0.cm}
\caption{Differential Aggregation Protocol}
\label{sysframe of DAP}
\end{figure}

\subsection{Grouping}First of all, the data collector determines a minimal acceptable privacy budget $\epsilon_0$ to bound the perturbation noise of normal values according to the privacy budget $\epsilon$ for users. Then data collector creates $h=\lceil log_2(\epsilon/\epsilon_0)\rceil+1$ equal-sized groups\footnote{That is, each group has the same number of users. }, denoted by $\{G_1, ...,G_t,...,G_h\}$, with decreasing budgets $\{\epsilon,\frac{1}{2}\epsilon,\frac{1}{4}\epsilon, ...,\epsilon_t , ...,\epsilon_0\}$.\footnote{Without loss of generality, we assume $\epsilon/\epsilon_0$ is a power of $2$.} Users are randomly assigned to the groups by the collector and perturb their values according to the $\epsilon_t$ of the groups they belong to. To guarantee all users have the same privacy budget, those assigned with smaller $\epsilon_t$ perturb and report multiple times until the overall privacy budget is depleted. Let $V'_t$ denote the collected values from group $G_t$ and $N_t=\frac{\epsilon N}{\epsilon_t h}$ denote the number of collected values from group $G_t$.

\subsection{Intra-group Mean Estimation}Upon collecting all perturbed values, we carry out EMF in each group to probe the Byzantine users' features from the frequency histogram, as described in Section \ref{EMF}. Specifically, we can obtain the poisoned side, the proportion of Byzantine users and the frequency histogram of poison values $\widehat{y(t)}$ in $G_t$. As such, the mean $M_t$ of group $G_t$ with privacy budget $\epsilon_t$, can be estimated by removing poison values:
\begin{equation}
\label{mean in Gt}
\setlength{\abovedisplayskip}{2pt}
\setlength{\belowdisplayskip}{2pt}
 M_t=\frac{\sum\limits_{v'_t\in V'_t}{v'_t}-\hat{m_t}\sum_{j=d'/2+1}^{d'}(\widehat{y_j(t)}\nu_j)}{N_t-\hat{m_t}},
\end{equation}
where $\widehat{y_j(t)} \in \widehat{y(t)}$, $\hat{m_t}=N_t\sum{\widehat{y(t)}}$ and $\nu_j$ denotes the median value of $B_{y_j}$.

 However, according to Theorem \ref{lemma of poisoned direction and the proportion of Byzantine users}, both $\hat{y_t}$ and $\hat{m_t}$ become less accurate as the user population is diluted by grouping, especially when the $\epsilon_t$ of a group is large. Hereafter, we propose two post-processing schemes, namely, EMF* and CEMF*, to further enhance the accuracy of convergence results and mean estimation in Equ. \ref{mean in Gt} in each group.

\textbf{EMF*.} It utilizes the proportion of Byzantine users $\hat{\gamma}$ (with small $\epsilon_t$) probed from EMF to improve the convergence process, by imposing $\sum \hat{y}=\hat{\gamma}$, $\sum \hat{x}=1-\hat{\gamma}$ as two additional conditions to satisfy in EMF. Accordingly, the maximization problem in \textbf{M} steps of EMF becomes:
\begin{equation}
\setlength{\abovedisplayskip}{1pt}
\setlength{\belowdisplayskip}{1pt}
\begin{split}
&\mathop{\arg\max}\limits_{\hat{x},\hat{y}}\ \sum_{k=1}^{d}P_{x_k}ln\hat{x_k}+\sum_{j=d'/2+1}^{d'}P_{y_j}ln\hat{y_j},  \\
&subject\ to\ \sum_{k=1}^{d}\hat{x_k}=1-\hat{\gamma},\ \sum_{j=d'/2+1}^{d'}\hat{y_j}=\hat{\gamma}
\end{split}
\label{Mmaximum estimator}
\end{equation}

EMF* can improve EMF in terms of the accuracy of converged poison value histogram because the additional restrictions eliminate those infeasible poison values. To resolve Equ. \ref{Mmaximum estimator}, we have the following theorem.
\begin{theorem}
\label{theorem:remf}
The maximum in Equ. \ref{Mmaximum estimator} is reached when the output frequency histograms are
\begin{equation*}
\setlength{\abovedisplayskip}{2pt}
\setlength{\belowdisplayskip}{0pt}
\hat{x_k}=(1-\hat{\gamma})\frac{Px_k}{\sum_{k=1}^{d}{Px_k}},\  \hat{y_j}=\hat{\gamma}\frac{Py_j}{\sum_{j=d'/2+1}^{d'}{Py_j}}.
\end{equation*}
\end{theorem}

{\small{
\begin{algorithm}[t]\small
\caption{EMF with Restrictions (EMF*)}
\begin{flushleft}
 {\bf Input:}
Transform matrix $\textbf{M}$, collected values $V'$, the estimated proportion of Byzantine users $\hat{\gamma}$\\
 {\bf Output:}
$\hat{x},\hat{y}$
\end{flushleft}
\begin{algorithmic}[1]
\State
Initialization: $\hat{x_k}=\hat{y_j}=\frac{1}{d+d'/2}$
\While{not converge}\State
\textbf{E-step:}
\textbf{\State }
Same as the E-step in EMF in Algorithm \ref{the algorithm of EMF}
 \State
\textbf{M-step:}
\State
$\forall\ k\in \{1,...,d\}$
\State
   \indent \indent  $\hat{x_k}=(1-\hat{\gamma})\frac{P_{x_k}}{\sum_{i=1}^{d}P_{x_i}}$
   \State
$\forall\ j\in \{d'/2+1,...,d'\}$
\State
\indent \indent  $\hat{y_j}=\hat{\gamma}\frac{P_{y_j}}{\sum_{i=d'/2+1}^{d'}P_{y_i}}$
\EndWhile
\State \Return $\hat{x},\hat{y}$
\end{algorithmic}
\label{the algorithm of REMF}
\end{algorithm}
 }}
\begin{proof}
We apply Lagrangian multiplier method to derive maximal value of Euq. \ref{Mmaximum estimator}. Let
\begin{equation*}
\setlength{\abovedisplayskip}{2pt}
\setlength{\belowdisplayskip}{2pt}
\begin{split}
\mathcal{L}&=\sum_{k=1}^{d} P_{x_k}ln\hat{x_k}+\sum_{j=d'/2+1}^{d'}P_{y_j}ln\hat{y_j}\\
&+\lambda_1(\sum_{k=1}^{d}\hat{x_k}-1+\hat{\gamma})+\lambda_2(\sum_{j=d'/2+1}^{d'}\hat{y_j}-\hat{\gamma}),
\end{split}
\end{equation*}
where $\lambda_1$ and $\lambda_2$ are two constants and the first-order partial derivatives of $\mathcal{L}$ w.r.t. $\hat{x_k}$ and $\hat{y_j}$ are
\begin{equation*}
\setlength{\abovedisplayskip}{2pt}
\setlength{\belowdisplayskip}{2pt}
\frac{\partial \mathcal{L}}{\partial \hat{x_k}}=\frac{P_{x_k}}{\hat{x_k}}+\lambda_1,\ \frac{\partial \mathcal{L}}{\partial \hat{y_j}}=\frac{P_{y_j}}{\hat{y_j}}+\lambda_2.
\end{equation*}
Let $\frac{\partial l}{\partial \hat{x_k}}$ and $\frac{\partial l}{\partial \hat{y_j}}$ be zero, and we have
\begin{equation}
\setlength{\abovedisplayskip}{2pt}
\setlength{\belowdisplayskip}{2pt}
 P_{x_k}+\lambda_1\hat{x_k}=0,\ P_{y_j}+\lambda_2\hat{y_j}=0.
 \label{the result of derivation2}
\end{equation}
From the restrictions in Equ. \ref{Mmaximum estimator}, we can deduce that
\begin{equation*}
\setlength{\abovedisplayskip}{2pt}
\setlength{\belowdisplayskip}{2pt}
\lambda_1=\frac{\sum_{i=1}^{d}P_{x_i}}{\hat{\gamma}-1},\  \lambda_2=\frac{\sum_{i=d'/2+1}^{d'}P_{y_i}}{-\hat{\gamma}}.
\end{equation*}
Replacing them in Equ. \ref{the result of derivation2}, we can reach that
\begin{equation*}
\setlength{\abovedisplayskip}{2pt}
\setlength{\belowdisplayskip}{2pt}
\hat{x_k}=(1-\hat{\gamma})\frac{Px_k}{\sum_{i=1}^{d}P_{x_i}},\  \hat{y_j}=\hat{\gamma}\frac{Py_j}{\sum_{i=d'/2+1}^{d'}P_{y_i}}.
\end{equation*}
 \label{proof of cutoffpoint}
\end{proof}
Algorithm~\ref{the algorithm of REMF} shows the pseudo-code of EMF*, where the M-step from EMF has been modified  with the results of Theorem~\ref{theorem:remf}. Note that instead of replacing EMF, EMF* is a post-processing of EMF, as the former needs the $\hat{\gamma}$ of latter as input.

\textbf{CEMF*.} This post-processing further considers a common scenario where Byzantine users inject poison values into a small domain only. In such cases, the mean estimation from EMF and EMF* is relatively poor since the frequency counts of some buckets in the histogram can be extremely high, which cannot be restored very accurately by EMF/EMF*. As such, we propose an alternative post-processing, namely EMF* with concentration (CEMF*) by suppressing those buckets $B_{y_j}$ with small frequency counts $\hat{y_j}$ in EMF's output. That is, we treat these buckets as if no poison values are there. Let $\mathcal{B}$ denote the set of buckets for poison values that are chosen by Byzantine users, and $\overline{\mathcal{B}}$ denote the set of remaining buckets. Here we suppress all buckets in $\overline{\mathcal{B}}$ first, which will be proved to achieve the optimal performance later in Theorem \ref{CEMF}. The log-likelihood function then becomes:
\begin{equation*}
\setlength{\abovedisplayskip}{1pt}
\setlength{\belowdisplayskip}{1pt}
\begin{split}
&l(F)=lnPr[V^{'}|F]=ln\prod_{i=1}^{N}Pr[v_i'|F]=
      \\&\sum_{i=1}^{N}ln(\sum_{k=1}^{d}\hat{x_k}Pr[v'_i|v_i\in B_{x_k}]+\sum_{B_{y_j}\in \mathcal{B}}\hat{y_j}Pr[v'_i|v_i\in B_{y_j}]).\\
\end{split}
\end{equation*}

Our objective is to maximize $l(F)$, subject to $\sum_{k=1}^{d}\hat{x_k}=1-\hat{\gamma}$ and $\sum_{B_{y_j}\in \mathcal{B}}\hat{y_j}=\hat{\gamma}$. The maximization can be obtained following Algorithm \ref{the algorithm of REMF}, by initially setting $\hat{y_i}=0$ (i.e., those buckets in $\overline{\mathcal{B}}$).

The following theorem shows CEMF* can improve the estimation accuracy if we have the knowledge about $\overline{\mathcal{B}}$ and suppress some buckets inside. Specifically, CEMF* achieves its optimal performance when all buckets in $\overline{\mathcal{B}}$ are suppressed:

\begin{theorem}
\label{CEMF}
The converged frequency histograms monotonically approach the actual ones with respect to $||\overline{\mathcal{B}}||$, the number of suppressed buckets.
\end{theorem}
\begin{proof}
Let $\mathcal{B}=\{B_{y_1},...,B_{y_t}\}$, and $\overline{\mathcal{B}}=\{B_{y_{t+1}},...,B_{y_{d'}}\}$. We start from the state where all buckets hold non-zero values, i.e., $B_{y_i}\neq 0, i\in\{1,...,d'\}$, and reconstruct the frequency histogram for poison values in $[-C, C]$. The likelihood estimator in Equ. \ref{Object function} becomes
\begin{equation*}
\setlength{\abovedisplayskip}{2pt}
\setlength{\belowdisplayskip}{2pt}
\begin{split}
l(F)&=\sum_{i=1}^{N}ln(\sum_{k=1}^{d}\hat{x_k}Pr[v_i'|v_i\in B_{x_k}]\\&+\sum_{j=1}^{d'}\hat{y_j}Pr[v_i'|v_i\in B_{y_j}])
      \\& =\sum_{t=1}^{d'}c_tln(\sum_{k=1}^{d}\hat{x_k} M_{b_t x_k} +\sum_{j=1}^{d'}\hat{y_j} M_{b_t y_j}).
\end{split}
\end{equation*}
Note that $\sum_{k=1}^{d} \hat{x_k}+\sum_{j=1}^{d'}\hat{y_j}=1$, we employ the Lagrangian multiplier method to derive the extremum. The Lagrangian function can be written as:
\begin{equation*}
\setlength{\abovedisplayskip}{2pt}
\setlength{\belowdisplayskip}{2pt}
L(F)=l(F)+\lambda(\sum_{k=1}^{d} \hat{x_k}+\sum_{j=1}^{d'}\hat{y_j}-1).
\end{equation*}
Let all first-order partial derivatives of L w.r.t. $\hat{x_k}$ and $\hat{y_j}$ equal zero
\begin{equation*}
\begin{split}
\setlength{\abovedisplayskip}{1pt}
\setlength{\belowdisplayskip}{1pt}
&\frac{\partial L(F)}{\partial \hat{x_k}}=\sum_{t=1}^{d'}\frac{c_tM_{b_tx_k}}{\sum_{k=1}^{d}\hat{x_k}M_{b_tx_k}+\hat{y_t}}+\lambda=0, \  k\in\{1,...,d\}\\
 &\frac{\partial L(F)}{\partial \hat{y_j}}=\sum_{t=1}^{d'}\frac{c_tM_{b_ty_j}}{\sum_{k=1}^{d}\hat{x_k}M_{b_tx_k}+\hat{y_t}}+\lambda=0,\  j\in\{1,...,d'\}
\end{split}
\end{equation*}
we have:
\begin{equation*}
\label{resultofwhole}
\setlength{\abovedisplayskip}{2pt}
\setlength{\belowdisplayskip}{2pt}
\hat{x_k}=0,\ k\in\{1,...,d\},\ \ \hat{y_j}=\frac{c_j}{N}, \ j\in\{1,...,d'\},\ \lambda=-N.
\end{equation*}
This result shows all collected values converge to poison values if no bucket is suppressed, and we can obtain:
\begin{equation*}
\setlength{\abovedisplayskip}{2pt}
\setlength{\belowdisplayskip}{2pt}
(\sum_{k=1}^{d}\hat{x_k}+\sum_{j=1}^{t}\hat{y_j})\bigg|_{y_{i}\neq 0, i\in\{1,...,d'\}}=\sum_{j=1}^{t}\frac{c_{j}}{N}.
\end{equation*}

When we suppress the bucket $B_{y_{d'}}$ (by setting $\hat{y_{d'}}=0$) and carry out EMF, the collected values in $B'_{b_{d'}}$ can only converge to $B_{x_k}(k\in\{1,...,d\})$, but not $B_{y_j}(j\in\{1,...,d'\})$. Hence, every $\hat{x_k}$ will increase. Therefore, suppressing $B_{y_{d'}}$ leads to the increase of all $\hat{x_k}$, which in turn results in the decrease of all $\hat{y_j}$. However, since the decrease of $\hat{y_j}(j\in\{1,...,t\})$ is a part of increment of $\hat{x}$, we can figure out $(\sum_{k=1}^{d}\hat{x_k}+\sum_{j=1}^{t}\hat{y_j})\bigg|_{y_{i}\neq 0, i\in\{1,...,d'\}}\leq (\sum_{k=1}^{d}\hat{x_k}+\sum_{j=1}^{t}\hat{y_j})\bigg|_{y_{d'}=0}$.
Suppress all buckets one by one in $\overline{\mathcal{B}}$ similarly, we have: \\$
(\sum_{k=1}^{d}\hat{x_k}+\sum_{j=1}^{t}\hat{y_j})\bigg|_{y_{i}\neq 0, i\in\{1,...,d'\}}\leq (\sum_{k=1}^{d}\hat{x_k}+\sum_{j=1}^{t}\hat{y_j})\bigg|_{y_{d'}=0} \\ \leq ...\leq (\sum_{k=1}^{d}\hat{x_k}+\sum_{j=1}^{t}\hat{y_j})\bigg|_{y_{d'}=0,...,y_{t+1}=0}.$

When the number of suppressed buckets in $\overline{\mathcal{B}}$ increases, the corresponding interference of $\overline{\mathcal{B}}$ decreases. Therefore, the collected values more accurately converge to the buckets that they should belong to, and thus achieve a better convergence result.

After suppressing all buckets in $\overline{\mathcal{B}}$, all collected values will convergence to normal values and poison values in $B_{y_j}(j\in{1,...,t})$ and we can infer that$
(\sum_{k=1}^{d}\hat{x_k}+\sum_{j=1}^{t}\hat{y_j})\bigg|_{y_{d'}=0,...,y_{t+1}=0}=1,$ which is the optimal case where none of the collected values will converge to buckets in $\overline{\mathcal{B}}$.
\end{proof}

\subsection{Aggregating Inter-group Estimations}After intra-group mean estimations, we finally combine all these estimated means into a single mean. However, the naive method of averaging all of them in equal weights does not provide the optimal data utility --- values perturbed with larger $\epsilon_t$ have higher accuracy. Hence, the group they belong to deserves higher weight.\footnote{We ignore the influence of poison values, which have been addressed in the previous steps.} At the end of this section, inspired by~\cite{yiwen2018utility}, we propose an aggregation strategy that can linearly combine all estimated intra-group mean values $\{M_1,..., M_t,..., M_h\}$ into $\tilde{M}$, while achieving minimal overall variance. Since the variance is related to the true value that we don't know, we thus only consider the minimal variance under the worst-case, i.e., when all the original values are either -1 or 1.

\begin{algorithm}\small
\caption{Mean Aggregation}
\begin{flushleft}
 {\bf Input:} Mean values $\{M_1,..., M_t,..., M_h\}$, privacy budgets $\{\epsilon_1,...,\epsilon_t,...,\epsilon_h$\} and the estimated numbers of normal users $\{\hat{n_1},...,\hat{n_t},...,\hat{n_h}$\}\\
 {\bf Output:}
The aggregated mean $\tilde{M}$
\end{flushleft}
\begin{algorithmic}[1]\small
\State Initialization:  $w_t=0,\ t=\{1,2,...,h\}$
\For{$t=1\ to\ h$}
\State
$w_t=[B_t\sum_{i=1}^{h}\frac{1}{B_i}]^{-1}$, $B_t=\hat{n_t}(\frac{1}{e^{\epsilon_{t}/2}-1}+\frac{e^{\epsilon_{t}/2}+3}{3(e^{\epsilon_{t}/2}-1)^2})$
\EndFor
\\
$\tilde{M}=\sum_{t=1}^{h}w_tM_t$
\State \Return $\tilde{M}$
\end{algorithmic}
\label{Combination}
\end{algorithm}
Algorithm \ref{Combination} shows the detailed optimization procedure of such aggregation. Specifically, the estimated numbers of normal users in $G_t$ can be obtained from $\hat{n_t}=(N_t-\hat{m_t})\frac{\epsilon_t}{\epsilon}$. All weights $w_t$ of group $G_t$ are initially set to 0 (line 1). Then they are assigned by the formula in line 3, based on which all mean values are combined in line 4. This aggregation satisfies $\epsilon$-LDP according to the parallel composition theorem of LDP~\cite{li2016differential}. The following theorem guarantees the weight assignment in line 3 is optimal:
\begin{theorem}
\label{theoremofvar}
The variance of $\tilde{M}$ reaches the minimum, if the following formula holds:
\begin{equation*}
\setlength{\abovedisplayskip}{2pt}
\setlength{\belowdisplayskip}{2pt}
w_t=\frac{1}{B_t\sum_{i=1}^{h}\frac{1}{B_i}},
\end{equation*}
where $B_t=\hat{n_t}(\frac{1}{e^{\epsilon_t/2}-1}+\frac{e^{\epsilon_t/2}+3}{3(e^{\epsilon_t/2}-1)^2})$.

The minimal variance is:
{\setlength{\abovedisplayskip}{1pt}
\setlength{\belowdisplayskip}{1pt}
\begin{equation*}
Var(\tilde{M})_{min}=[\sum_{t=1}^{h}\frac{\hat{n_t}^2}{B_t}]^{-1}.
\end{equation*}}
\end{theorem}

\begin{proof}
Let $v_{tj}$ denote the $j$-th value in group $G_t$, $v'_{tj}$ denote the perturbed $v_{tj}$, and $M_t$ denote the mean value of $v'_{tj}$. The variance of $\tilde{M}$, which is a linear combination of $M_t$, can be written as:
\begin{equation}
\setlength{\abovedisplayskip}{2pt}
\setlength{\belowdisplayskip}{2pt}
\begin{split}
Var(\tilde{M})&=Var(\sum_{t=1}^{h}w_{t}M_t)=\sum_{t=1}^{h}w_{t}^{2}Var(M_t)
\\&=\sum_{t=1}^{h}w_{t}^{2}Var(\frac{\sum_{j=1}^{\hat{n_t}}v'_{tj}}{\hat{n_t}})
        =\sum_{t=1}^{h}\frac{w_{t}^{2}}{\hat{n_t}^{2}}\sum_{j=1}^{\hat{n_t}}Var(v'_{tj})
\end{split}
\label{variance of M}
\end{equation}
where $\sum{w_t}=1$.

Since $Var(v'_{tj})$ in Equ. \ref{variance of M} relies on the input of each user, we consider the worst-case at the maximum variance, i.e., all inputs $v_{tj}$ are either 1 or -1. The worst-case variance $Var_{worst}(v'_{tj})$ can be expressed as:
\begin{equation*}
\setlength{\abovedisplayskip}{2pt}
\setlength{\belowdisplayskip}{2pt}
\begin{split}
Var_{worst}(v'_{tj})&=\frac{v^2_{tj}}{e^{\epsilon_t/2}-1}+\frac{e^{\epsilon_t/2}+3}{3(e^{\epsilon_t/2}-1)^2}\bigg|_{v_{tj}=\pm 1}\\
    &=\frac{1}{e^{\epsilon_t/2}-1}+\frac{e^{\epsilon_t/2}+3}{3(e^{\epsilon_t/2}-1)^2}.
  \end{split}
\end{equation*}
Let $B_t=\hat{n_t}Var_{worst}(v'_{tj})$. Equ. \ref{variance of M} can be rewritten as:
 \begin{equation}
 \label{target3}
\begin{split}
\setlength{\abovedisplayskip}{0pt}
\setlength{\belowdisplayskip}{0pt}
  Var(\tilde{M})=\sum_{t=1}^{h}\frac{w_{i}^{2}}{\hat{n_t}^2}B_t.
  \end{split}
\end{equation}
We regard the variance as a function of $w_t$, and the minimal variance is the extreme point of Equ. \ref{target3}. By the Lagrangian method, we have:
 \begin{equation*}
 \setlength{\abovedisplayskip}{2pt}
\setlength{\belowdisplayskip}{2pt}
\mathcal{L}=\sum_{t=1}^{h}\frac{w_t^{2}}{\hat{n_t}^2}B_t+C_0(1-\sum_{t=1}^{h}w_t).
\end{equation*}
The first partial derivatives of $\mathcal{L}$ w.r.t. $w_t$ is $\setlength{\abovedisplayskip}{2pt}
\setlength{\belowdisplayskip}{2pt}
\frac{\partial \mathcal{L}}{\partial w_t}=\frac{2w_t}{\hat{n_t}^2}B_t-C_0.$
Let $\frac{\partial \mathcal{L}}{\partial w_t}=0$, then we have $wt=\frac{C_0\hat{n_t}^2}{2B_t}$. Through the restriction $\sum_{t=1}^{h}w_t=1$, we figure out
\begin{equation*}
\setlength{\abovedisplayskip}{2pt}
\setlength{\belowdisplayskip}{2pt}
C_0=\frac{2}{\hat{n_t}^2\sum_{t=1}^{h}\frac{1}{B_t}},\ w_t=\frac{1}{B_t\sum_{i=1}^{h}\frac{1}{B_i}}.
\end{equation*}
And the final minimal variance of $\tilde{M}$ is:
\begin{equation*}
\setlength{\abovedisplayskip}{2pt}
\setlength{\belowdisplayskip}{2pt}
{Var(\tilde{M})}_{min}=[\sum_{t=1}^{h}\frac{\hat{n_t}^2}{B_t}]^{-1}.
\end{equation*}
 \label{proof of cutoffpoint}
\end{proof}

\subsection{Discussion}
\color{black}
\textbf{Extension to Other Perturbation Mechanisms and Statistics.} DAP can also work with other LDP perturbation mechanisms such as Square Wave (SW)~\cite{li2020estimating}. SW can estimate the distribution and thus the mean by the Expectation Maximization with Smoothing (EMS) protocol. The original value range of SW is $[0, 1]$, and the perturbed data range is $[-b, 1+b]$. To incorporate SW into DAP, we can design the matrix $M$ for SW similarly to that for PM in Section \ref{EMF}, where the key is to change the size and the transformation probability in $M$ according to the SW mechanism. In addition, similar to Theorem~\ref{theo of cutoffpoint}, $O'$ for SW can be obtained by estimating the mean with EMS upon the removal of 50\% values from $V'$. Then we can estimate the poison value distribution with EMF (resp. EMF*, CEMF*), remove these poison values, and estimate the mean of the rest with EMS. Finally, the same mean aggregation algorithm (Algorithm \ref{Combination}) can still be applied for mean aggregation from different groups.

Furthermore, DAP is not limited to mean estimation and can be generalized to other statistics estimation, e.g., frequency estimation. This is because $\hat{x}$ derived in Algorithms \ref{the algorithm of EMF} and \ref{the algorithm of REMF} is essentially the frequency histogram of numerical/categorical data. To exemplify how to use DAP to estimate frequency of categorical data, let us consider k-RR~\cite{kairouz2014extremal, wang2016private}, a common LDP mechanism for categorical data where $O'$ is not well-defined. We can still divide the categorical values into two ``sides" to determine which categorical value is poisoned, whose approach is similar to probing the poisoned side on numerical data by running Algorithm \ref{Poisoned side2}. If more than one categorical value is poisoned, we just need to further divide sides and run Algorithm \ref{Poisoned side2} recursively and locate those categorical values which result in smaller variance of $\hat{x}$ than others. Once the poisoned categorical value is determined, the remaining steps for frequency estimation in DAP is the same as mean estimation on numerical data.

\textbf{Robustness of Evasion.} Let us assume that Byzantine users are aware of our proposed DAP technique and try to evade it by injecting a
fraction of poison values onto the opposite side of the poisoned side. Without loss of generality, a total of $m$ Byzantine users inject $(1-a)m$ {\bf true poison values} on $[O', C]$ and $am$ {\bf evasive poison values} on $[-C, O']$, where $a$ is the percentage of evasive poison values. When $a$ is small, we can expect DAP to ignore the evasive poison values by choosing the correct poisoned side. When $a$ becomes large, DAP may be affected by these evasive poison values by choosing the wrong side, but such strong evasion also weakens the utility of poison values. Formally, the utility lower bound of poison values $U_{max}$ is given by the case where all true poison values are at $C$ and $a=0$, but none of them is removed:
{\setlength{\abovedisplayskip}{1pt}
\setlength{\belowdisplayskip}{1pt}
\begin{equation}
U_{max}=\frac{mC+nO}{m+n}-O
\label{Umax}
\end{equation}}
On the other hand, the utility lower bound is also given by the case when all evasive values are slightly less than $O'$: 
{\setlength{\abovedisplayskip}{2pt}
\setlength{\belowdisplayskip}{2pt}
\begin{equation}
U_{eva}=\frac{m(aO'+(1-a)C)+nO}{m+n}-O
\label{Ueva}
\end{equation}}
Since both cases must be satisfied simultaneously, the lowest utility loss for Byzantine users is
{\setlength{\abovedisplayskip}{1pt}
\setlength{\belowdisplayskip}{1pt}
\begin{equation}
U_{max}-U_{eva}=\frac{ma(C-O')}{m+n}
\label{Udelta}
\end{equation}}
From Equ. \ref{Udelta}, while there is a minimum $a$ to satisfy for evasive poison values to make DAP choose the wrong poisoned side, increasing $a$ only weakens the utility of poison values. We also conduct some experiments that coincide with this analysis in our Section \ref{Exp}.
\color{black}

\section{Experimental Results}
\label{Exp}
\begin{figure*}[th]
\centering
\subfigure[\textbf{Beta(2,5)}, O=-0.3994.]{
\begin{minipage}[t]{0.21\linewidth}
\centering
\includegraphics[width=1\textwidth]{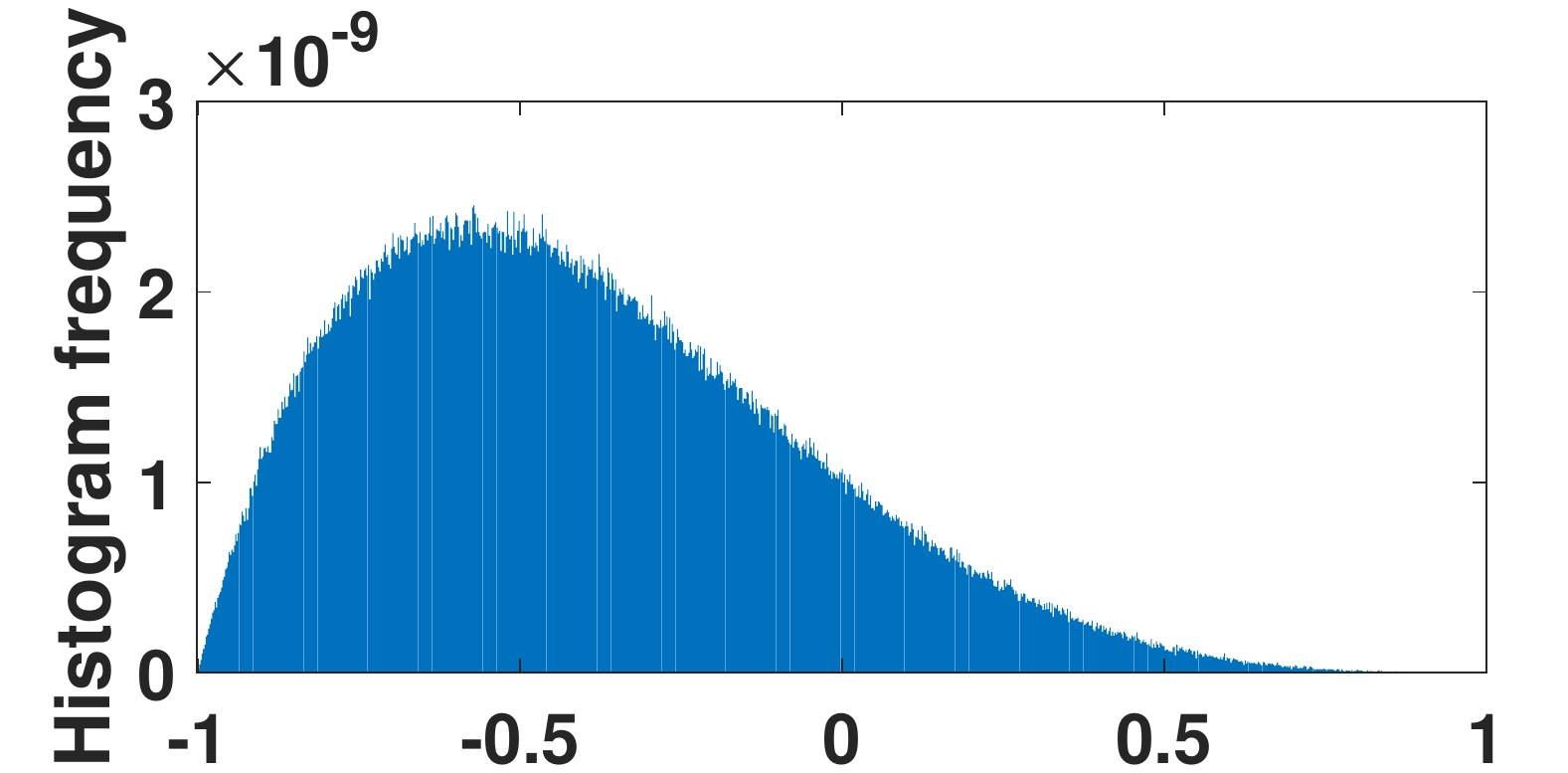}
\end{minipage}
}%
\subfigure[\textbf{Beta(5,2)}, O=0.4136.]{
\begin{minipage}[t]{0.21\linewidth}
\centering
\includegraphics[width=1\textwidth]{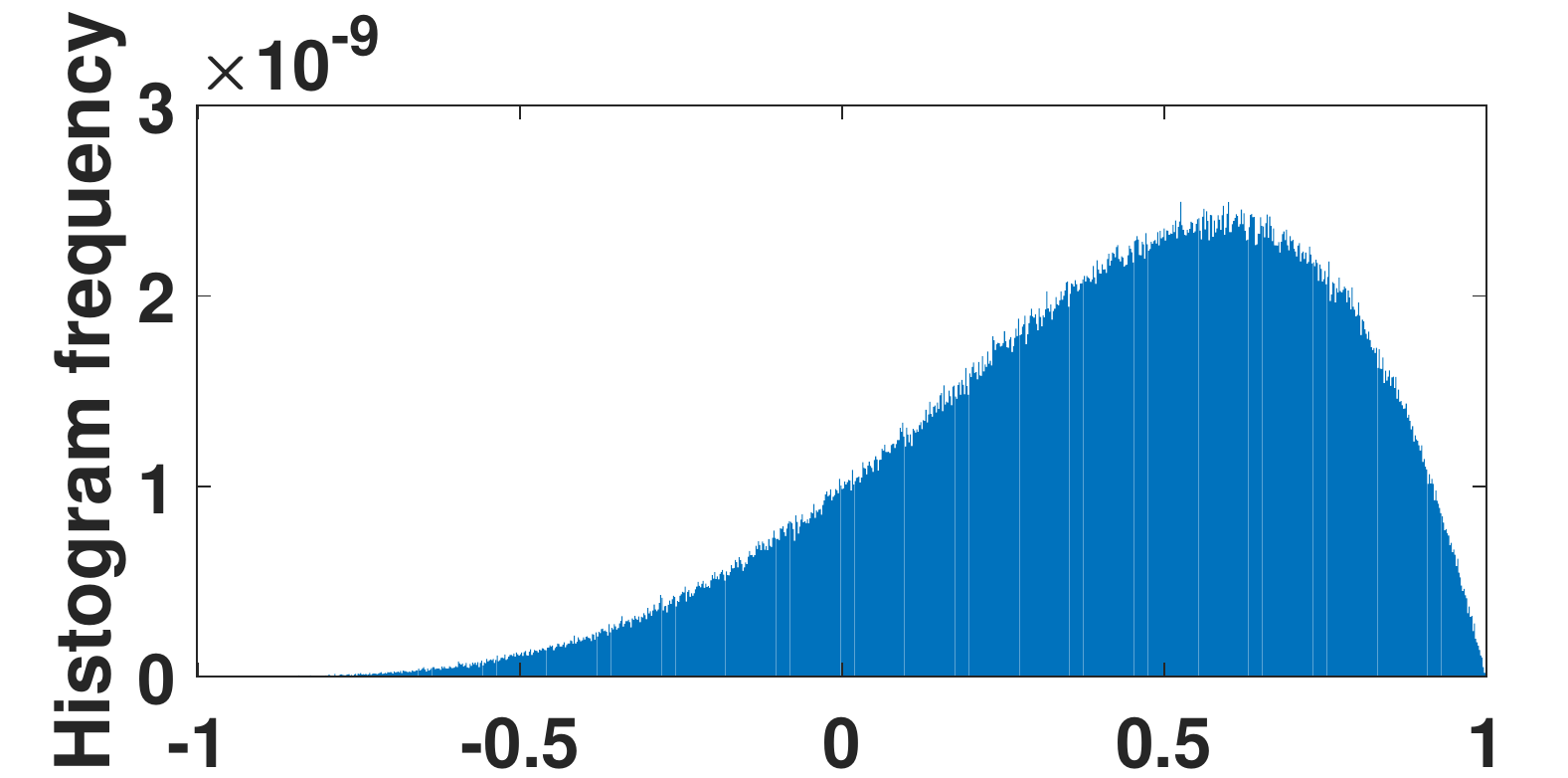}
\end{minipage}%
}%
\subfigure[\textbf{Taxi}, O=0.1190.]{
\begin{minipage}[t]{0.21\linewidth}
\centering
\includegraphics[width=1\textwidth]{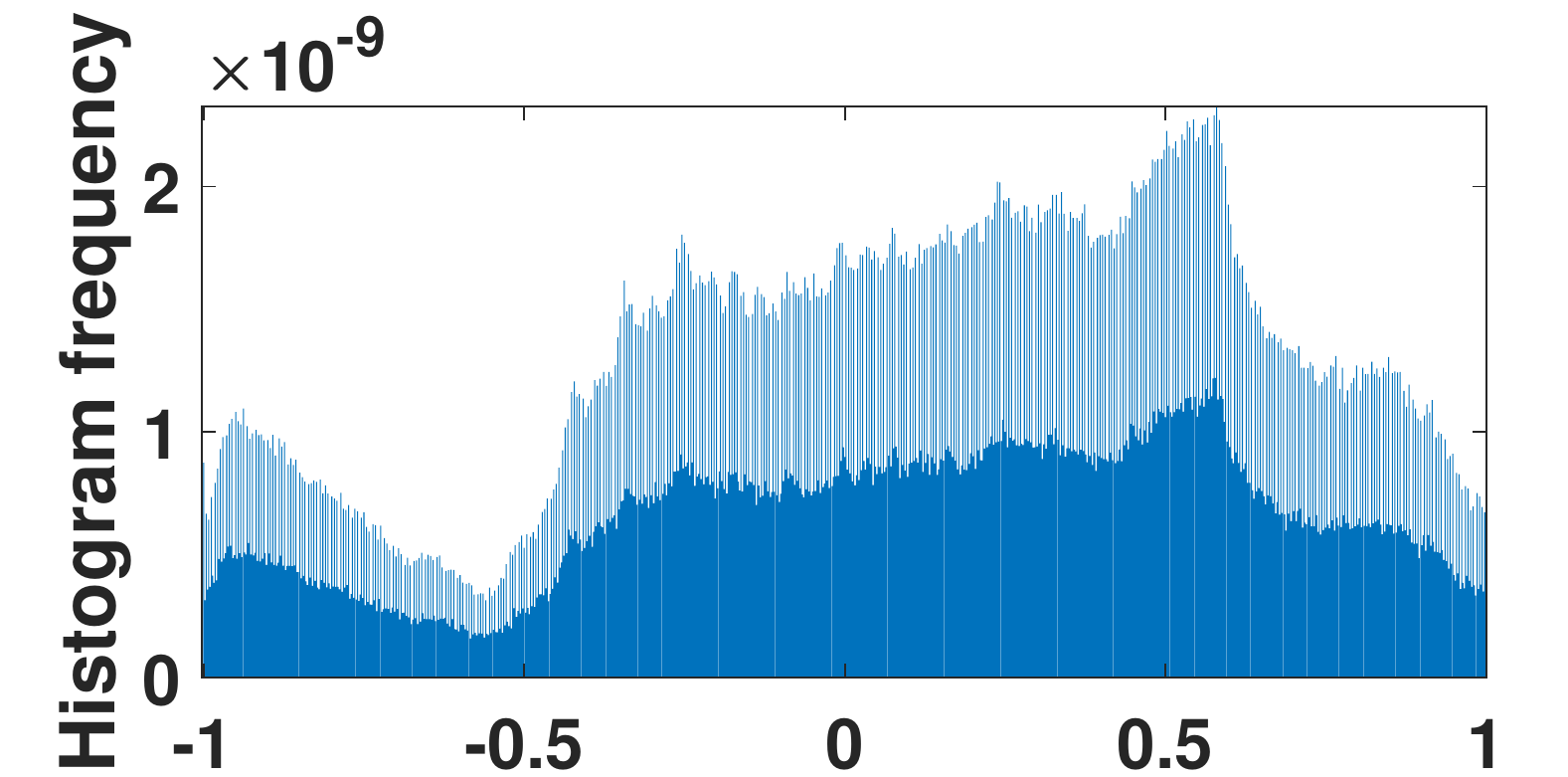}
\end{minipage}
}%
\subfigure[\textbf{Retirement}, O=-0.6240.]{
\begin{minipage}[t]{0.21\linewidth}
\centering
\includegraphics[width=1\textwidth]{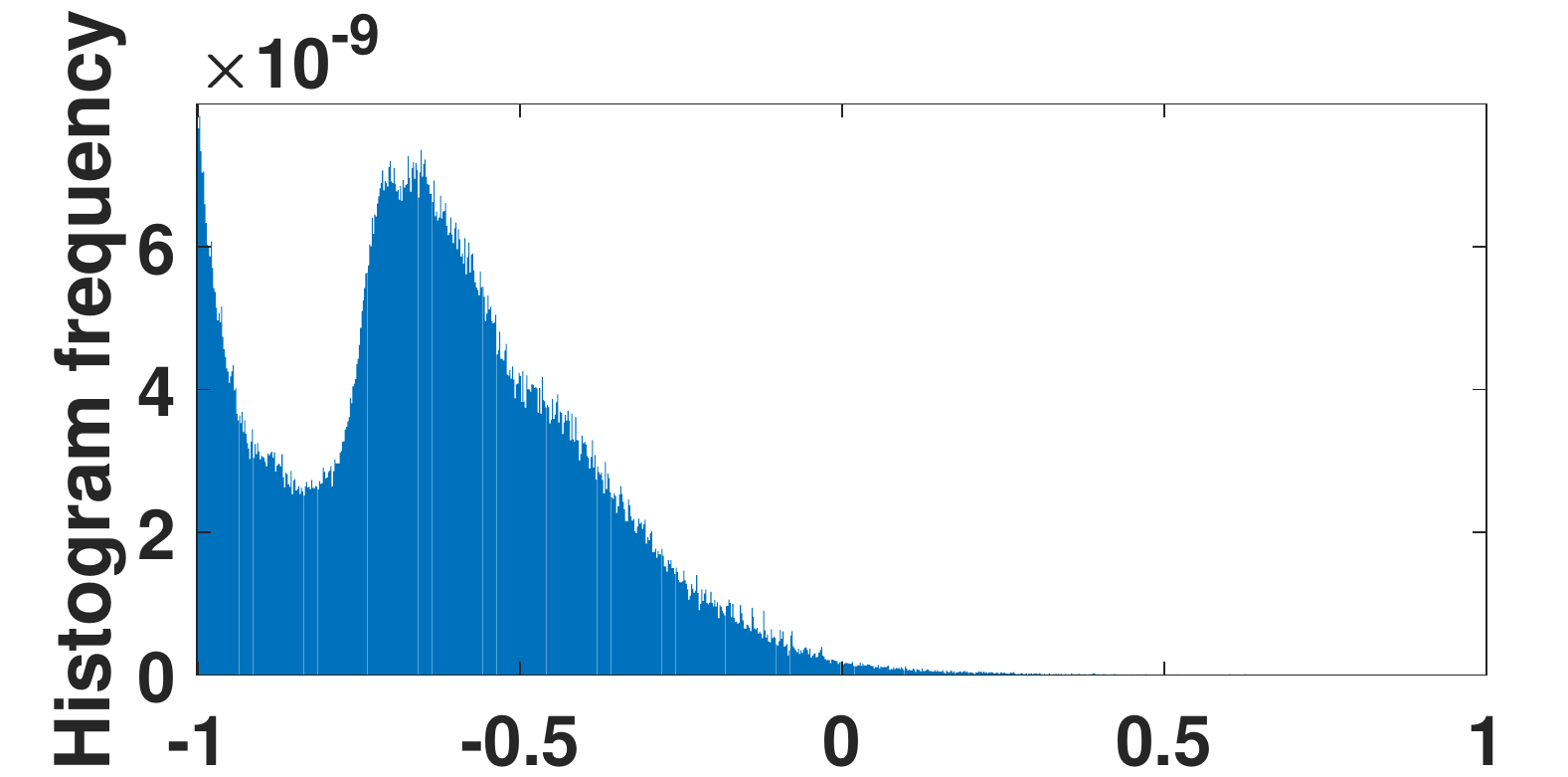}
\end{minipage}
}%
\centering
\vspace{-0.1in}
\caption{Normalized frequencies of datasets}
\label{The distribution of datasets}
\end{figure*}
In this section, we evaluate the performance of both EMF and DAP, on both real-world and synthetic datasets. \textcolor{black}{The source code and datasets are available in~\cite{CODEFORDAP}.}
\subsection{Experiment Setup}
\textbf{Datasets.} We adopt two synthetic and two real-world numerical datasets. $\textbf{Beta(2,5)}$ and $\textbf{Beta(5,2)}$ are two synthetic datasets drawn from Beta distribution~\cite{johnson1995continuous}, each with 1,000,000 samples in the interval $[0,1]$. \textbf{Taxi}~\cite{taxi2018} is the pick-up time in a day extracted from 2018 January New York Taxi data, which contains 1,048,575 integers from 0 to 86,340 (the number of seconds in 24 hours). \textbf{Retirement}~\cite{retirement2013} is extracted from the San Francisco employee retirement plans, which contains the salary and benefits paid to city employees since fiscal year 2013. We employ the total compensation, which comprises a subset of 606,507 items in the interval $[10000, 60000]$. All these datasets are then normalized into interval $[-1,1]$. The normalized frequency histograms of all datasets are plotted in Fig. \ref{The distribution of datasets}.

\textbf{Parameters Setting.} We use PM as our default perturbation mechanism, and set the right side as the default poisoned side. Therefore, the poison values are injected into poison values' range $[r_l, r_r] \subset [O, C]$, denoted by $Poi_{[r, l]}$ in the figures. To evaluate the scalability of our protocol, we apply our schemes with different parameters, including poison values' ranges and poison value distributions. We set a series of $r_l<r_r$ for different poison values' ranges. When the poisoned side is estimated as the right-hand one by Algorithm~\ref{Poisoned side2}, we set the probing range to $[O', C]$  (resp. $[-C, O']$ if the side is to the left).
The proportion of Byzantine users is set to 25\% and the poison values are uniformly distributed in $[r_l, r_r]$ by default. The termination condition for EMF/EMF*/CEMF* is $|l(F)^t-l(F)^{t+1}|<\tau$, where $\tau$ is set to $0.01e^\epsilon$. We choose $d'=\lfloor \sqrt{N}\rfloor$ and $d=\lfloor\frac{d'(e^{\epsilon/2}-1)}{e^{\epsilon/2}+1}\rfloor$.

\subsection{Effectiveness of EMF}
\begin{table}[h]
\scriptsize
        \caption{Variance of reconstructed normal data}
	\label{tab1}
\centering
\begin{tabular}{|c|c|ccccc|}
\hline
\multirow{2}{*}{$Poi_{[l, r]}$}     & \multirow{2}{*}{Side} & \multicolumn{5}{c|}{Privacy budget   $\epsilon$}                                                                                                                                      \\ \cline{3-7}
                                    &                       & \multicolumn{1}{c|}{2}                 & \multicolumn{1}{c|}{0.5}               & \multicolumn{1}{c|}{0.25}              & \multicolumn{1}{c|}{0.125}             & 0.0625            \\ \hline
\multirow{2}{*}{$[\frac{3}{4}C,C]$} & L                     & \multicolumn{1}{c|}{3.8E-03}          & \multicolumn{1}{c|}{1.4E-03}          & \multicolumn{1}{c|}{2.7E-03}          & \multicolumn{1}{c|}{2.7E-03}          & 1.5E-03          \\ \cline{2-7}
                                    & R                     & \multicolumn{1}{c|}{\textbf{2.0E-07}} & \multicolumn{1}{c|}{\textbf{4.9E-06}} & \multicolumn{1}{c|}{\textbf{1.1E-05}} & \multicolumn{1}{c|}{\textbf{1.6E-05}} & \textbf{1.9E-05} \\ \hline
\multirow{2}{*}{$[\frac{1}{2}C,C]$} & L                     & \multicolumn{1}{c|}{1.4E-04}          & \multicolumn{1}{c|}{5.7E-04}          & \multicolumn{1}{c|}{1.1E-03}          & \multicolumn{1}{c|}{1.0E-03}          & 7.0E-04          \\ \cline{2-7}
                                    & R                     & \multicolumn{1}{c|}{\textbf{4.1E-07}} & \multicolumn{1}{c|}{\textbf{4.2E-06}} & \multicolumn{1}{c|}{\textbf{9.0E-06}} & \multicolumn{1}{c|}{\textbf{1.3E-05}} & \textbf{1.4E-05} \\ \hline
\multirow{2}{*}{$[O,\frac{1}{2}C]$} & L                     & \multicolumn{1}{c|}{7.6E-06}          & \multicolumn{1}{c|}{1.7E-05}          & \multicolumn{1}{c|}{3.4E-05}          & \multicolumn{1}{c|}{4.3E-05}          & 4.1E-05          \\ \cline{2-7}
                                    & R                     & \multicolumn{1}{c|}{\textbf{3.0E-07}} & \multicolumn{1}{c|}{\textbf{3.9E-06}} & \multicolumn{1}{c|}{\textbf{8.4E-06}} & \multicolumn{1}{c|}{\textbf{1.2E-05}} & \textbf{1.3E-05} \\ \hline
\multirow{2}{*}{$[O, C]$} & L                     & \multicolumn{1}{c|}{1.3E-05}          & \multicolumn{1}{c|}{1.8E-04}          & \multicolumn{1}{c|}{4.3E-04}          & \multicolumn{1}{c|}{5.0E-04}          & 4.4E-04          \\ \cline{2-7}
                                    & R                     & \multicolumn{1}{c|}{\textbf{2.7E-07}} & \multicolumn{1}{c|}{\textbf{3.1E-06}} & \multicolumn{1}{c|}{\textbf{5.8E-06}} & \multicolumn{1}{c|}{\textbf{7.4E-06}} & \textbf{7.9E-06} \\ \hline
\end{tabular}
\label{Poisoned side}
\vspace{-3mm}
\end{table}

\textbf{Poisoned Side Estimation.} We conduct experiments on both sides to evaluate the effectiveness of Algorithm \ref{Poisoned side2} that determines the poisoned side. We can see from results on $\textbf{Taxi}$  shown in Table \ref{Poisoned side} that EMF on the right side always has a smaller variance under different privacy budgets and poison values' ranges. Based on this, Algorithm \ref{Poisoned side2} can determine the correct poisoned side, which is the right side with smaller variance. The correct estimation is attributed to two reasons. On one hand, $\hat{x}$ of the correct side can have a lower variance than the incorrect side, which is proved in Theorem \ref{lemma of poisoned direction and the proportion of Byzantine users}. On the other hand, even if the wrong side is initially chosen, all poison values will be assigned to the non-poisoned side and cause a surge in the right of the distribution for normal users, which will also lead to a larger variance.

\textbf{Proportion of Byzantine Users.} In the first set of experiments, we verify the accuracy of the proportion of Byzantine users obtained by EMF w.r.t. $\epsilon$ in Fig. \ref{fig of EMF}.

Given a fixed privacy budget, we compare the four poison values' ranges $[3C/4,C],[C/2, C],[O, C/2]$ and $[O,C]$ in Fig. \ref{fig of EMF} (a) (b). When $\epsilon$ is large, the estimated proportion deviates from the real proportion of Byzantine users $\gamma$. As $\epsilon$ decreases, the result becomes more accurate. The reason is that according to Theorem \ref{lemma of poisoned direction and the proportion of Byzantine users}, when $\epsilon$ becomes smaller, EMF distinguishes normal and poison values better. In these figures, regardless of poison values' ranges, datasets and $\gamma$, $|\hat{\gamma}-\gamma|$ converges to 0 as $\epsilon\rightarrow 0$. This shows the robustness of EMF even under adverse conditions where many normal values are mistreated as poison values.

\color{black}Fig. \ref{fig of EMF} (c) evaluates $\hat{\gamma}$ when there is no poison value in the system, and the proportion of Byzantine users estimated by EMF is equivalent to the false positive rate ($fpr$). When $\epsilon_0$ gets smaller, $fpr$ becomes less relevant to normal users' distribution. According to Theorem 3, the smaller $\epsilon_0$ is, the more accurate the estimated proportion of Byzantine users becomes. We observe for all 4 datasets and $\epsilon_0=\frac{1}{16}$, and the false positives range from 0.02 to 0.04, which are quite small. In Fig. \ref{fig of EMF} (d), we evaluate $\hat{\gamma}$ when Byzantine users inject their poison values as 1 and then perturb as normal users, which is also known as an input manipulation attack (IMA)~\cite{li2022fine}. We observe that $\hat{\gamma}$ is between $0.03\sim 0.04$, which means EMF cannot filter out these Byzantine users due to the perturbation. Nonetheless, we can further enhance the utility with existing detection techniques (such as k-means clustering~\cite{li2022fine}), which will be shown in Fig.\ref{discussion} (b).
\color{black}

\begin{figure*}[htbp]
 \vspace{-0.15in}
\hspace{-2.3in}
  {
  \begin{minipage}{3cm}
   \centering
   \includegraphics[scale=0.38]{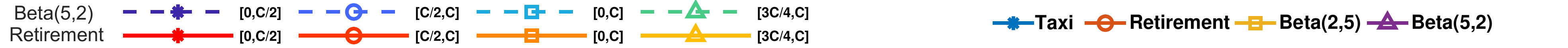}
  \end{minipage}
 }
 \\
 \vspace{-0.1in}
\centering
\subfigure[\textcolor{black}{$\gamma=0.1$.}]{
\begin{minipage}[t]{0.21\linewidth}
\centering
\includegraphics[width=1\textwidth]{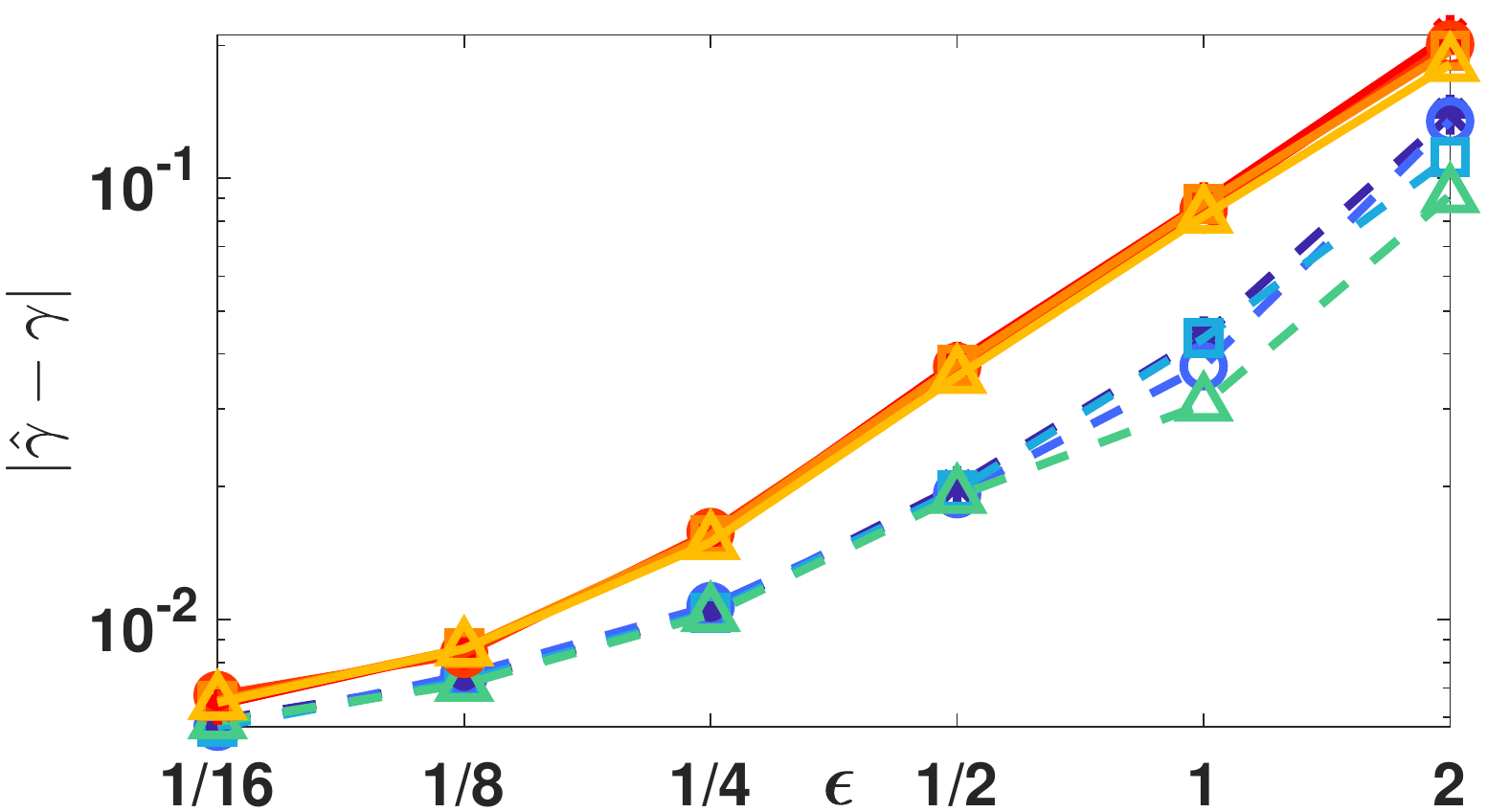}
\end{minipage}%
}%
\subfigure[\textcolor{black}{$\gamma=0.4$.}]{
\begin{minipage}[t]{0.21\linewidth}
\centering
\includegraphics[width=1\textwidth]{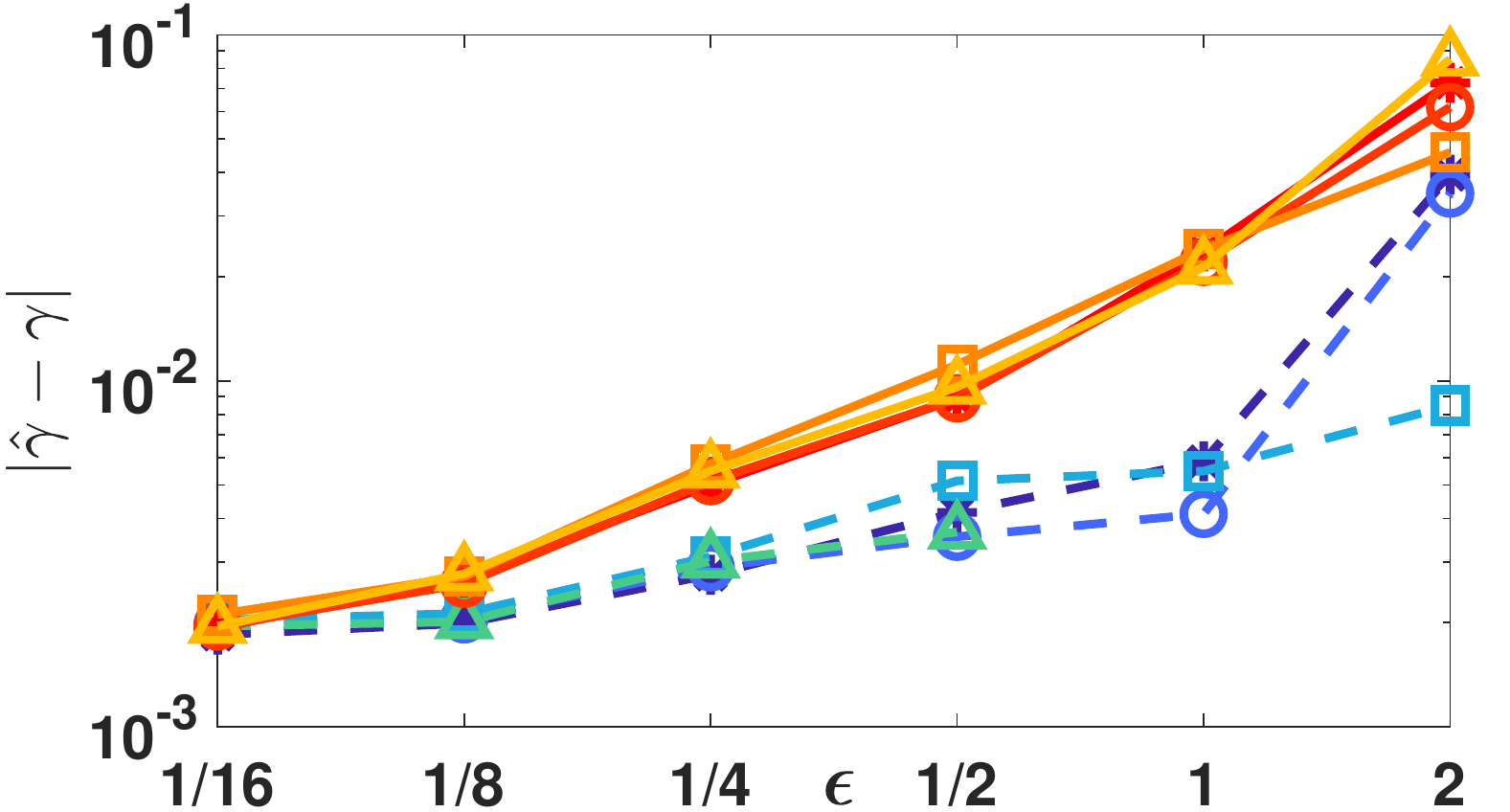}
\end{minipage}
}%
 \vspace{0.1in}
\centering
\subfigure[\textcolor{black}{$\gamma=0$.}]{
\begin{minipage}[t]{0.21\linewidth}
\centering
\includegraphics[width=1\textwidth]{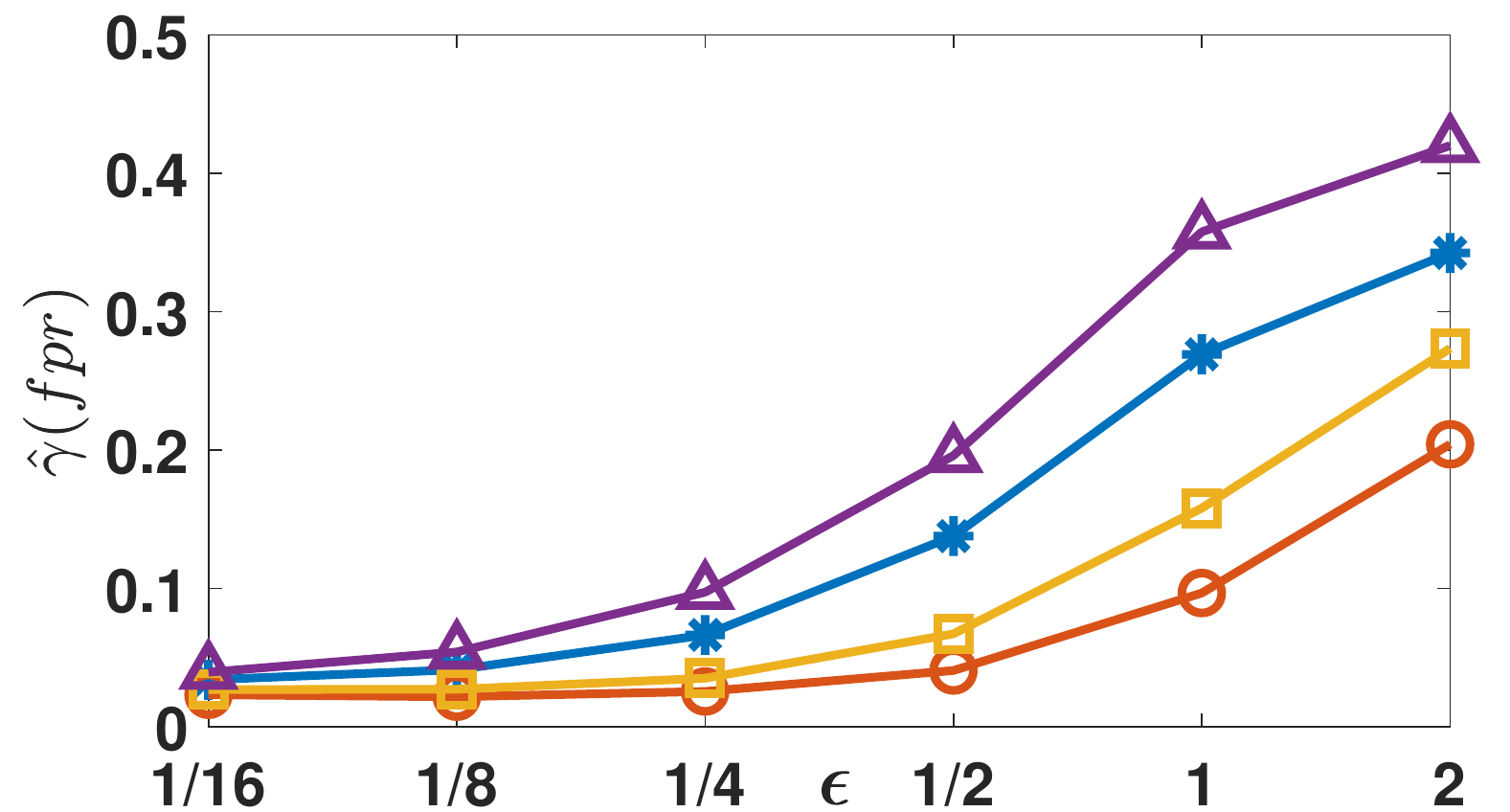}
\end{minipage}
}%
\subfigure[\textcolor{black}{IMA, $\gamma=0.25$.}]{
\begin{minipage}[t]{0.21\linewidth}
\centering
\includegraphics[width=1\textwidth]{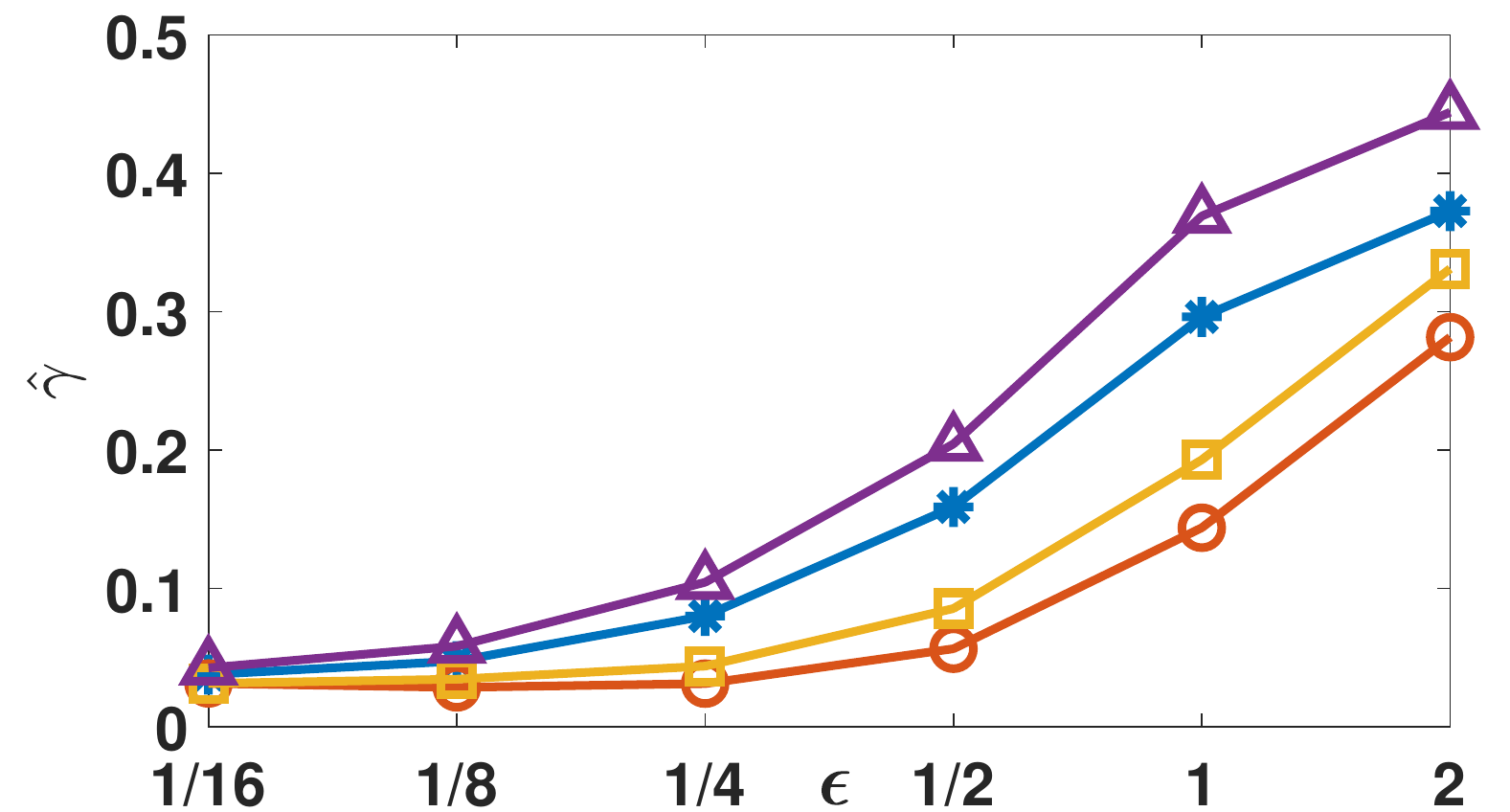}
\end{minipage}
}
\centering
 \vspace{-0.15in}
\caption{\textcolor{black}{The proportion of Byzantine users estimated by EMF w.r.t. $\epsilon$}}
\label{fig of EMF}
\end{figure*}

\begin{figure*}[ht]
 \vspace{-0.15in}
 \hspace{-1.25in}
  {
  \begin{minipage}{3cm}
   \centering
   \includegraphics[scale=0.5]{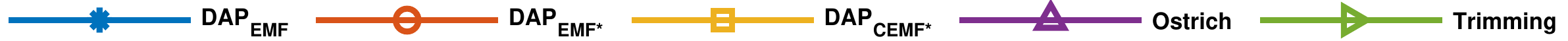}
  \end{minipage}
   \vspace{-0.05in}
 }
 \\
 \vspace{-0.1in}
\centering
\subfigure[\textbf{Beta(2,5)}, $Poi_{[3C/4,C]}$.]{
\begin{minipage}[t]{0.215\linewidth}
\centering
\includegraphics[width=1\textwidth]{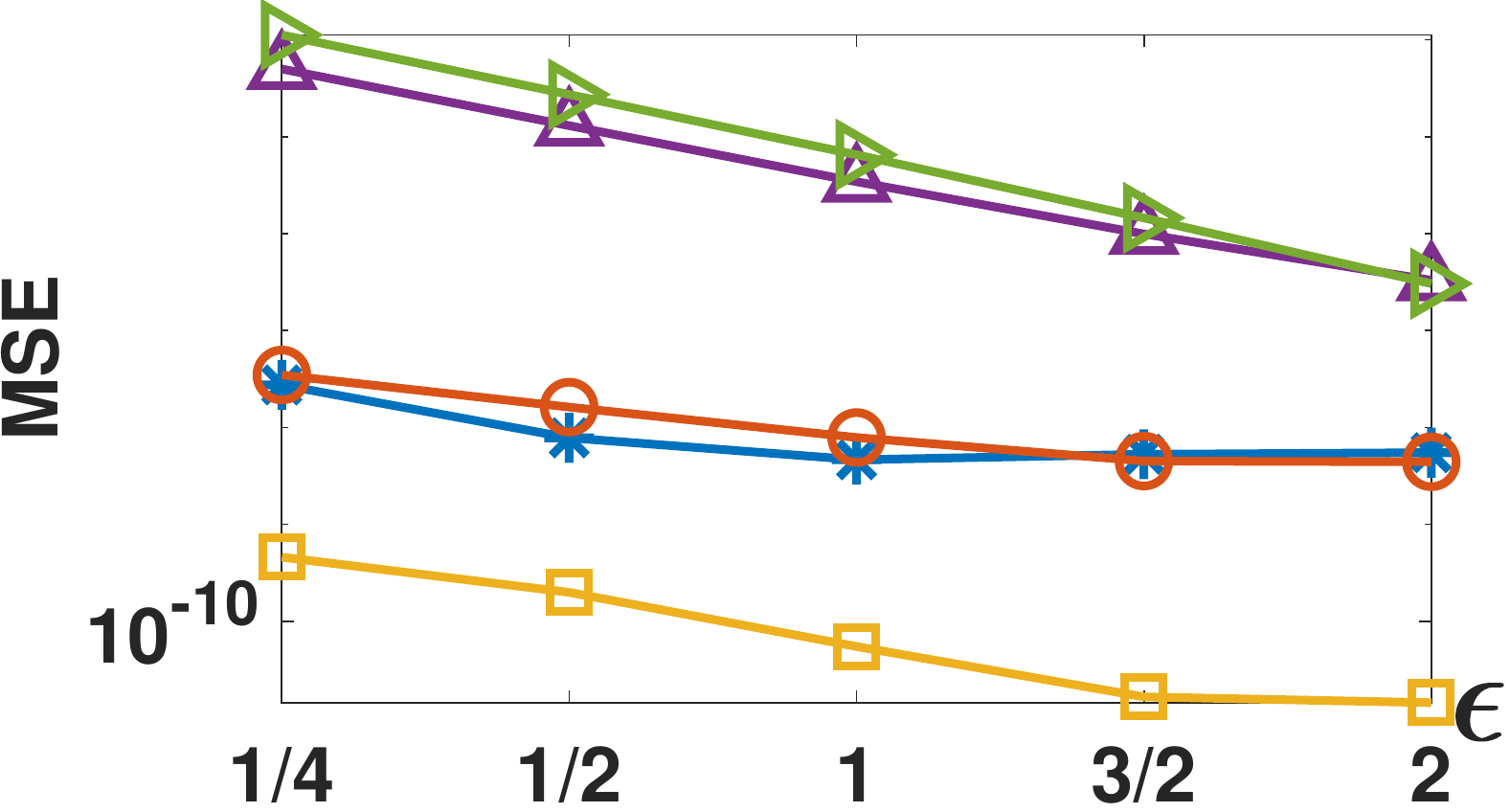}
\end{minipage}%
}%
\subfigure[\textbf{Beta(5,2)}, $Poi_{[3C/4,C]}$.]{
\begin{minipage}[t]{0.215\linewidth}
\centering
\includegraphics[width=1\textwidth]{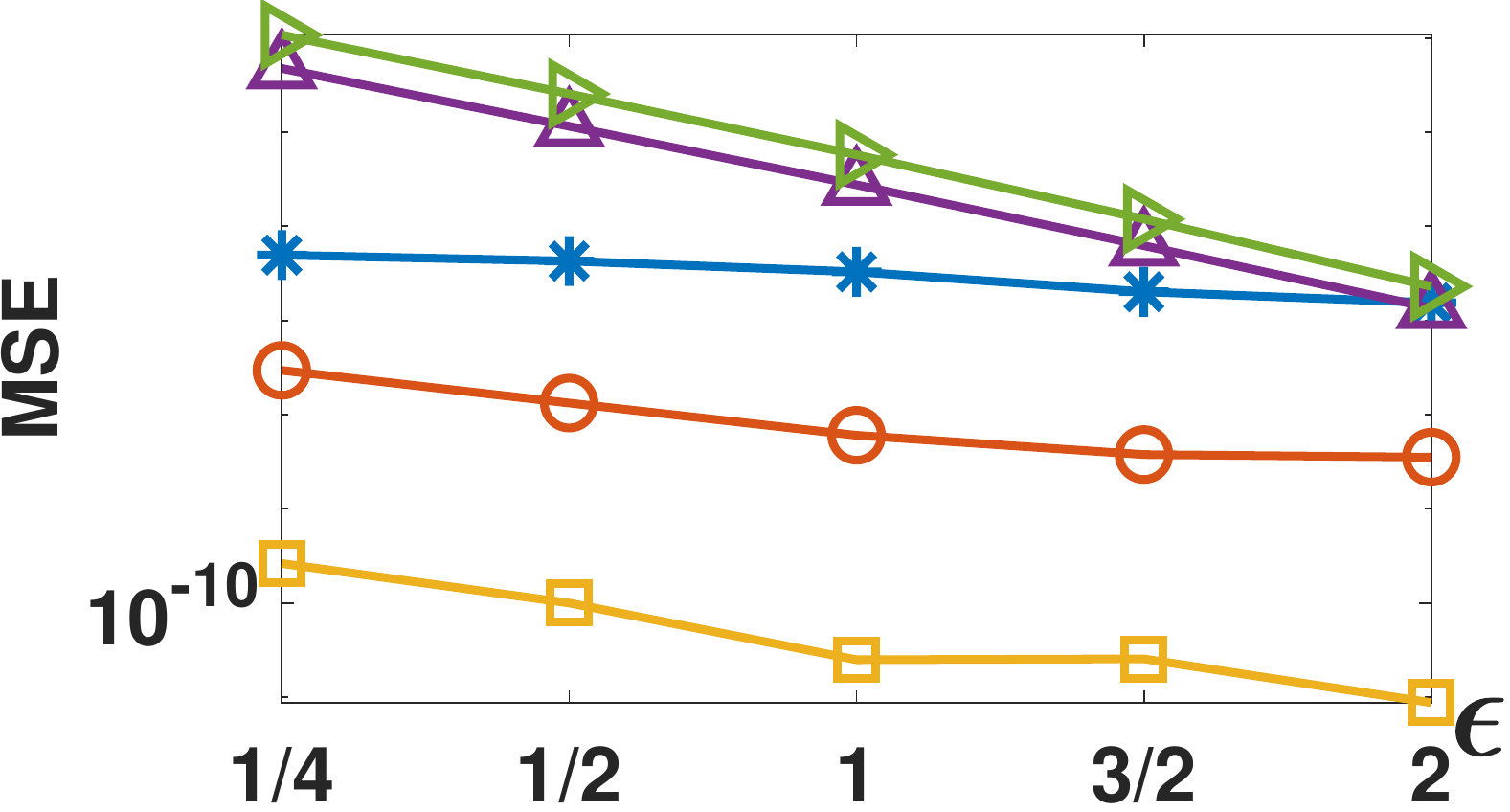}
\end{minipage}%
}%
\subfigure[\textbf{Taxi}, $Poi_{[3C/4,C]}$. ]{
\begin{minipage}[t]{0.215\linewidth}
\centering
\includegraphics[width=1\textwidth]{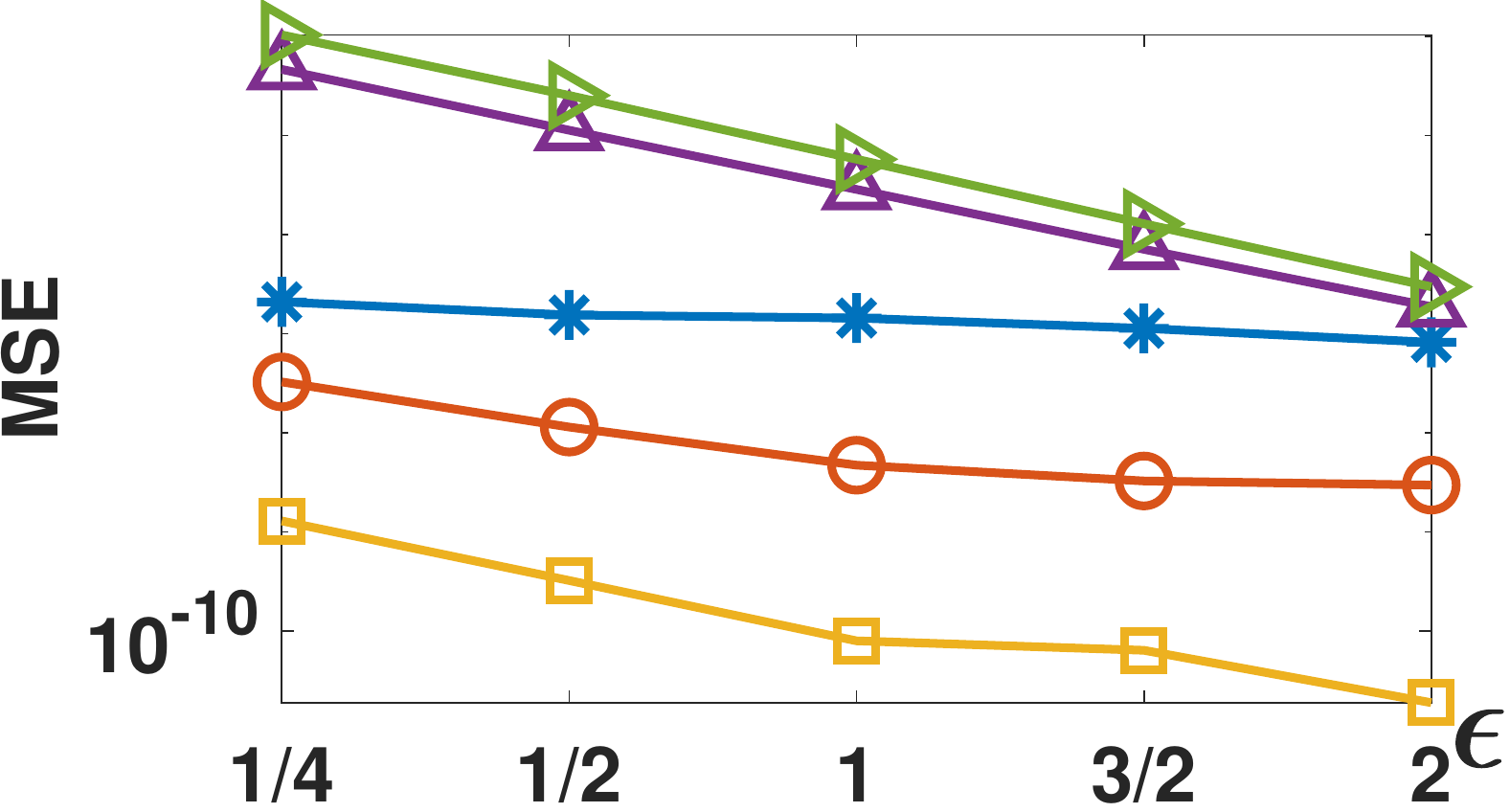}
\end{minipage}
}%
\centering
\subfigure[\textbf{Retirement}, $Poi_{[3C/4,C]}$.]{
\begin{minipage}[t]{0.215\linewidth}
\centering
\includegraphics[width=1\textwidth]{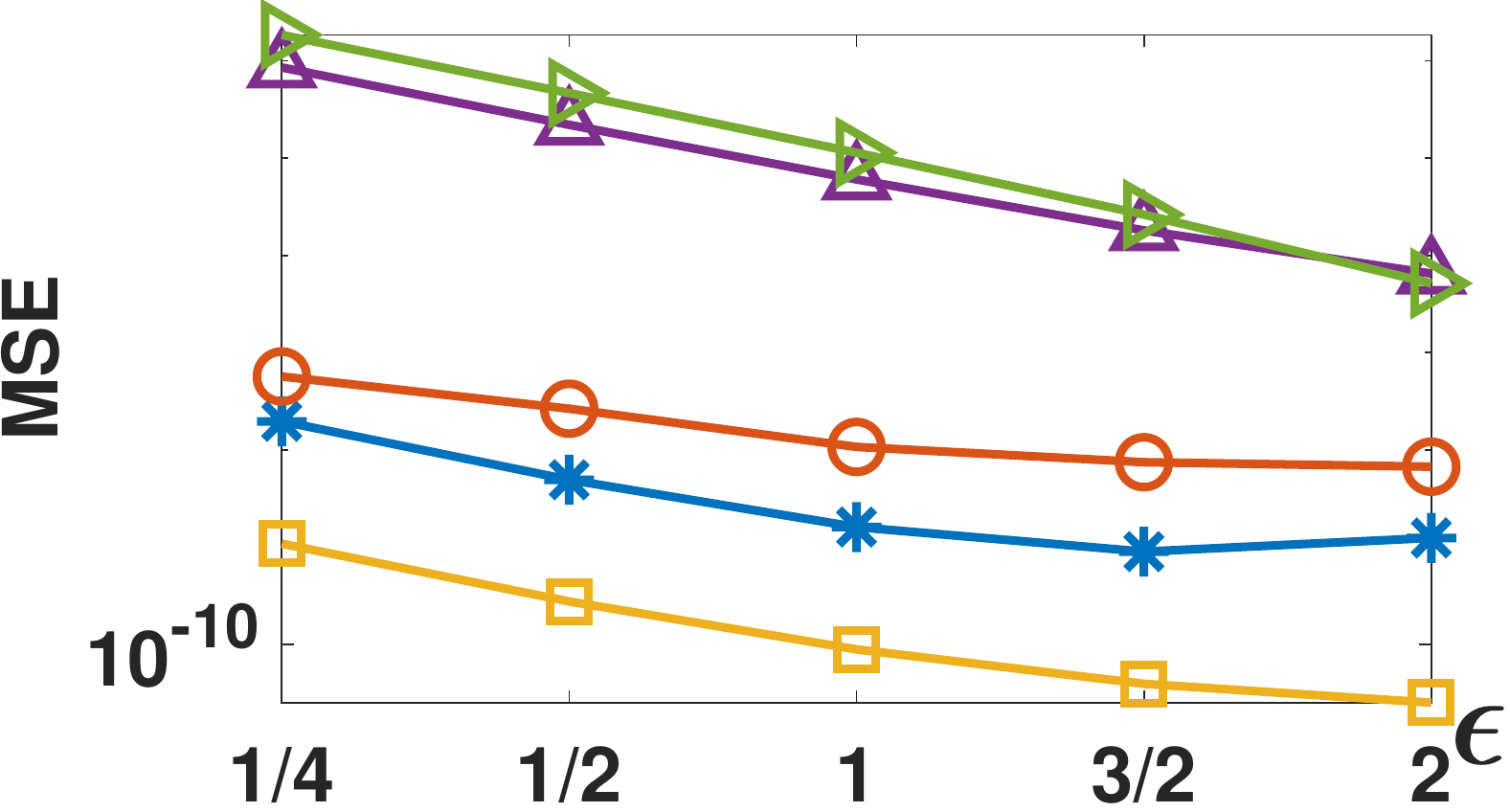}
\end{minipage}
}%

\vspace{-0.05in}
\centering
\subfigure[\textbf{Beta(2,5)}, $Poi_{[1C/2,C]}$.]{
\begin{minipage}[t]{0.215\linewidth}
\centering
\includegraphics[width=1\textwidth]{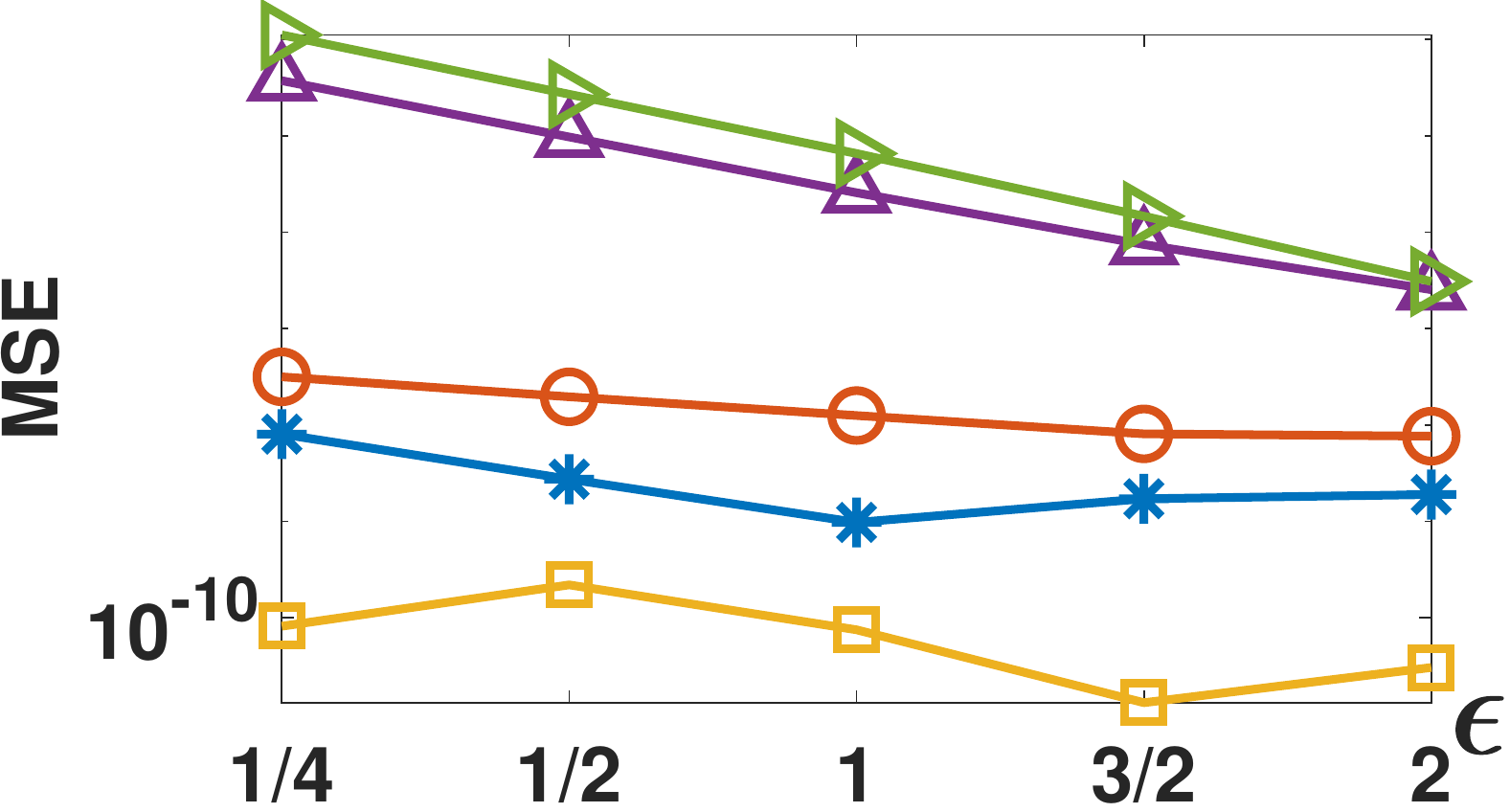}
\end{minipage}%
}%
\subfigure[\textbf{Beta(5,2)}, $Poi_{[1C/2,C]}$.]{
\begin{minipage}[t]{0.215\linewidth}
\centering
\includegraphics[width=1\textwidth]{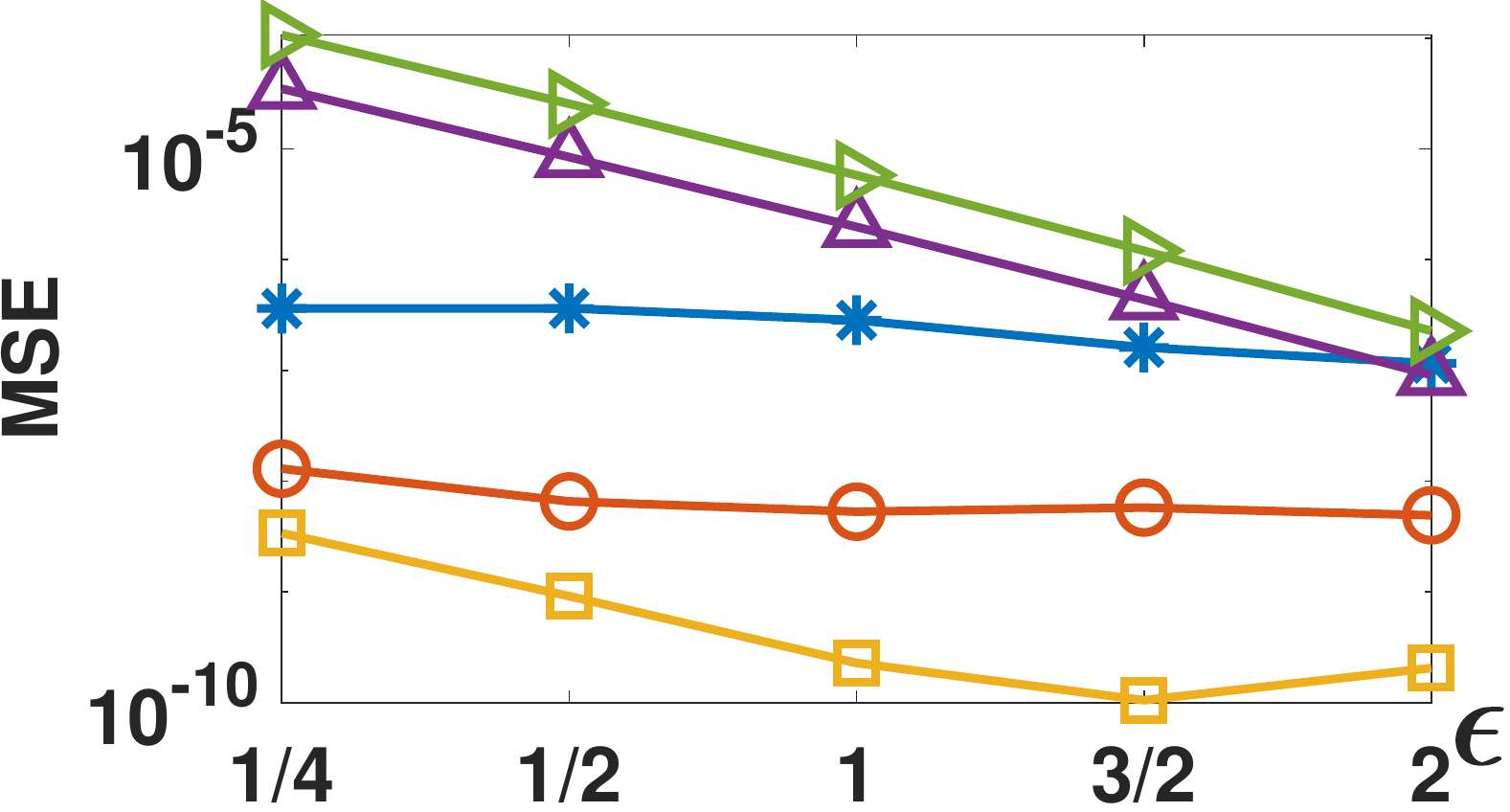}
\end{minipage}%
}%
\subfigure[\textbf{Taxi}, $Poi_{[1C/2,C]}$.]{
\begin{minipage}[t]{0.215\linewidth}
\centering
\includegraphics[width=1\textwidth]{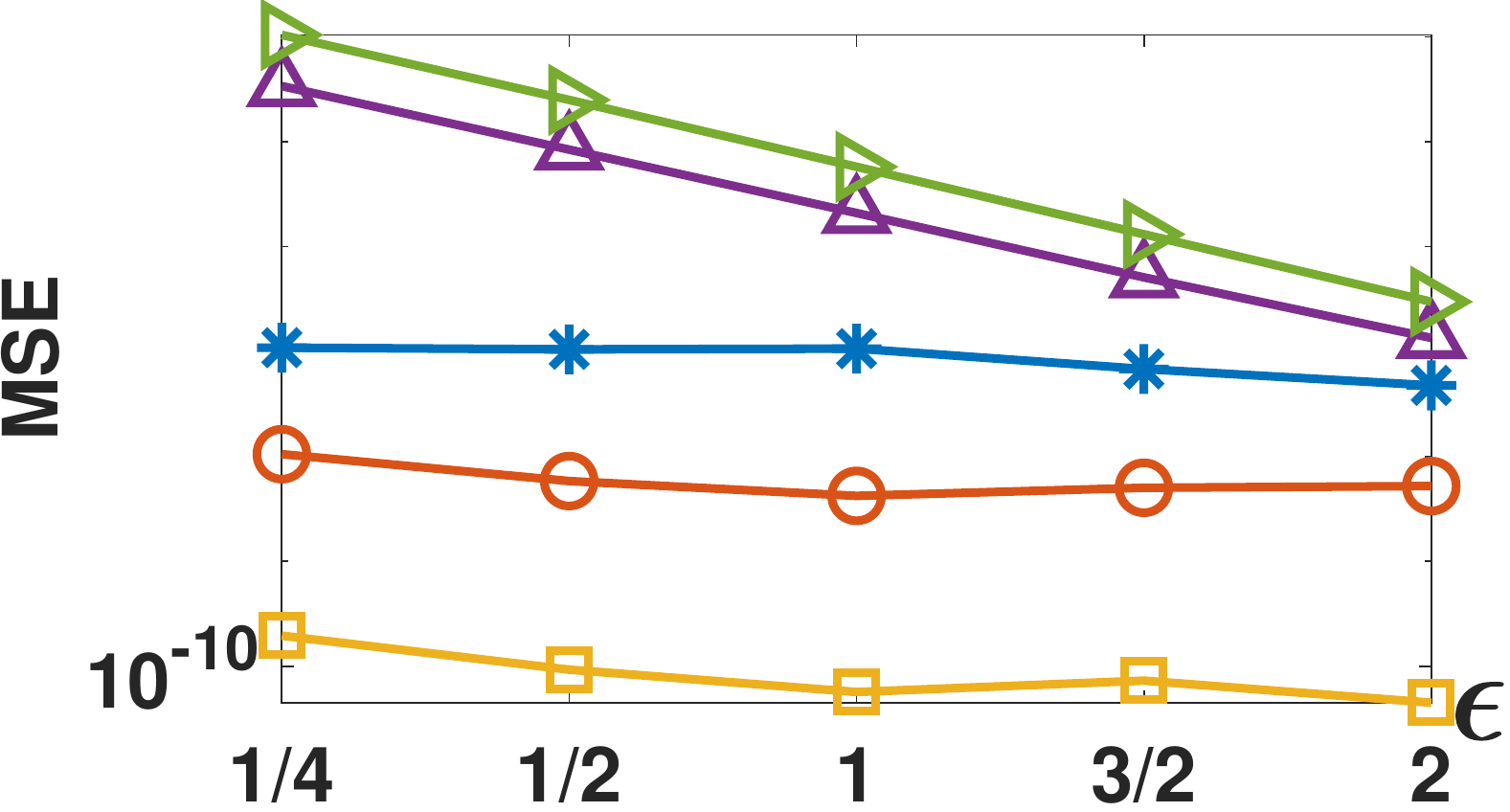}
\end{minipage}
}%
\subfigure[\textbf{Retirement}, $Poi_{[1C/2,C]}$.]{
\begin{minipage}[t]{0.215\linewidth}
\centering
\includegraphics[width=1\textwidth]{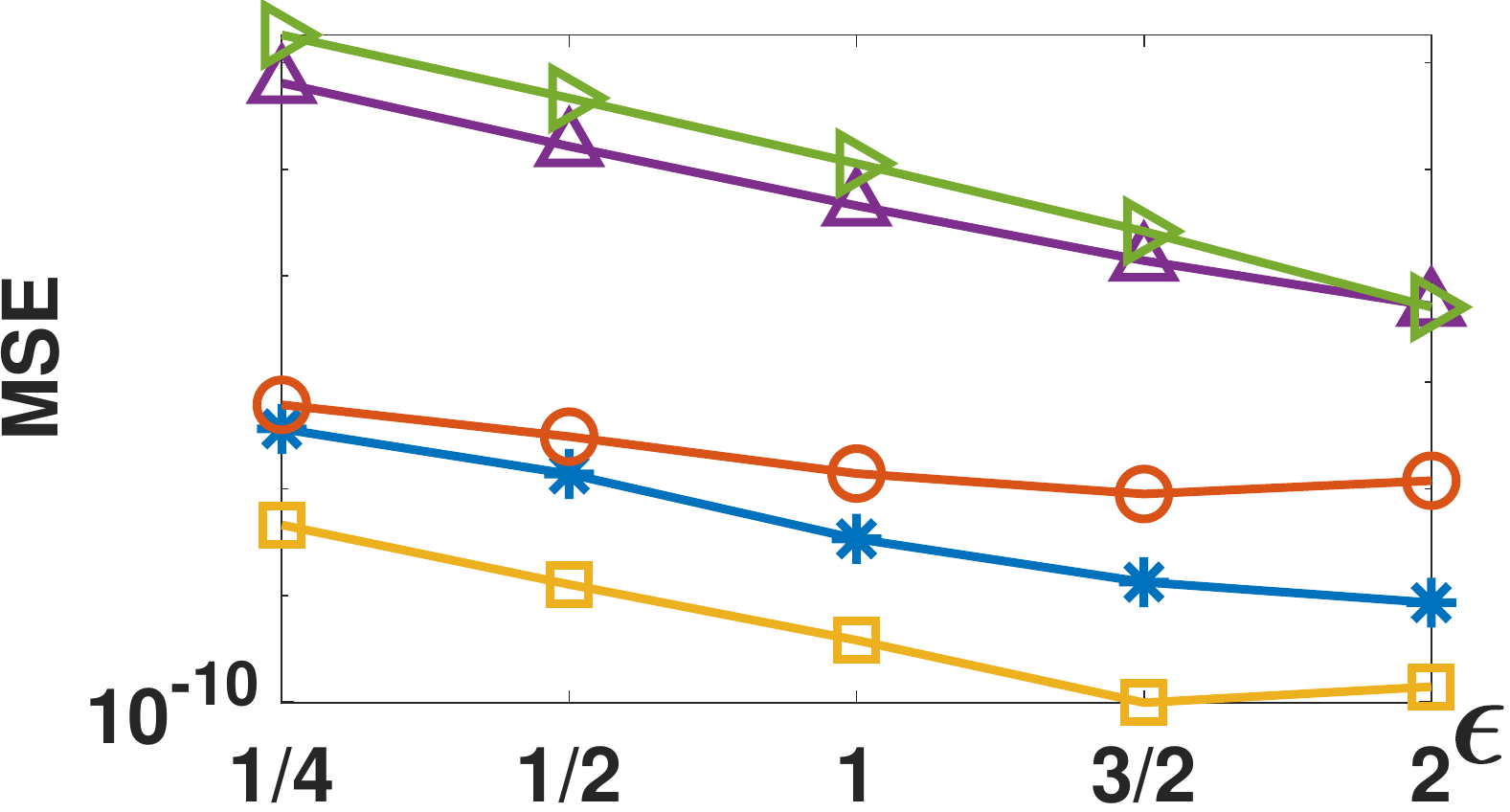}
\end{minipage}
}%
 \vspace{-0.15in}
\centering
\subfigure[\textbf{Beta(2,5)}, $Poi_{[O,1C/2]}$.]{
\begin{minipage}[t]{0.215\linewidth}
\centering
\includegraphics[width=1\textwidth]{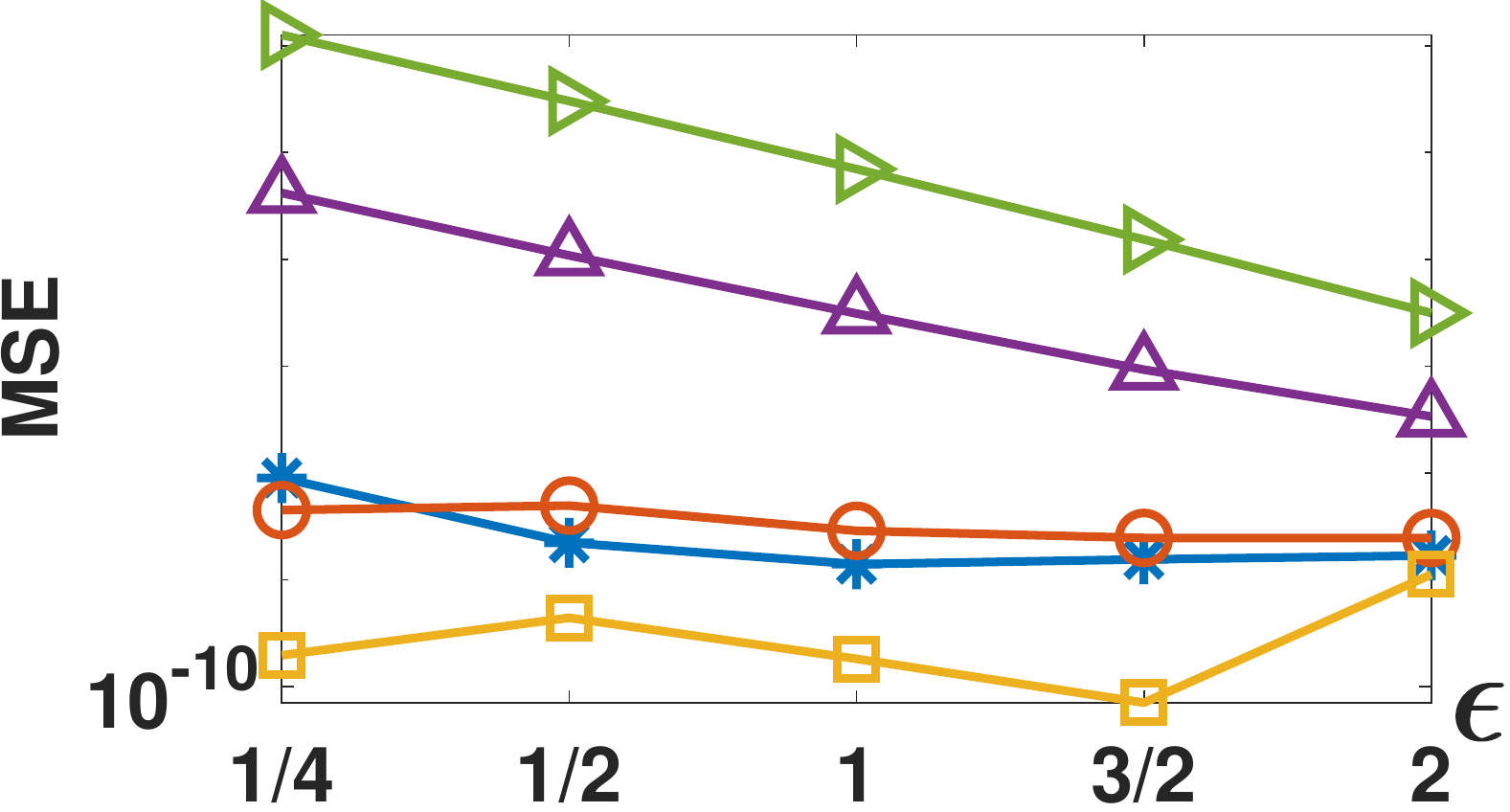}
\end{minipage}%
}%
\subfigure[\textbf{Beta(5,2)}, $Poi_{[O,1C/2]}$.]{
\begin{minipage}[t]{0.215\linewidth}
\centering
\includegraphics[width=1\textwidth]{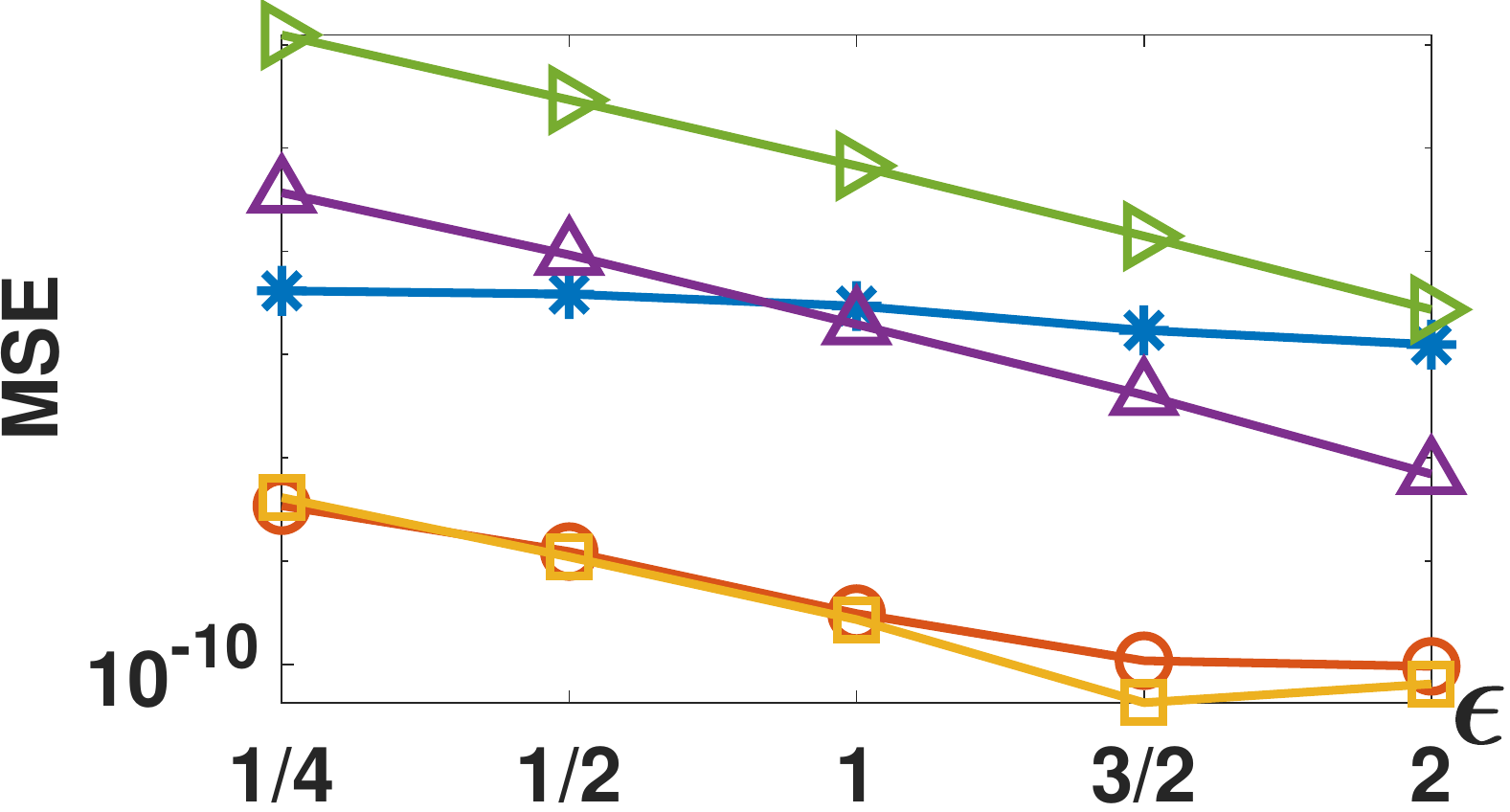}
\end{minipage}%
}%
\subfigure[\textbf{Taxi}, $Poi_{[O,1C/2]}$.]{
\begin{minipage}[t]{0.215\linewidth}
\centering
\includegraphics[width=1\textwidth]{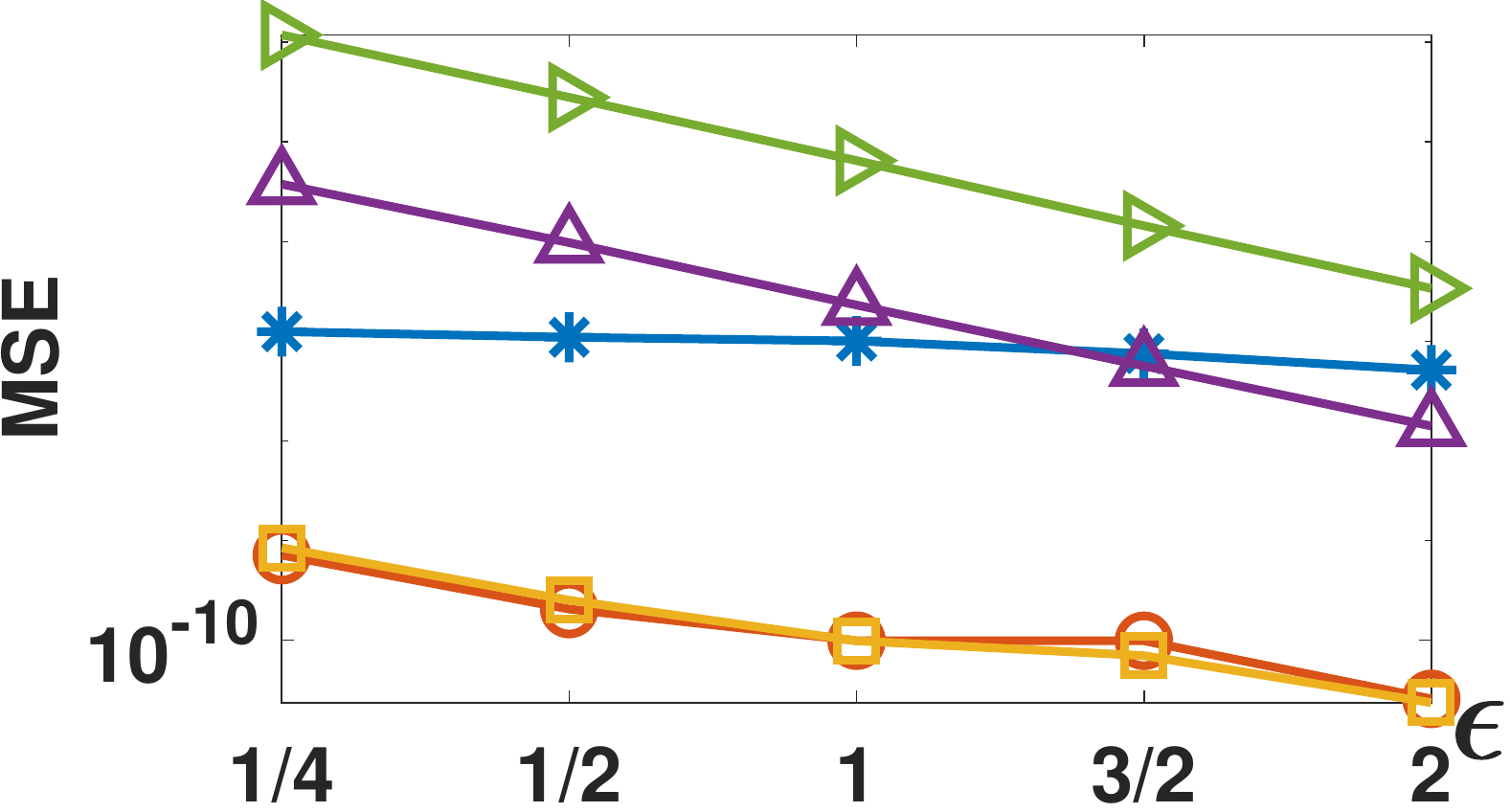}
\end{minipage}
}%
\subfigure[\textbf{Retirement}, $Poi_{[O,1C/2]}$.]{
\begin{minipage}[t]{0.215\linewidth}
\centering
\includegraphics[width=1\textwidth]{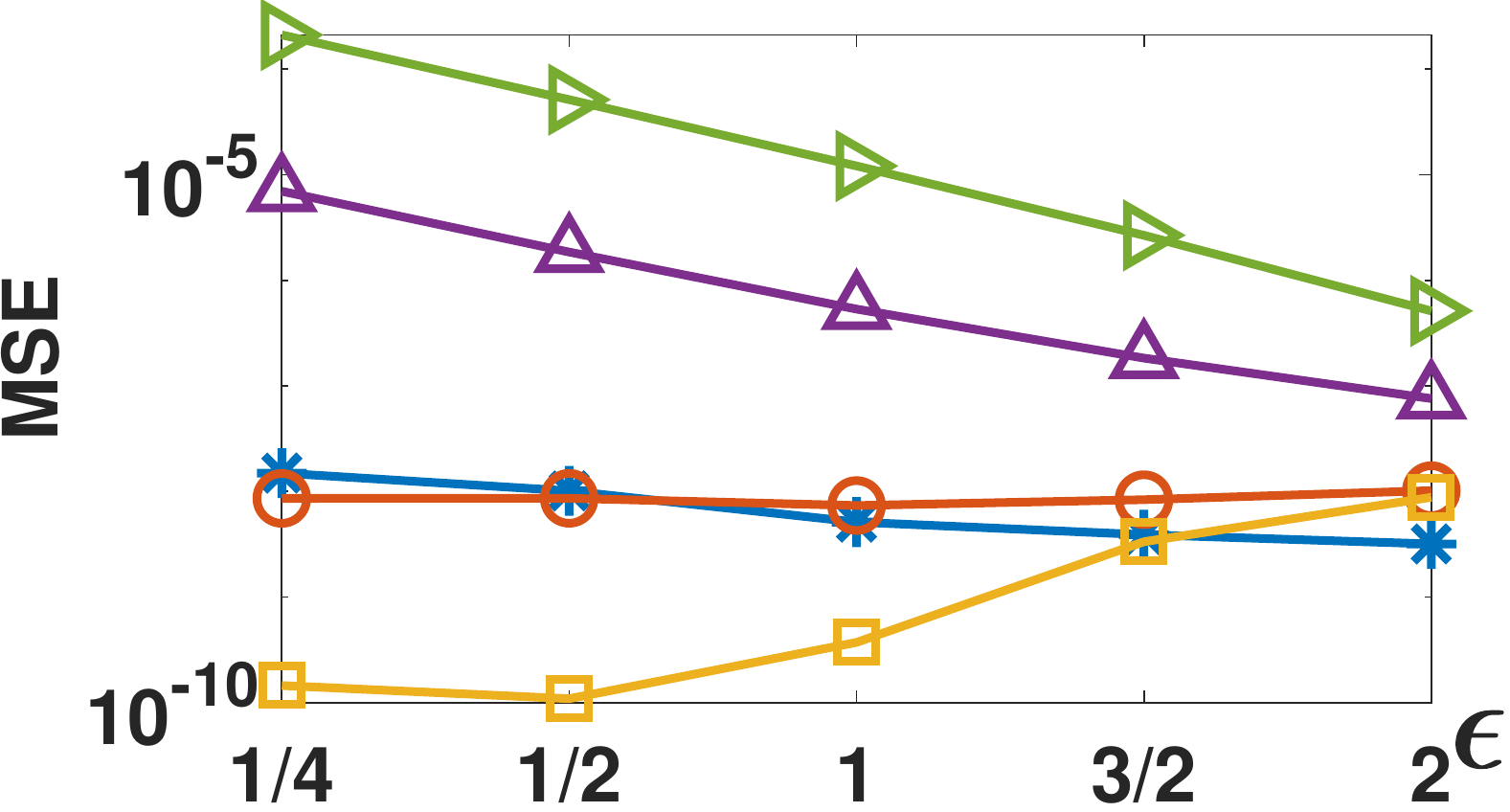}
\end{minipage}
}%
 \vspace{-0.15in}
\centering
\subfigure[\textbf{Beta(2,5)}, $Poi_{[O,C]}$.]{
\begin{minipage}[t]{0.215\linewidth}
\centering
\includegraphics[width=1\textwidth]{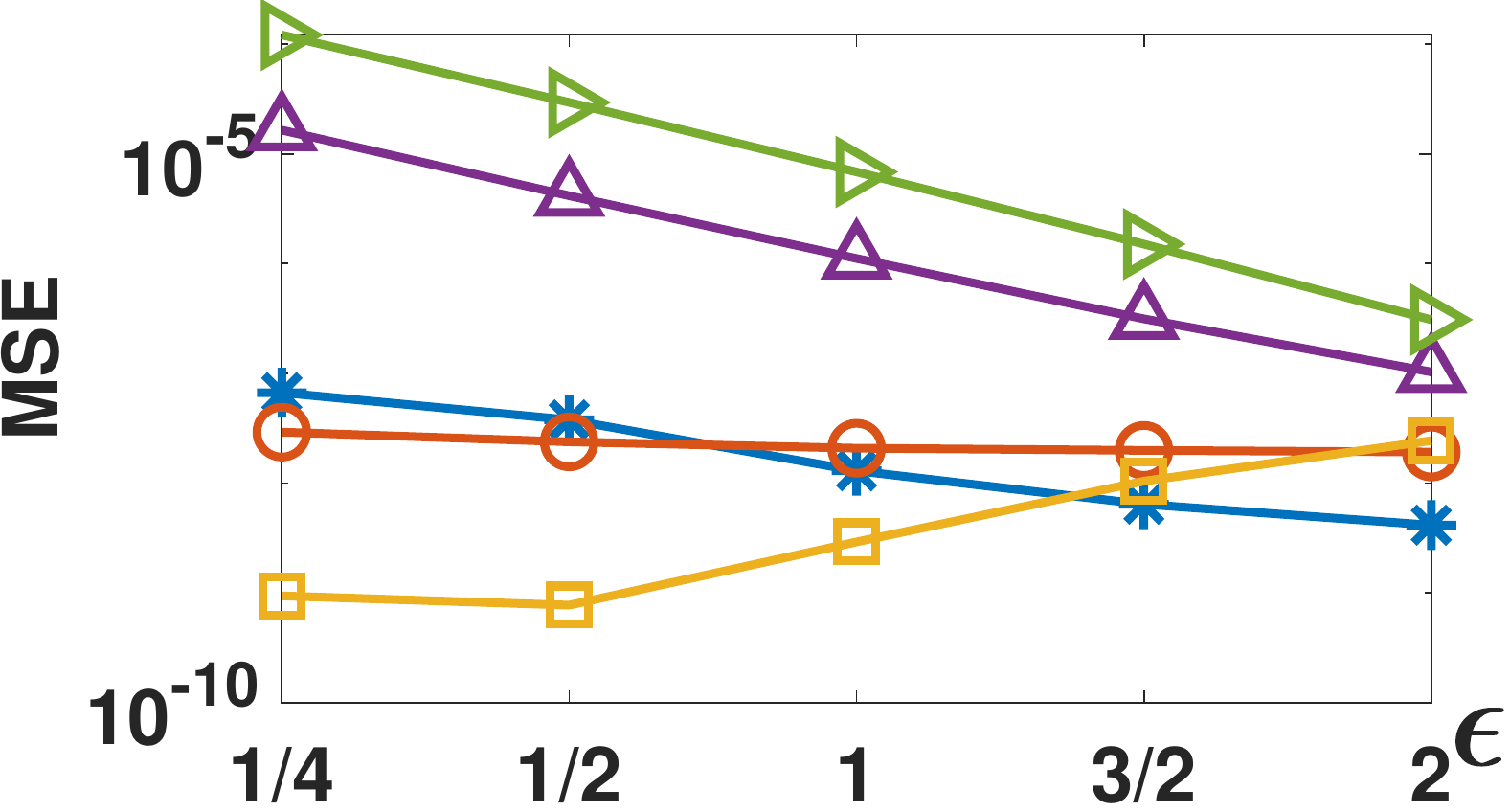}
\end{minipage}%
}%
\subfigure[\textbf{Beta(5,2)}, $Poi_{[O,C]}$.]{
\begin{minipage}[t]{0.215\linewidth}
\centering
\includegraphics[width=1\textwidth]{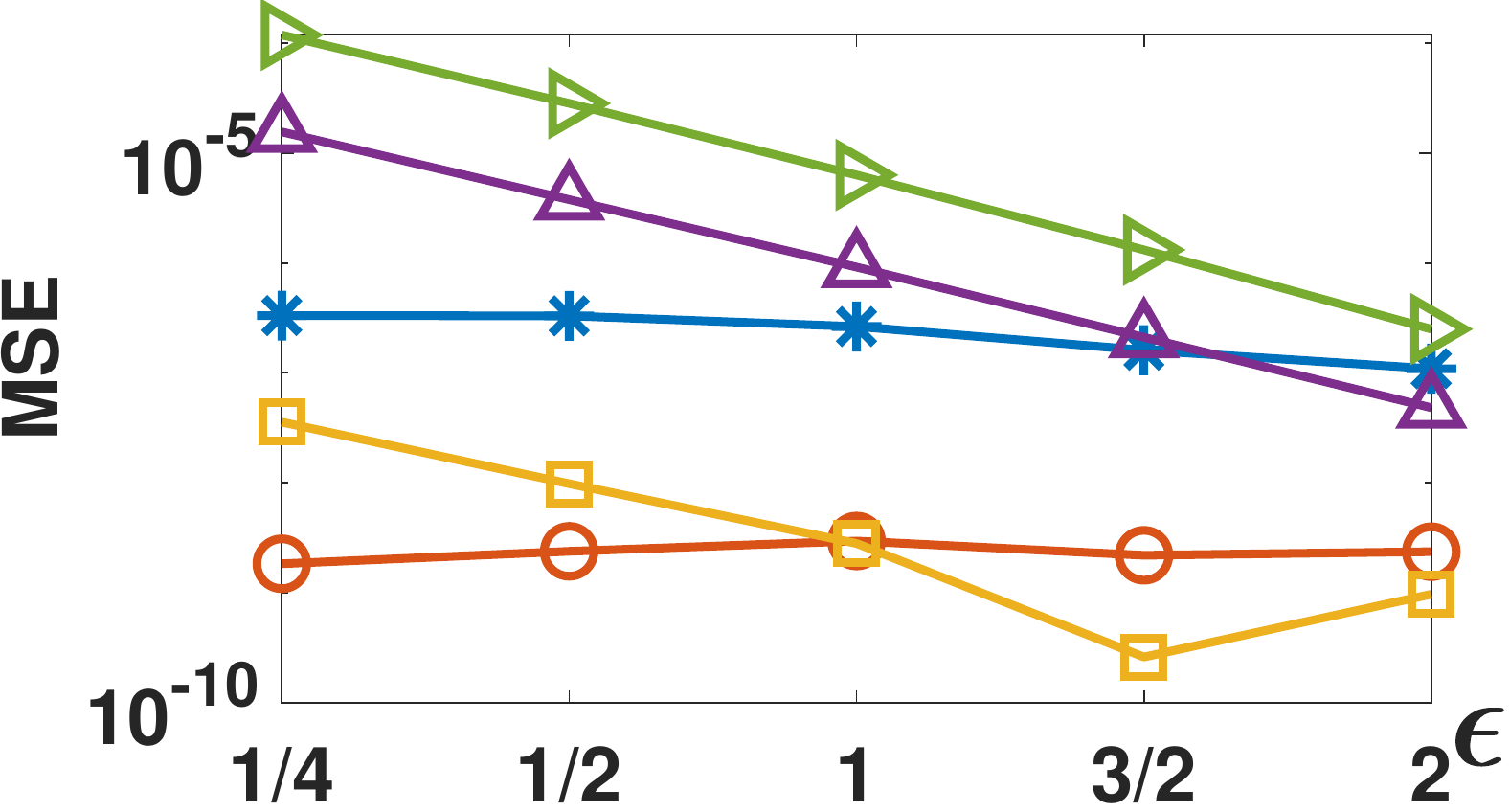}
\end{minipage}%
}%
\subfigure[\textbf{Taxi}, $Poi_{[O,C]}$. ]{
\begin{minipage}[t]{0.215\linewidth}
\centering
\includegraphics[width=1\textwidth]{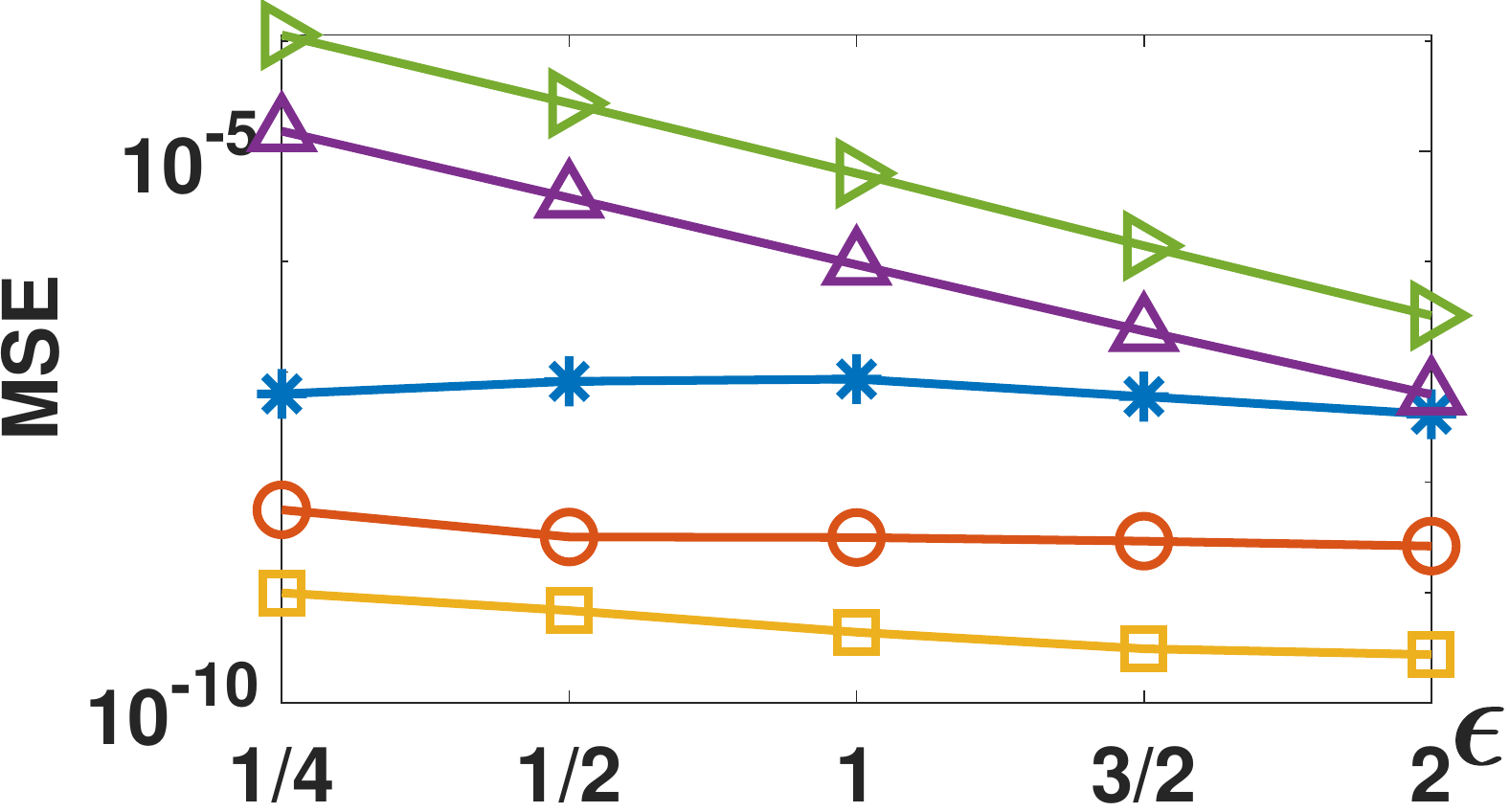}
\end{minipage}
}%
\subfigure[\textbf{Retirement}, $Poi_{[O,C]}$.]{
\begin{minipage}[t]{0.215\linewidth}
\centering
\includegraphics[width=1\textwidth]{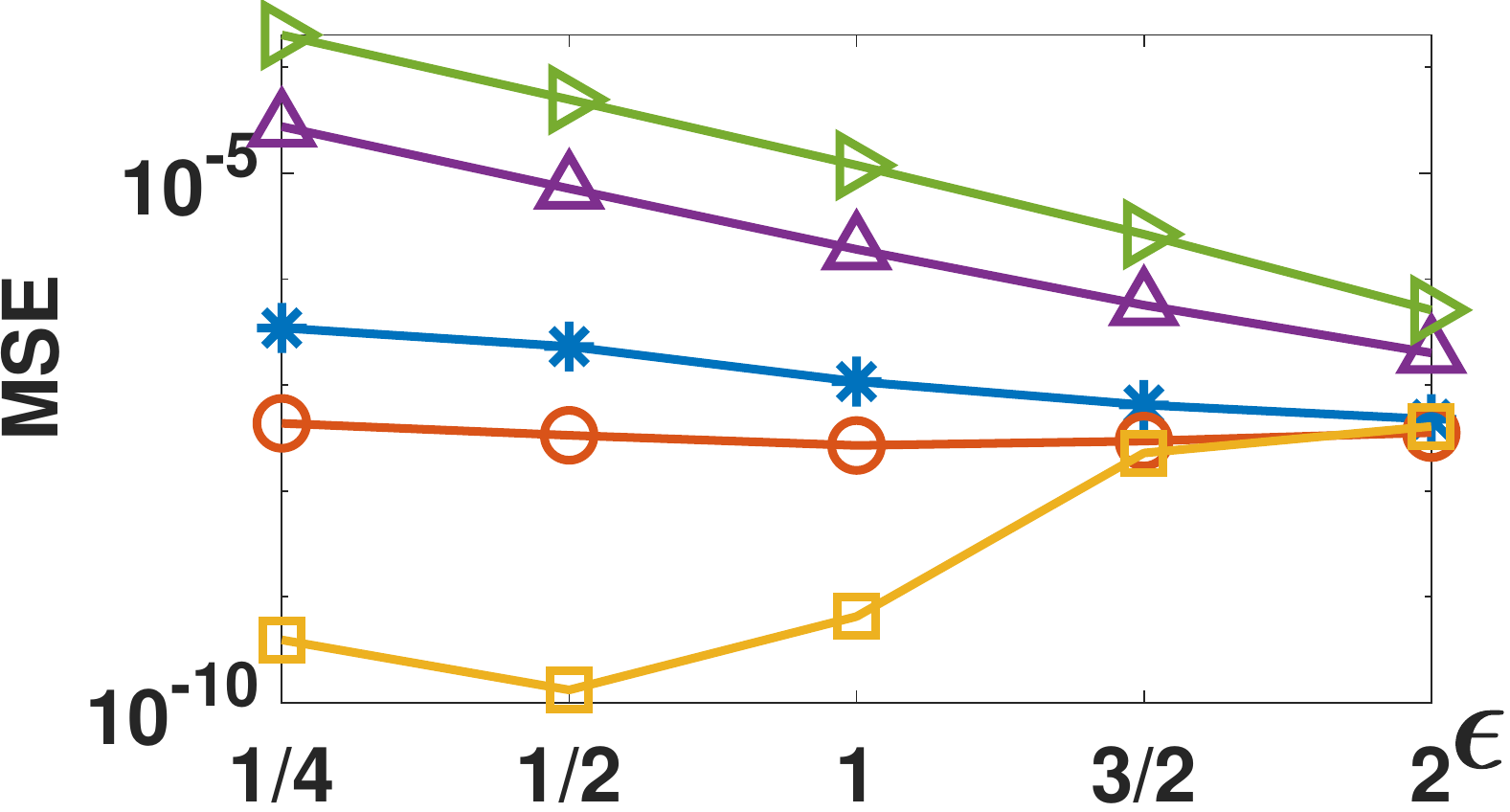}
\end{minipage}
}%
\centering
  \vspace{-0.15in}
\caption{MSE on mean estimation w.r.t. $\epsilon$}
\label{Mean Estimation}
 \vspace{-0.15in}
\end{figure*}
\subsection{Performance of DAP}
\textbf{Mean Estimation.} Fig. \ref{Mean Estimation} compares the MSE~\cite{Mood1963Introduction} of mean estimation among five schemes, including $DAP_{EMF}$, $DAP_{EMF^*}$, $DAP_{CEMF^*}$, Ostrich, and Trimming. Ostrich is the baseline method where all values are directly averaged and the existence of Byzantine users is ignored. For Trimming, the data collector simply removes the largest (or smallest if the poison side is to the left) 50\% values in $V'$ and averages the remaining values.

EMF, EMF*, and CEMF* estimations are used in three DAP schemes, known as $DAP_{EMF}$, $DAP_{EMF^*}$, and $DAP_{CEMF^*}$, respectively. Note that the latter two are post-processing of EMF, so essentially they mean EMF+EMF* and EMF+CEMF*. In the DAP protocol, the minimum privacy budget $\epsilon_0$ in all groups is set to 1/16, and $\epsilon$ is set to $\{1/4, 1/2, 1, 3/2, 2\}$. The lower threshold of $DAP_{CEMF^*}$ to suppress buckets is set to $0.5\hat{\gamma}/(d'/2)$.

In Fig. \ref{Mean Estimation}, in most cases Trimming underperforms others as it overkills many normal values. Furthermore, all proposed schemes outperform Ostrich and Trimming, as they more or less manage to probe some poison values. Among the three schemes we propose, $DAP_{EMF^*}$ performs better than $DAP_{EMF}$ in most cases because $DAP_{EMF}$ may remove an inappropriate number of values, which is resolved by $DAP_{EMF^*}$. In most cases, $DAP_{CEMF^*}$ behaves better than $DAP_{EMF^*}$, as the former further suppresses some buckets unchosen by Byzantine users and therefore has a smaller error on the convergence result. \textcolor{black}{In Figs. \ref{Mean Estimation} (j) (k) (n), we observe that EMF performs worse than the Ostrich scheme when $\epsilon$ is large in \textbf{Beta(5,2)} and \textbf{Taxi}. This is because the poison values are coincidentally close to the true mean value $O$, while EMF probes too many poison values and incorrectly removes some normal values.}

In Figs.~\ref{Mean Estimation} (l) (m) (p), we observe that the MSE of $DAP_{CEMF^*}$ decreases with increasing $\epsilon$, because EMF tends to overestimate poison values when $\epsilon$ and the poison values' range are large. Therefore, CEMF* is less likely to suppress buckets without poison values, which degrades the utility. Nonetheless, $DAP_{CEMF^*}$ is still better than most other schemes. As a final note, since empirically EMF* probes poison values better than EMF, $DAP_{EMF*}$ is expected to outperform $DAP_{EMF}$. However, it contradicts with the results in Fig.~\ref{Mean Estimation} (d) (e) (h). We attribute this phenomenon to false positives, which are misjudge values that should have contributed to mean estimation as normal users. Nonetheless, the performance gap between  $DAP_{EMF*}$ and $DAP_{EMF}$ is not significant.

\begin{figure*}[]

 \hspace{-0.75in}
  {
  \begin{minipage}{3cm}
   \centering
   \includegraphics[scale=0.4]{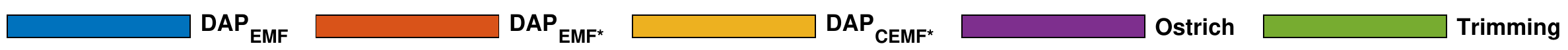}
  \end{minipage}
   \vspace{-0.05in}
 }
 \\
 \vspace{-0.02in}
\centering
\subfigure[\textbf{Taxi}, $Poi_{[O,C/2]}$.]{
\begin{minipage}[t]{0.215\linewidth}
\centering
\includegraphics[width=1\textwidth]{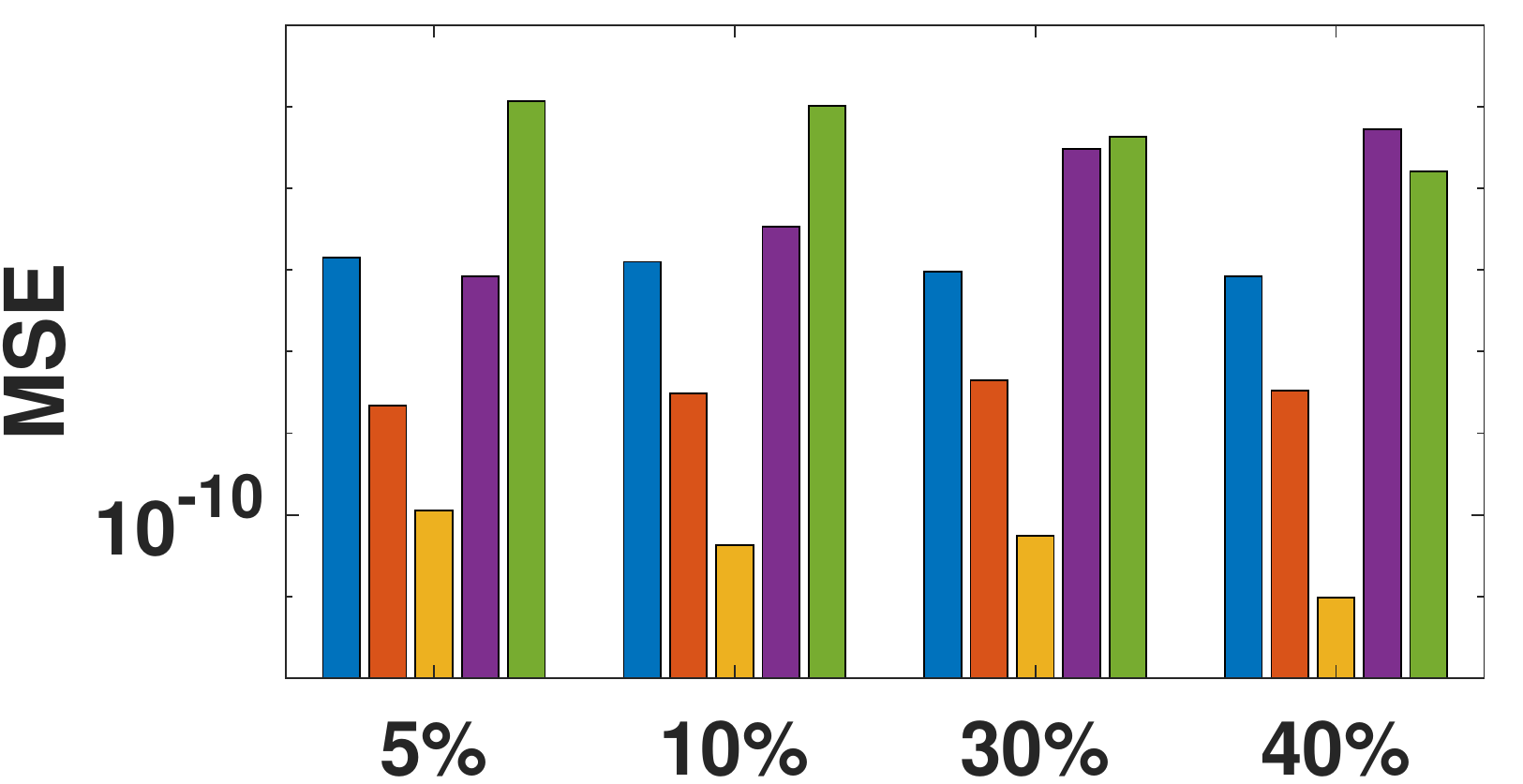}
\end{minipage}
}%
\subfigure[\textbf{Taxi}, $Poi_{[C/2,C]}$.]{
\begin{minipage}[t]{0.215\linewidth}
\centering
\includegraphics[width=1\textwidth]{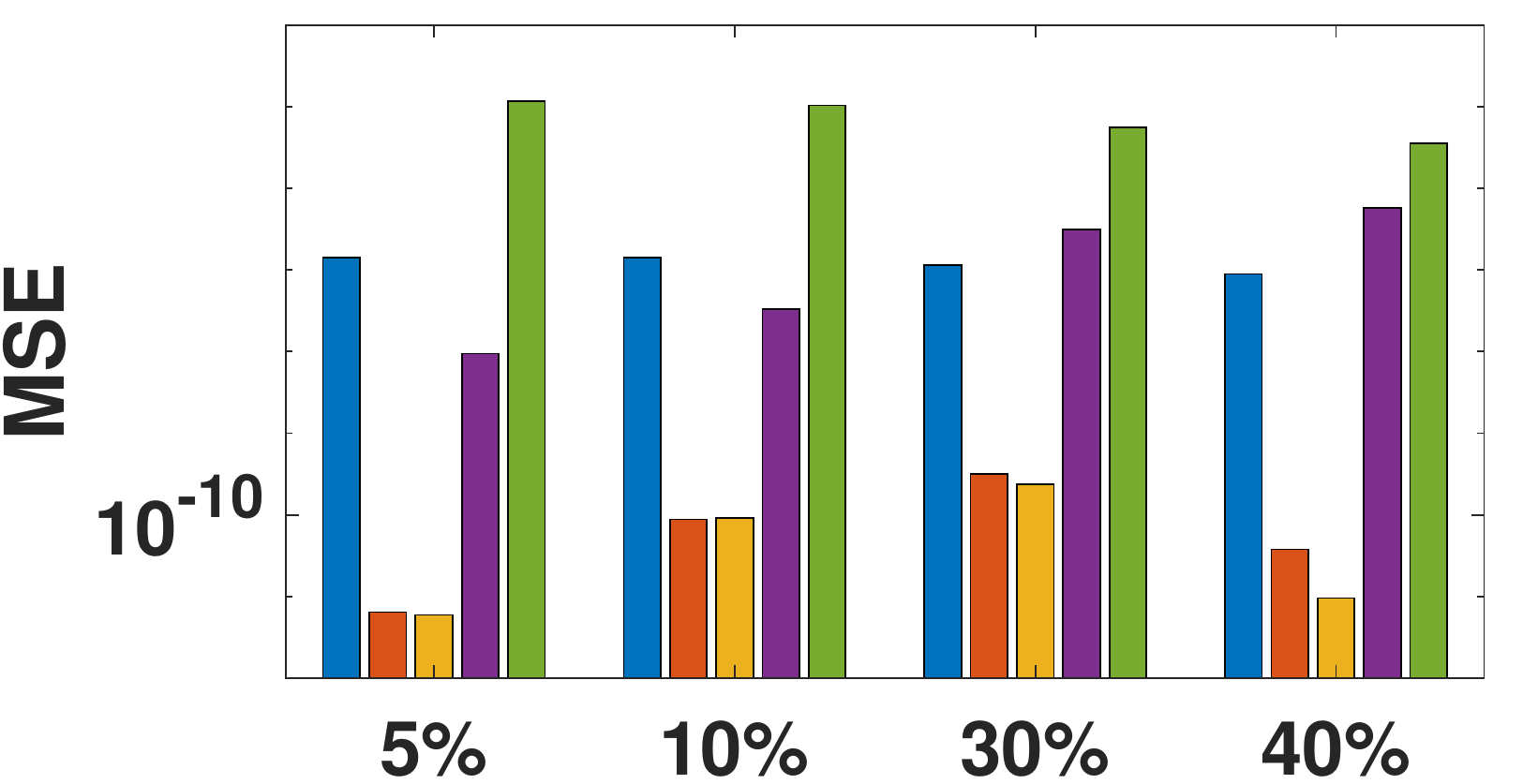}
\end{minipage}
}%
\subfigure[\textbf{Taxi}, $Poi_{[O,C/2]}$.]{
\begin{minipage}[t]{0.215\linewidth}
\centering
\includegraphics[width=1\textwidth]{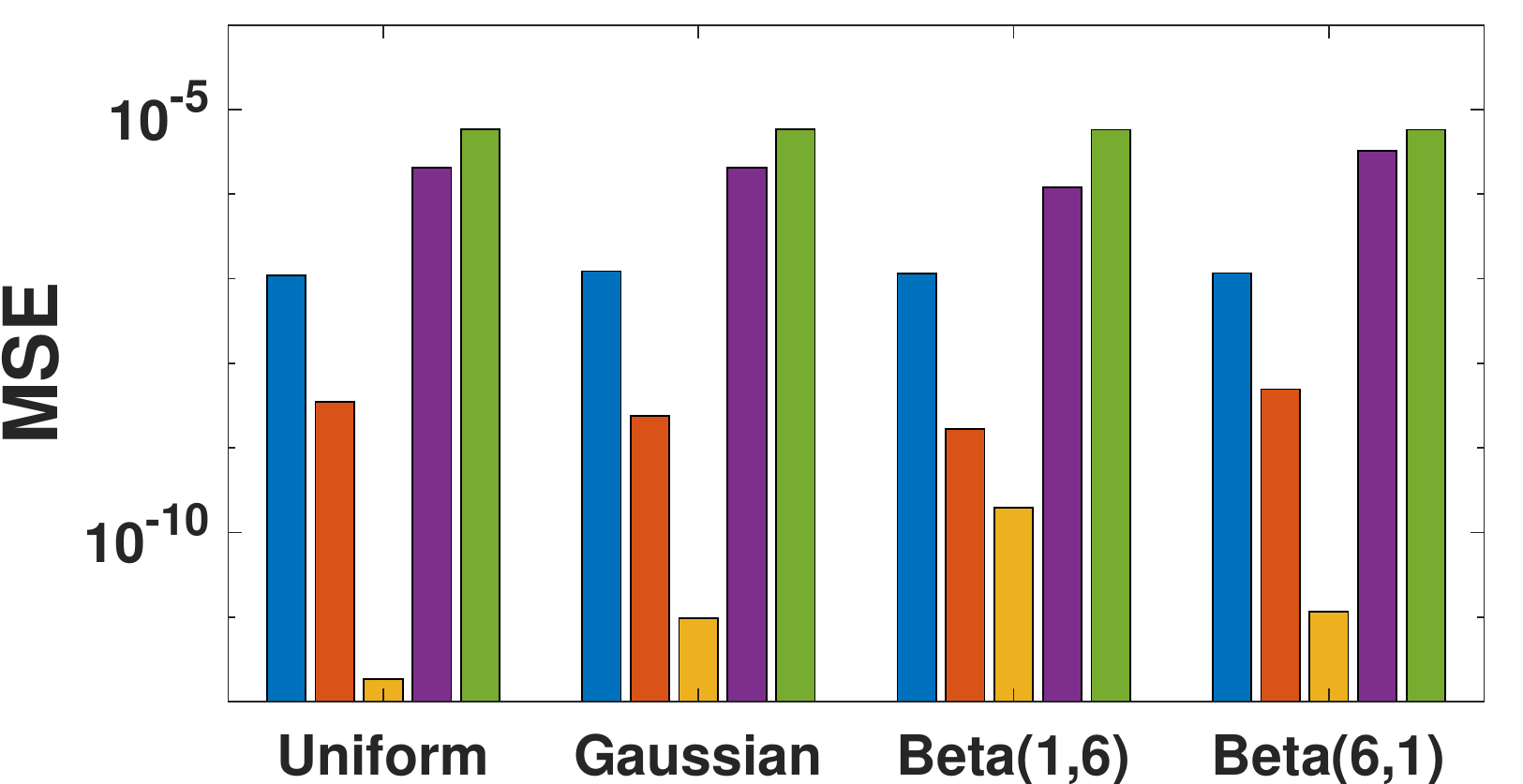}
\end{minipage}
}%
\subfigure[\textbf{Taxi}, $Poi_{[C/2,C]}$.]{
\begin{minipage}[t]{0.215\linewidth}
\centering
\includegraphics[width=1\textwidth]{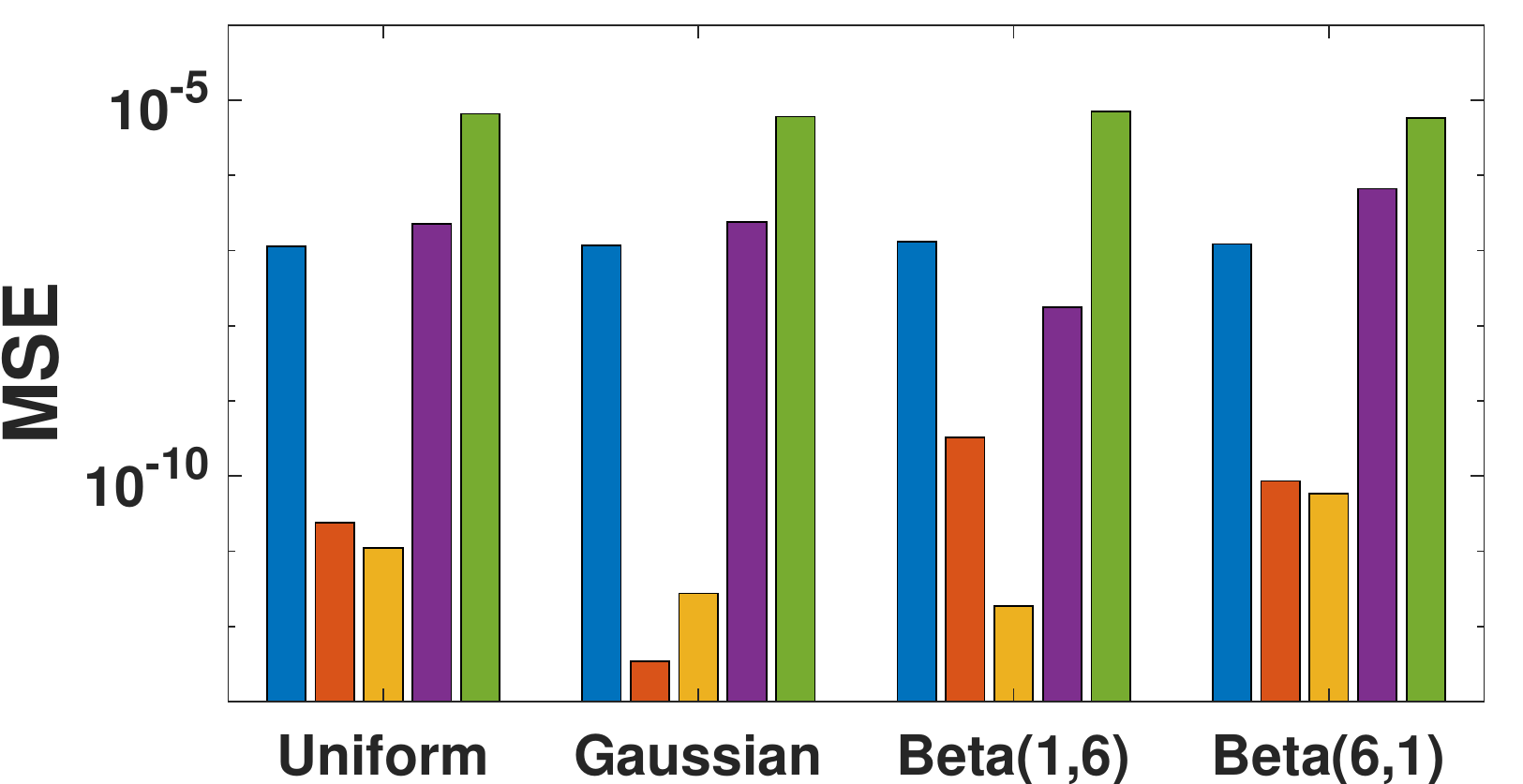}
\end{minipage}
}%
\centering
 \vspace{-0.1in}
\caption{The results of MSE w.r.t. $\gamma$ for (a) (b) and w.r.t. the distribution of Byzantine users for (c) (d)}
\label{Robustness}
\end{figure*}
\textbf{Robustness of Poison Values.} Fig. \ref{Robustness} shows the robustness of MSE according to different proportions of Byzantine users and distributions of poison values in $\textbf{Taxi}$ at $\epsilon=1$. Fig. \ref{Robustness} (a) (b) show that, even with increasing Byzantine users ($\gamma$ reaches 40\%), our proposed schemes can still achieve very low MSE. Fig. \ref{Robustness} (c) (d) show the MSE varies w.r.t. different distributions of poison values. We can find the performance varies according to what the distribution of poison values is, and the proposed schemes always outperform others. In Figs. \ref{Robustness} (a) (b) (d), Ostrcih behaves better than $DAP_{EMF}$ in some cases, when the proportion of Byzantine users is small and the poison values are close to $O$. In Fig. \ref {Robustness} (d), when the poison values follows Gaussian distribution, $DAP_{EMF^*}$ outperforms $DAP_{CEMF^*}$, because the latter may wrongly suppress some buckets and introduce more errors. To sum up, the proposed schemes consistently outperform others under various proportions of Byzantine users or the distributions of poison values.

\subsection{Generalizability Analysis}
\color{black}
\textbf{Extension to Distribution Estimation and SW Mechanism.} Fig. \ref{SW} (a) evaluates the distribution estimation accuracy by using the Wasserstein distance\cite{li2020estimating}. The results demonstrate that our schemes achieve at least 10\% higher utility than Ostrich, which ignores the presence of poison values, whereas our schemes can still refine the distribution by removing them. Figs. \ref{SW}(b)-(d) evaluate the DAP performance on SW. Fig. \ref{SW} (b) shows that the estimated proportion of Byzantine users becomes more accurate as $\epsilon$ decreases. In  Figs. \ref{SW} (c) (d), we compare the MSE of different schemes with poison values' range $[1+b/2,1+b]$, and observe that our proposed schemes outperform the others in most cases. Specifically, Trimming and Ostrich only perform well for small $\epsilon$ on \textbf{Beta(2,5)} and \textbf{Beta(5,2)}, respectively. This is because when $\epsilon$ is small, the perturbed values are close to a uniform distribution, so the estimated mean tends to be 0.5. But Trimming will estimate a smaller mean as it removes large values, while Ostrich may accidentally admit some large values, making the estimated mean closer to the ground truth.

\begin{figure*}[]
 \vspace{-0.2in}
 \hspace{-5.0in}
 \\
 \vspace{-0.09in}
\centering
\subfigure[\textcolor{black}{\textbf{Beta(2,5)}, O=0.3003.}]{
\begin{minipage}[t]{0.215\linewidth}
\centering
\includegraphics[width=1\textwidth]{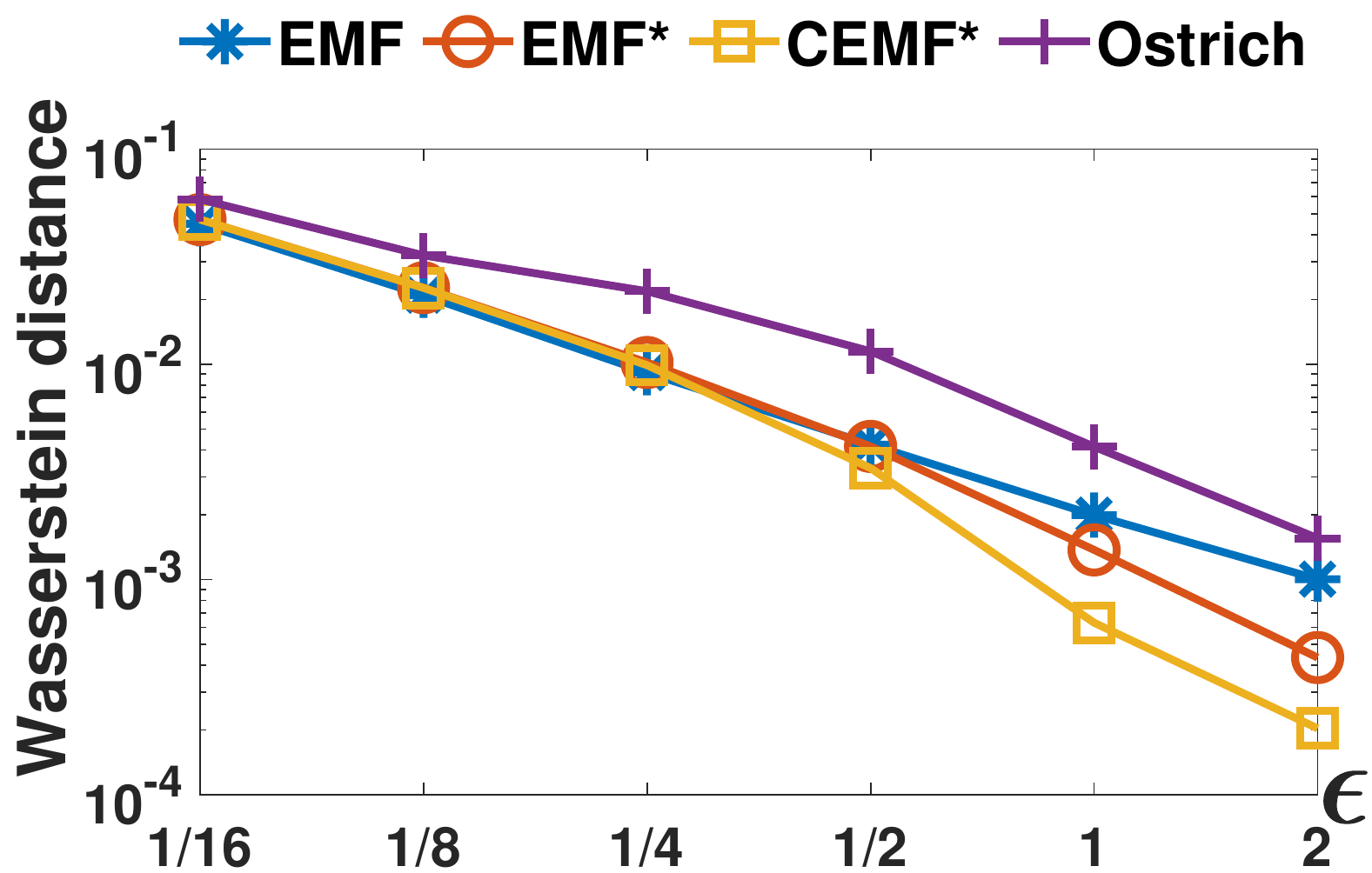}
\end{minipage}
}%
\subfigure[\textcolor{black}{\textbf{Beta(5,2)}, O=0.7068.}]{
\begin{minipage}[t]{0.215\linewidth}
\centering
\includegraphics[width=1\textwidth]{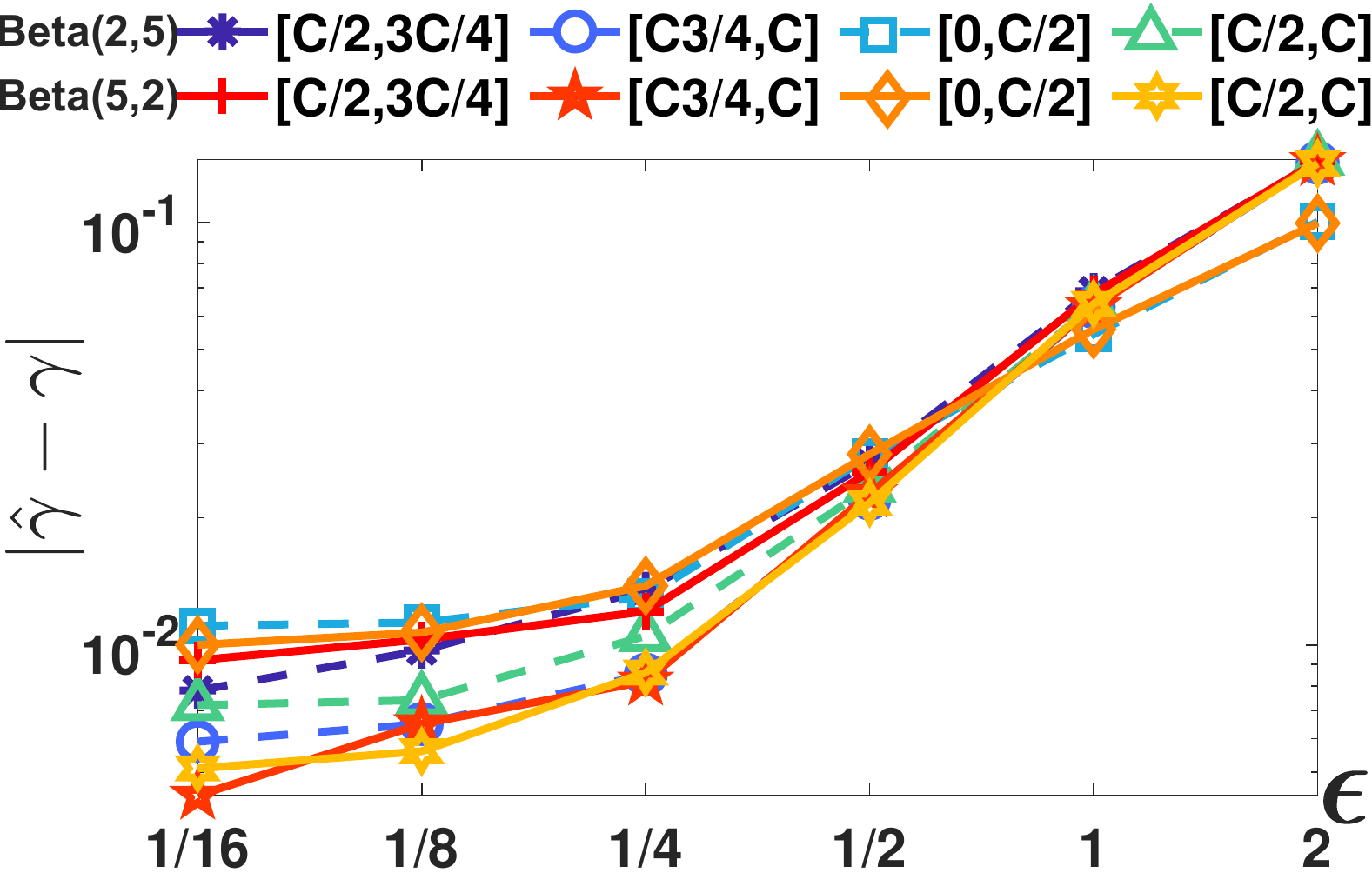}
\end{minipage}
}%
\subfigure[\textbf{Beta(2,5)}, O=0.3003.]{
\begin{minipage}[t]{0.215\linewidth}
\centering
\includegraphics[width=1\textwidth]{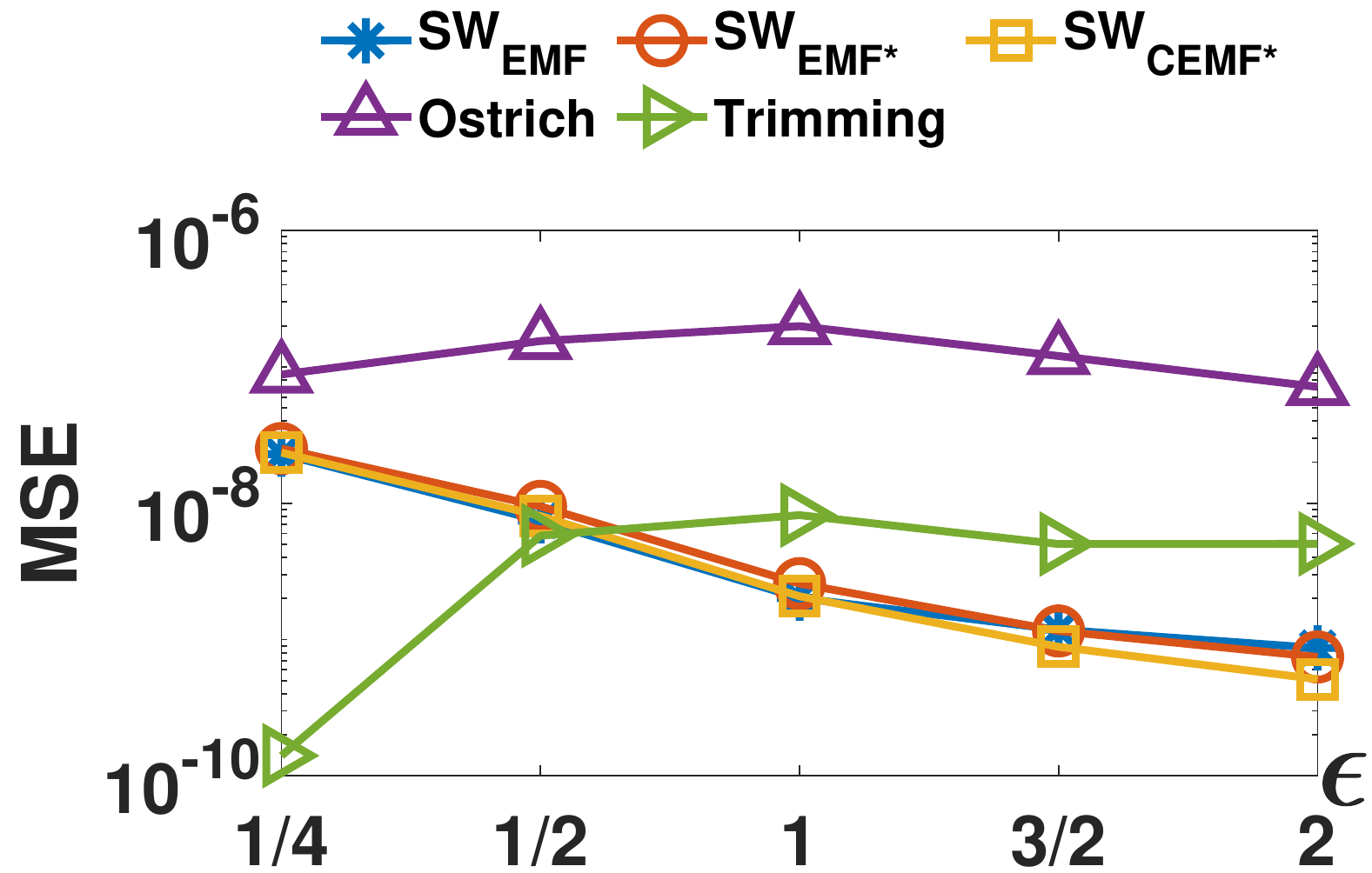}
\end{minipage}
}%
\subfigure[\textbf{Beta(5,2)}, O=0.7068.]{
\begin{minipage}[t]{0.215\linewidth}
\centering
\includegraphics[width=1\textwidth]{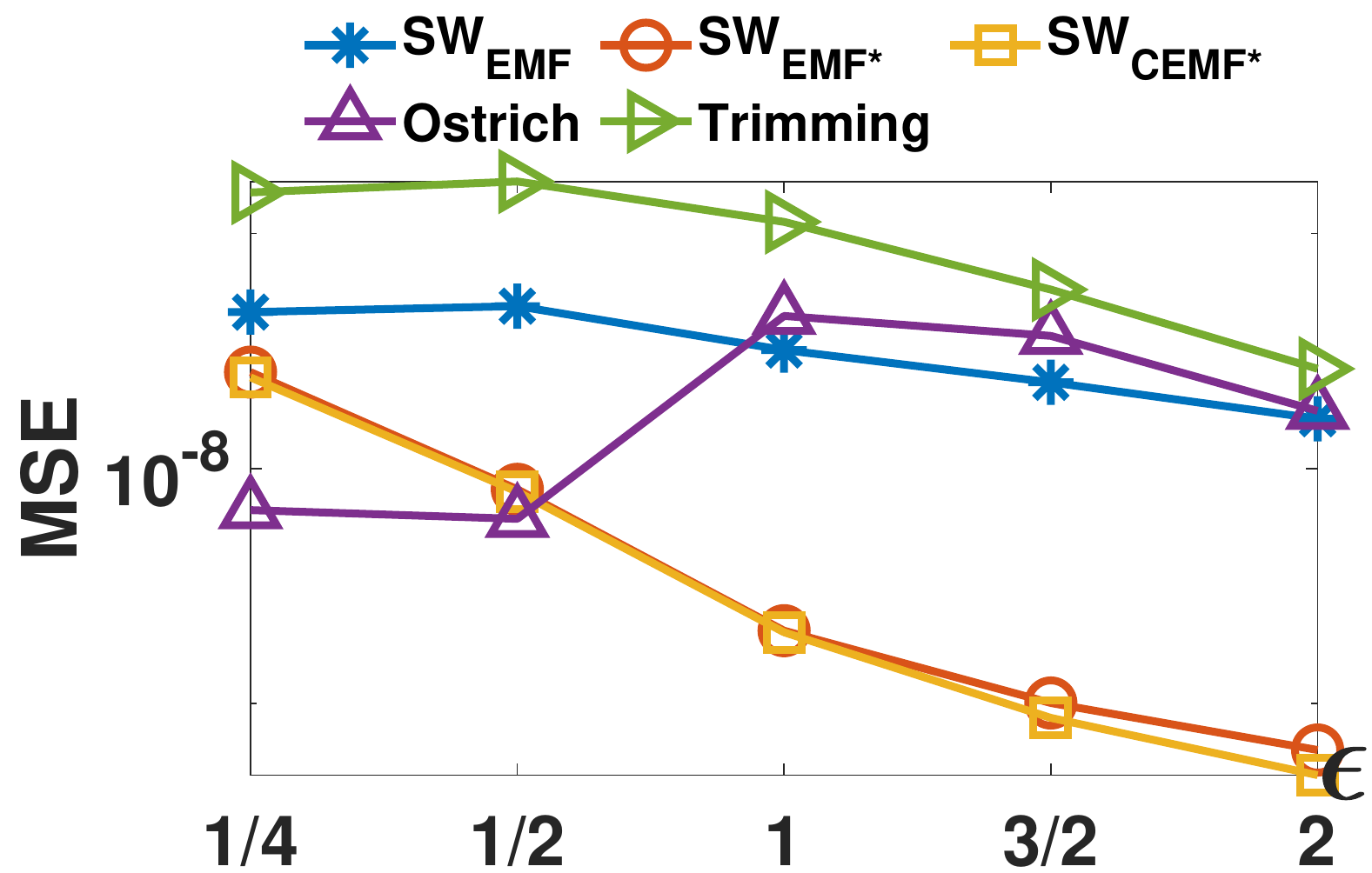}
\end{minipage}
}%
\centering
\caption{\textcolor{black}{The results of distribution estimation for (a), SW w.r.t. $\epsilon$, $|\hat{\gamma}-\gamma|$ for (b) and MSE for (c) (d)}}
\label{SW}
\end{figure*}

\color{black}
\textbf{Comparison with k-means-based Defense.}
We compare our solution with the k-means-based defense~\cite{li2022fine} against both BBA and input manipulation attacks on \textbf{Taxi}. K-means-based defense samples multiple subsets of users and then clusters them into two clusters. The one with more subsets is treated as normal users and used for estimation, while the other is discarded. Let $\beta$ be the sampling rate, and $\beta N$ users are sampled randomly in each subset, with one million subsets in each attack.

In Fig. \ref{discussion} (a), the poison values are distributed on $[C/2, C]$ uniformly, and our scheme outperforms the k-means-based defense under BBA attack. The MSE of the k-means family is between $10^{-7}\sim 10^{-5}$, whereas the MSE of $DAP_{EMF}$ is about $10^{-7}$ and that of $DAP_{EMF^*}$ and $DAP_{CEMF^*}$ is about $10^{-10}$. Fig. \ref{discussion} (b) evaluates the input manipulation attack~\cite{cheu2021manipulation} based on PM, i.e., Byzantine users generating an input poison value $g$ and then strictly following PM to make it less detectable. To integrate EMF and EMF* with existing k-means-based defense, we first use EMF to determine whether $\hat{\gamma}$ is relatively small (i.e., evading) as shown in Fig. \ref{fig of EMF} (d), and then we use EMF* to estimate the distribution of inputs by setting $\hat{\gamma}=0$ in Equ. \ref{Mean Estimation}, and finally we evaluate the mean using k-means. Fig. 9 (b) demonstrates that by such integration (denoted by EMF-based), we can further improve the estimation accuracy under input manipulation attacks. When $g=-1$, the MSE of EMF-based is between $1.40\sim 1.44\times 10^{-7}$ while that of k-means alone is between $1.77\sim 1.88\times 10^{-7}$, resulted in 28\% improvement. Likewise, we have about 30\% improvement for $g=1$ and $g=0$ cases.

\textbf{Frequency Estimation on Categorical Data.} We evaluate our schemes under frequency estimation on categorical data using dataset \textbf{COVID-19}, which records the number of coronavirus disease 2019 deaths for females in California as of December 14, 2022, by age~\cite{covid192022}. All death records are divided into $15$ age groups, and every record is perturbed locally by k-RR~\cite{kairouz2014extremal, wang2016private}. In Fig. \ref{discussion} (c), Byzantine users are injected into the 10th group only, and we observe that the MSE of Ostrich is about 0.1 and keeps steady regardless of $\epsilon$, while that of our schemes is lower than 0.01 and decreases significantly with respect to $\epsilon$. When Byzantine users uniformly inject poison values into the 10th, 11th, and 12th groups, as shown in Fig. \ref{discussion} (d), the MSE of Ostrich still significantly underperforms our schemes. This experiment shows that our DAP schemes can also work well in other statistics and data types than mean estimation on numerical data.
\begin{figure*}[]
 \vspace{-0.25in}
\centering
\subfigure[\textcolor{black}{\textbf{Taxi}, $Poi_{[C/2,C]}$, $\gamma=0.25$.}]{
\begin{minipage}[t]{0.22\linewidth}
\centering
\includegraphics[width=1\textwidth]{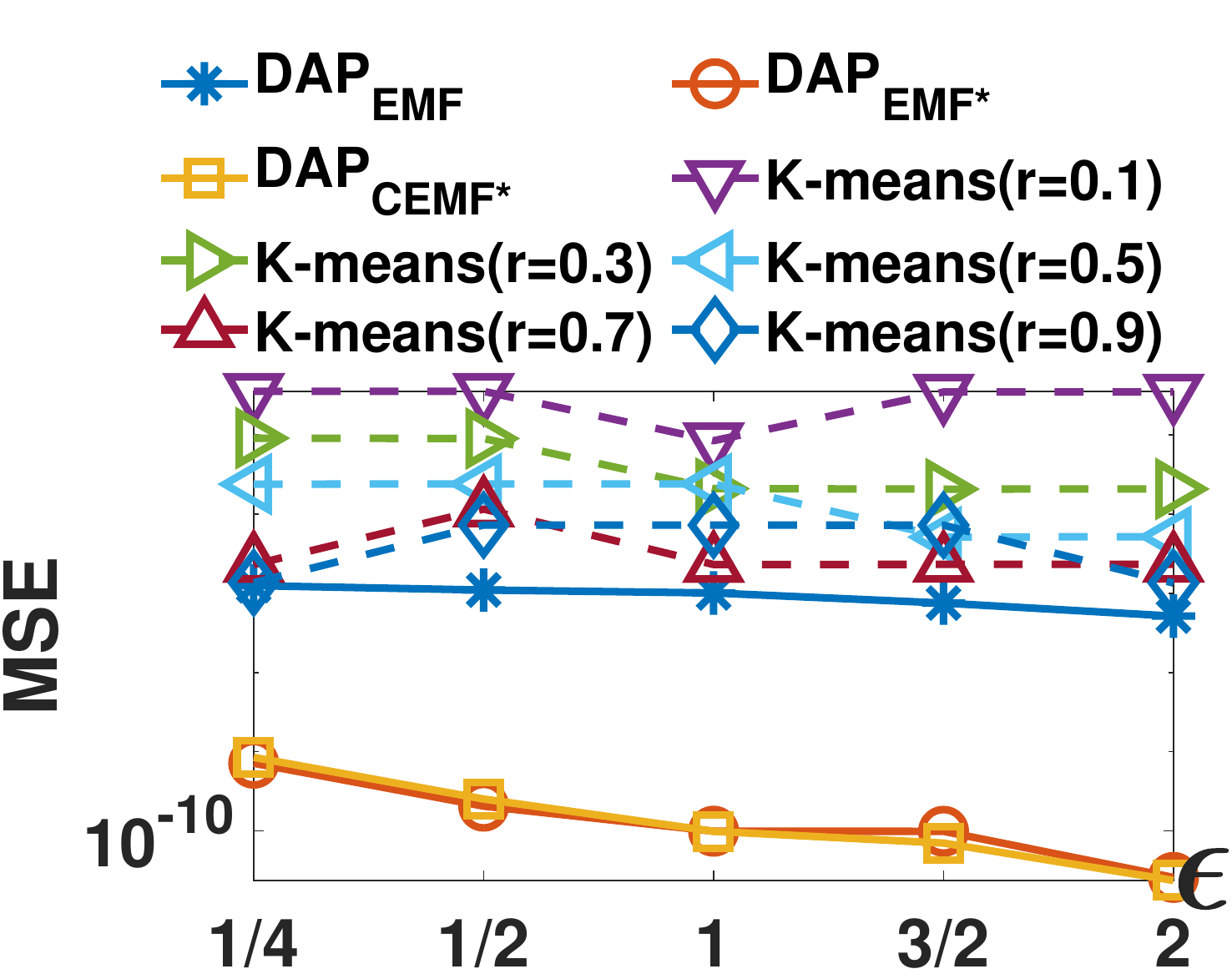}
\end{minipage}
}%
\subfigure[\textcolor{black}{\textbf{Taxi}, IMA, $\gamma=0.25$.}]{
\begin{minipage}[t]{0.218\linewidth}
\centering
\includegraphics[width=1\textwidth]{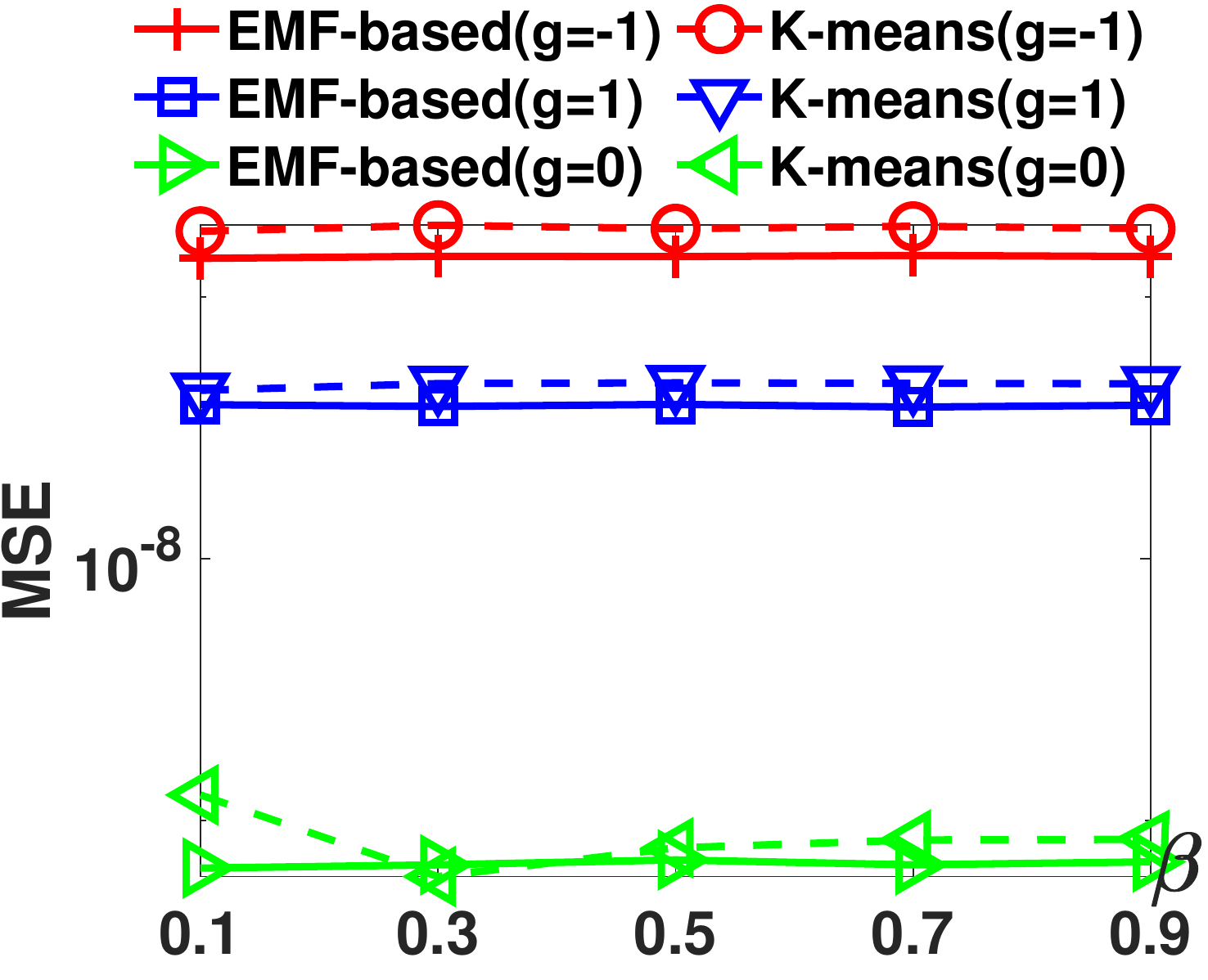}
\end{minipage}
}%
\subfigure[\textcolor{black}{\textbf{COVID-19}, $Poi_{10}$.}]{
\begin{minipage}[t]{0.215\linewidth}
\centering
\includegraphics[width=1\textwidth]{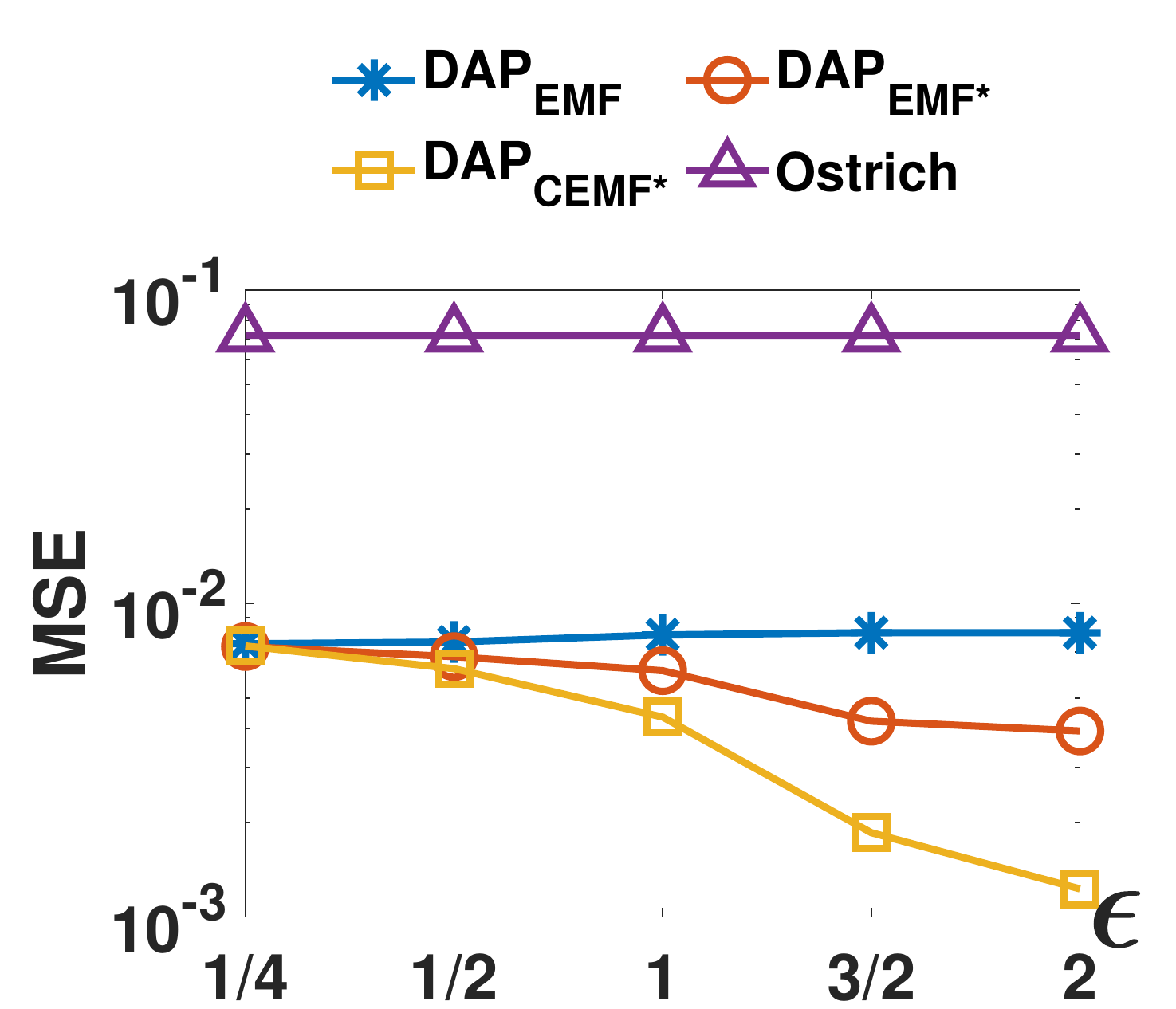}
\end{minipage}
}%
\subfigure[\textcolor{black}{\textbf{COVID-19}, $Poi_{10,11,12}$.}]{
\begin{minipage}[t]{0.215\linewidth}
\centering
\includegraphics[width=1\textwidth]{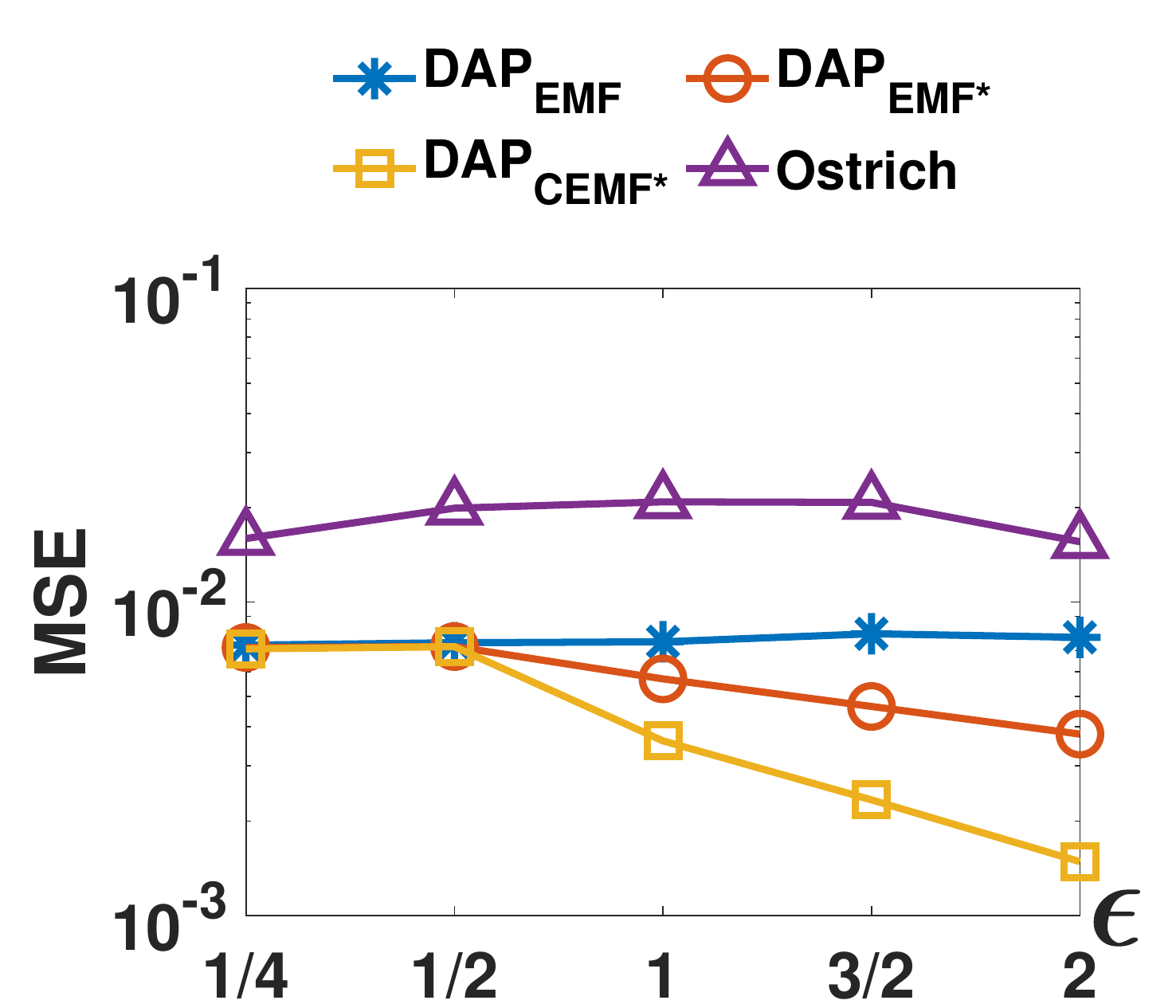}
\end{minipage}
}%
\centering
 \vspace{-0.15in}
\caption{\textcolor{black}{Comparison with k-means-based defense for (a) (b) and extension to frequency estimation for (c) (d)}}
\label{discussion}
 \vspace{-0.15in}
\end{figure*}
\color{black}

\begin{figure*}[]

 \hspace{0.25in}
  {
  \begin{minipage}{3cm}
   \centering
   \includegraphics[scale=0.4]{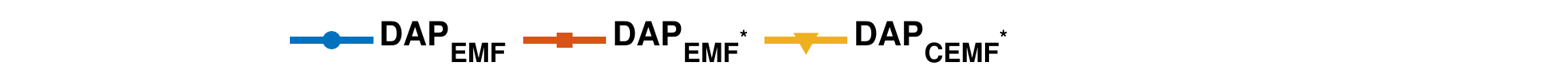}
  \end{minipage}
   \vspace{-0.05in}
 }
 \\
 \vspace{-0.02in}
\centering
\subfigure[\textbf{Beta(2,5)}, $O=0.4136$.]{
\begin{minipage}[t]{0.215\linewidth}
\centering
\includegraphics[width=1\textwidth]{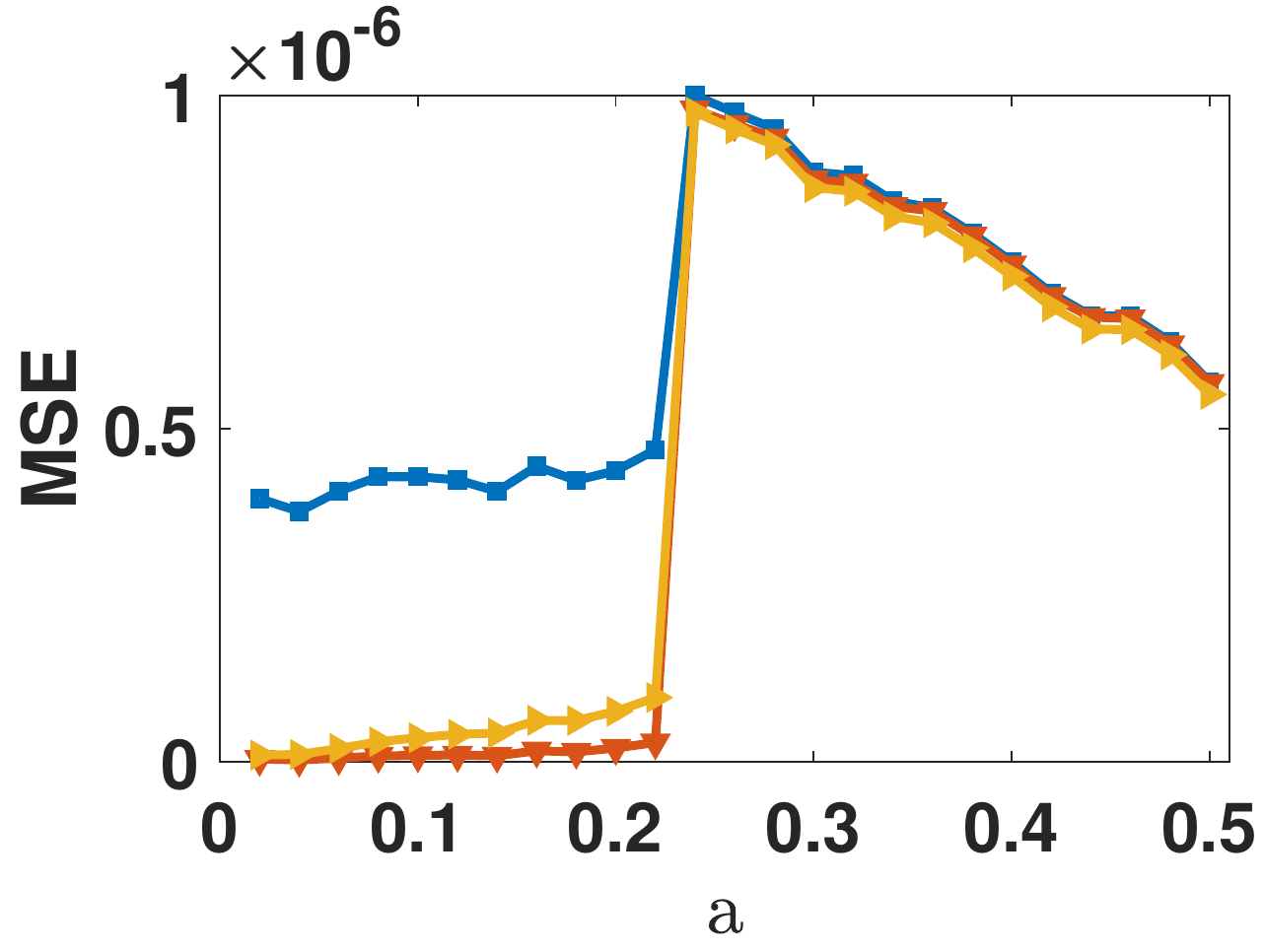}
\end{minipage}
}%
\subfigure[\textbf{Beta(5,2)}, $O=-0.3994$.]{
\begin{minipage}[t]{0.215\linewidth}
\centering
\includegraphics[width=1\textwidth]{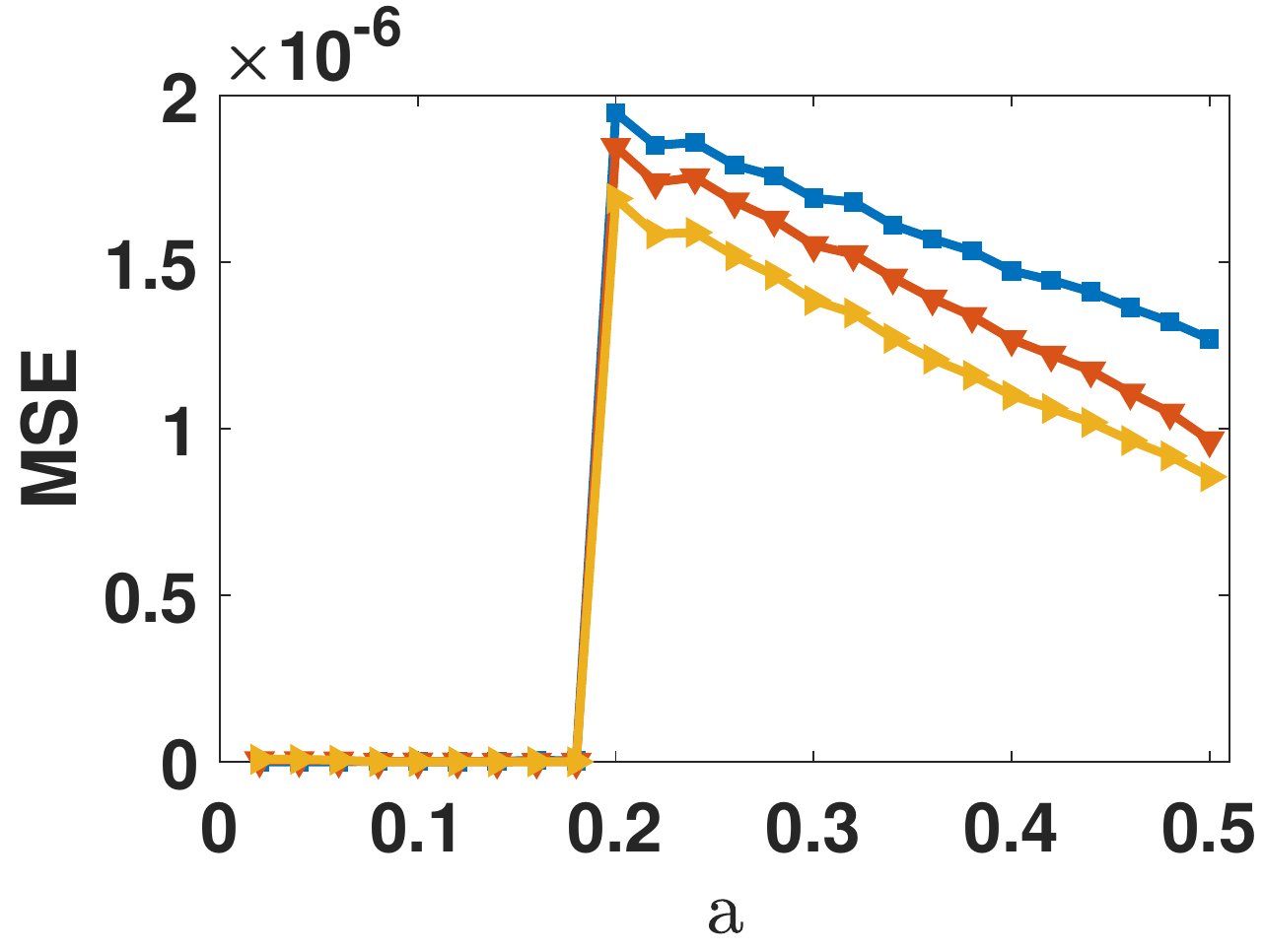}
\end{minipage}
}%
\subfigure[\textbf{Taxi}, $O=0.1190$.]{
\begin{minipage}[t]{0.215\linewidth}
\centering
\includegraphics[width=1\textwidth]{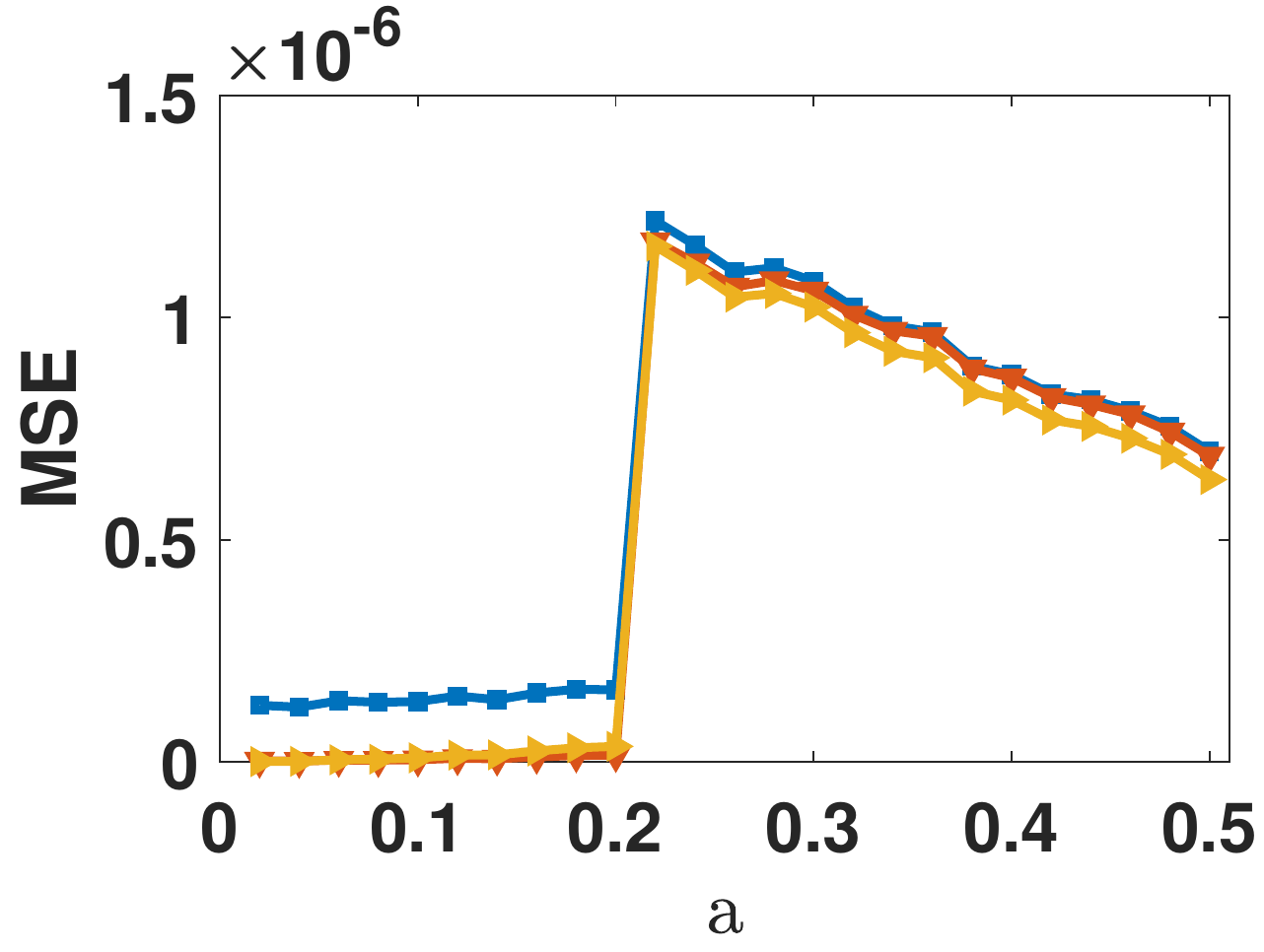}
\end{minipage}
}%
\subfigure[\textbf{Taxi}, $O=-0.6240$.]{
\begin{minipage}[t]{0.215\linewidth}
\centering
\includegraphics[width=1\textwidth]{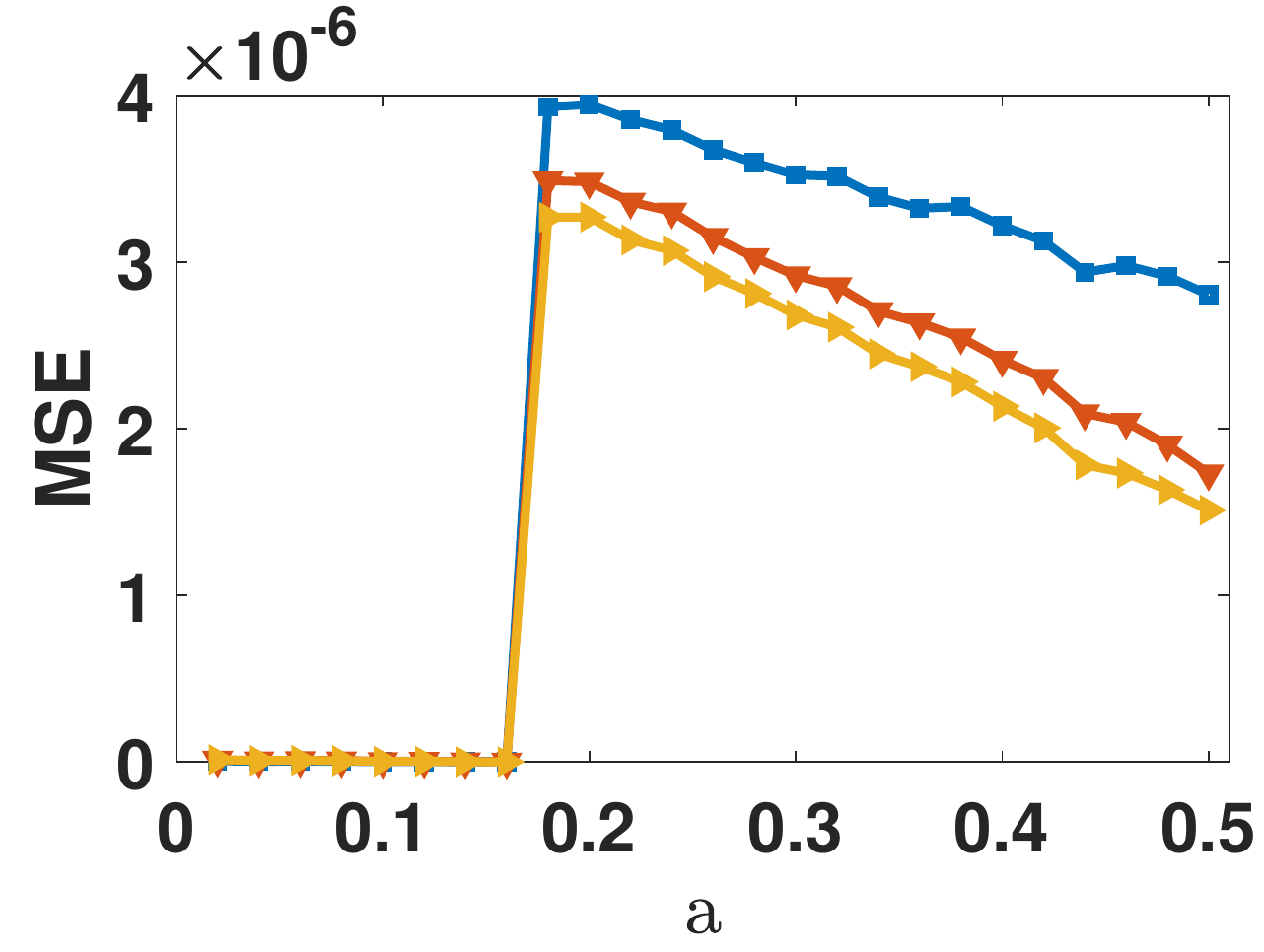}
\end{minipage}
}%
\centering
 \vspace{-0.1in}
\caption{The results of MSE w.r.t. $a$}
\label{changea}
\end{figure*}
\textbf{Robustness of Evasion.} We test the evasion effect on four datasets where $\epsilon=1/2$ and $\gamma=0.25$. A total of $m$ Byzantine users inject $am$ {\bf evasive poison values} at $-C/2$ (where $-C/2<O'$) and $(1-a)m$ {\bf true poison values} uniformly on $[C/2, C]$, where $a$ is the percentage of evasive poison values. In Fig. \ref{changea}, we plot the MSE of mean estimation, which essentially measures the utility of poison values. We observe that when $a$ is small, all DAPs still work well and ignore the evasive poison values. When $a$ is larger than a threshold (around 20\%--30\%), the MSE suddenly goes up high, because these evasive poison values finally become effective, and our schemes begin to misjudge the poisoned side. After that, the MSE starts to drop, which means such evasion will lose effectiveness as there are fewer true poison values.
\section{Related Work}
\label{Relatedwork}
Local differential privacy (LDP)~\cite{chen2016private, duchi2013local, kasiviswanathan2011can}, a variant of DP~\cite{dwork2008differential, dwork2006calibrating, mcsherry2007mechanism}, is proposed to protect the data privacy of users in an untrusted environment. RAPPOR~\cite{erlingsson2014rappor}, an extension deployed within Chrome by Google, was the first client-based practical privacy solution. After that, the LDP technique has been widely applied in the industry to protect their users' privacy, like iOS for Apple~\cite{team2017learning} Win 10 for Microsoft~\cite{ding2017collecting}, and Samsung~\cite{nguyen2016collecting}. Since the concept of LDP was proposed, it has been widely applied in multiple fields, including multi-attribute values estimation~\cite{ren2018textsf, du2021collecting} marginal release~\cite{cormode2018marginal, zhang2018calm}, time series data release~\cite{ye2021beyond}, graph data collection~\cite{sun2019analyzing, ye2020towards, 9306906}, key-value data collection~\cite{ye2019privkv, gu2020pckv,ye2021privkvm} and private learning~\cite{zheng2019bdpl, zheng2020protecting}. 

Mean and frequency estimations are commonly seen in LDP scenarios. Several mechanisms~\cite{acharya2019hadamard, bassily2017practical, bassily2015local, ding2017collecting, erlingsson2014rappor, wang2017locally, li2020estimating} are proposed for frequency estimation under LDP. Among them, the most relevant work to estimate numerical data by using EM algorithm in LDP is the SW mechanism ~\cite{li2020estimating}. However, the SW mechanism is not designed for combating Byzantine attacks, and is therefore unable to eliminate the impact of poison values. For mean estimation, Duchi et. al propose a 1-bit mechanism~\cite{duchi2018minimax} and Wang et. al~\cite{wang2019collecting} propose the Piecewise Mechanism are the state of the art methods for mean estimation.


The problems of the Byzantine attacks, that is, data poisoning attacks have recently been studied in many fields, such as crowdsourcing and crowdsensing scenarios~\cite{chang2016protecting, gadiraju2015understanding}, applications of Internet of Things~\cite{illiano2015detecting, rezvani2014secure}, electric power grids~\cite{liu2011false} and machine learning algorithms~\cite{ carlini2021poisoning, fang2020local, fang2020influence}. However, combating Byzantine attacks in LDP protocols is a relatively new topic that has few state-of-the-art papers. Literature~\cite{cheu2021manipulation} figures out that LDP is vulnerable to manipulation attacks. With a small privacy budget or a large input domain, a few poisoned values can completely ruin the real distribution. To combat this kind of attack, sampling is an easy but effective approach. Literature~\cite{cao2021data} formulates the data poisoning attack as an optimization problem and proposes three attacking patterns to maximize their attacking effectiveness, and design some countermeasures accordingly. Literature~\cite{wu2021poisoning} is the first attempt at poisoning attacks for key-value data in LDP protocols. They formulate an attack with two objectives, which are to simultaneously maximize the frequencies and mean values and to design two countermeasures against this attack. Literature~\cite{kato2021preventing} proposes a novel verifiable LDP protocol based on Multi-Party Computation (MPC) techniques. They propose a verifiable randomization mechanism in which the data collector can verify the completeness of executing an agreed randomization mechanism for every data provider. However, this method is only proposed for the categorical frequency oracles, such as kRR~\cite{kairouz2014extremal}, OUE~\cite{wang2017locally} and OLH~\cite{wang2017locally} instead of mean and distribution estimation on numerical values.
\section{Conclusion}
\label{conclusion}This paper address the general Byzantine attacks in LDP mechanisms, which advances state-of-the-art works by eliminating prior knowledge of either the attacking pattern or the poison value distribution. We present EMF, a novel algorithm to estimate key Byzantine users' features, including the proportion of Byzantine users and the poisoned side, which are then utilized to improve the mean estimation accuracy. To further enhance the utility and security of our scheme, we propose a group-wise protocol DAP and two post-processing schemes (EMF* and CEMF*) of EMF that collectively achieves the optimized mean value in terms of variance. Extensive empirical studies verify the correctness and effectiveness of the DAP protocol.

\bibliographystyle{plain}
\bibliography{references}
\clearpage

\end{document}